\documentclass[peerreview,11pt]{IEEEtran}
\IEEEoverridecommandlockouts

\usepackage[english]{babel}
\usepackage[usenames]{color}
\usepackage{amsfonts}
\usepackage{amsthm}
\usepackage{graphicx}
\usepackage[section]{placeins}
\usepackage{mathrsfs}
\usepackage{amsmath}
\usepackage{bm}
\usepackage{epstopdf}
\usepackage{url}

\DeclareGraphicsExtensions{.ps}

\theoremstyle{plain}
\newtheorem{df}{Definition}

\newtheorem{lemma}[df]{Lemma}

\newtheorem{example}[df]{Example}

\renewcommand{\thefootnote}{\fnsymbol{footnote}}

\renewcommand{\baselinestretch}{1.0}

\def\parsec{\par\noindent}
\def\med{\medskip\parsec}

\def\ulines{\underline{s}}
\def\ulineS{\underline{S}}
\def\ulinez{\underline{z}}
\def\ulineh{\underline{h}}

\def\ulineg{\underline{g}}
\def\ulineG{\underline{G}}
\def\ulinea{\underline{a}}
\def\ulineb{\underline{b}}
\def\ulinep{\underline{p}}

\def\uliner{\underline{r}}
\def\ulineR{\underline{R}}
\def\ulined{\underline{d}}

\def\ulinec{\underline{c}}

\def\ulinet{\underline{t}}

\def\ulinef{\underline{f}}

\def\ulinelambda{\underline{\lambda}}
\def\ulinemu{\underline{\mu}}

\def\ulinehatr{\underline{\hat{r}}}
\def\ulinetilder{\underline{\tilde{r}}}

\def\ulinetildeg{\underline{\tilde{g}}}
\def\ulinehath{\underline{\hat{h}}}
\def\ulinetildeh{\underline{\tilde{h}}}

\def\3To1BC{$3-$to$-1$}

\def\define{:{=}~}

\def\naturals{\mathbb{N}}
\def\reals{\mathbb{R}}
\def\integers{\mathbb{Z}}

\def\Expectation{\mathbb{E}}

\newif\ifProofForORDBC

\def\nhist{\mathcal{H}^{n}}
\def\nhistext{\mathcal{H}^{n}_{\mbox{\tiny ext}}}
\def\histof{\texttt{h}}
\def\histdist{\mathcal{F}}
\def\hist{\mathcal{H}}
\def\databasespace{\mathcal{R}^{n}}
\def\neighbor{\mathcal{N}}

\def\EhrSer{\rm Ehr}
\def\EhrFaceSer{\mathscr{E}_{\mathcal{P},f}}
\def\AnalExpDis{\mathscr{D}_{K}}
\def\AnalExpDisK2{\mathscr{D}_{2}}
\def\WOfGeomMech{\mathbb{W}_{\mathscr{G}}}
\def\binpmf{\mathscr{C}_{i}^{n}}

\newif\ifTITVersion
\newif\ifSODAVersion
\SODAVersiontrue

\usepackage{dsfont}

\newcommand{\comment}[1]{}
\begin{document}
\sloppy
\newtheorem{remark}{\it Remark}
\newtheorem{thm}{Theorem}
\newtheorem{corollary}{Corollary}
\newtheorem{definition}{Definition}
\newtheorem{prop}{Proposition}

\title{The Trade-off between Privacy and Fidelity via Ehrhart Theory}
\author{Arun Padakandla, P. R. Kumar~\IEEEmembership{Fellow,~IEEE} and Wojciech 
Szpankowski~\IEEEmembership{Fellow,~IEEE}
\thanks{This work was supported in part by
NSF Center for Science of Information (CSoI) Grant CCF-0939370, NSF Grants
CCF-1524312, ECCS-1646449, and CNS-1719384, NIH Grant 1U01CA198941-01, and 
USARO under Contract 
W911NF-15-1-0279.}
}
\maketitle
\begin{abstract}

As an increasing amount of data is gathered nowadays and stored in databases, 
the question arises of how to protect the privacy of individual records in a 
database even while providing accurate answers to queries on the 
database. Differential Privacy (DP) has gained acceptance as a framework to 
quantify vulnerability of algorithms to privacy breaches. We consider 
the problem of how to sanitize an entire database via a DP mechanism, on which 
unlimited further querying is performed. While protecting privacy, it is 
important that the sanitized database still provide accurate responses to 
queries. The central contribution of this work 
is to characterize the amount of information preserved in an optimal DP database
sanitizing mechanism (DSM). We precisely characterize the utility-privacy 
trade-off of mechanisms that sanitize databases in the asymptotic regime of 
large databases. We study this in an information-theoretic
framework by modeling a generic distribution on the data, and
a measure of fidelity between the histograms of the original
and sanitized databases. We consider the popular
$\mathbb{L}_{1}-$distortion metric, i.e., the total variation norm that leads
to the formulation as a linear program (LP).
This optimization problem is prohibitive in complexity with the number of 
constraints growing exponentially in the parameters of the problem. Leveraging 
tools from discrete geometry, analytic combinatorics, and duality theorems of 
optimization, we fully characterize the optimal solution in terms of a power 
series whose coefficients are the number of integer points on a 
multidimensional convex cross-polytope studied by Ehrhart in 1967. Employing 
Ehrhart theory, we determine a simple closed form computable 
expression for the asymptotic growth of the optimal privacy-fidelity trade-off 
to infinite precision. At the heart of the findings is a deep connection 
between the minimum expected distortion and a fundamental construct in Ehrhart 
theory - Ehrhart series of an integral convex polytope.

\end{abstract}
\begin{IEEEkeywords}
Differential Privacy, fidelity, distortion, information theory, linear 
programming 
optimization,
Ehrhart theory, discrete geometry, dual LP, analytic combinatorics.
\end{IEEEkeywords}

\renewcommand{\thefootnote}{\arabic{footnote}}
\setcounter{footnote}{0}

\section{Introduction : Motivation, Contribution and Significance}
\label{Sec:Introduction}
Nowadays, fine grained and high-dimensional data containing information about 
their preferences/characteristics is being increasingly gathered from subjects. 
The data 
is stored in modern databases (DBs) that 
permit unrestrained and continuous querying. It is then mined for social, 
scientific, commercial and economic benefits. Dependencies discovered via such 
querying, among attributes previously
not known to be related, can lead to significant scientific breakthroughs and/or
commercial benefits. Due to their value, DBs are therefore being traded among 
corporations and
governmental agencies to facilitate informed policy making. However, 
such trading of DBs containing private information, amongst untrusted 
agencies, and their unrestrained querying, results in catastrophic loss 
of subject privacy \cite{1997MMTJLME_Swe,2008IEEESSP_NarShm}.

To protect privacy, data needs to be somehow obfuscated, but the utility of the
database for statistical inference degrades with increasing obfuscation. It has 
therefore become imperative to determine \emph{what} to store in a DB so that 
it simultaneously 1) permits unrestrained
querying and 2) provides acceptably accurate responses, even while
3) providing provable guarantees against privacy breaches. What is the precise 
utility-privacy trade-off, and what should be the mechanism by which the data is
obfuscated? A precise information-theoretic study of the utility-privacy 
trade-off is the subject of this paper.

The need to quantify 
vulnerability of a DB sanitizing mechanism (DSM) to privacy violation has led 
to the notion of \emph{differential privacy} 
(DP) \cite{2006ICALP_Dwo, 2006CTOC_DwoMcsNisSmi}. DP models a DSM, and more 
generally a query-response mechanism, as a randomized algorithm and quantifies 
the vulnerability of the latter via its sensitivity to individual records. 
Let $\uliner$ denote a DB, and $\mathcal{N}$ the set of all ordered pairs 
($\uliner$, $\ulinehatr$) of DBs that differ in a single record.
Consider a probabilistic mechanism, that when asked a certain query 
about a database $\uliner$, randomly outputs a response $y$ with a probability 
$\mathbb{W}(y|\uliner)$. The random response can be regarded
as adding noise to the answer of the query, though more randomization than mere 
addition is
allowed. Such a mechanism $M$ is $\theta-$DP for $\theta \in [0,1]$, if
\begin{equation}
\label{Eqn:DPDefnInIntro}
\theta \leq  \max_{(\uliner, 
\ulinehatr) \in \mathcal{N}} \max_{y \in 
\mathcal{Y}} ~\frac{\mathbb{W}_{M}(y|\uliner)}{\mathbb{W}_{M}(y|\ulinehatr)} 
\leq \frac{1}{\theta}. \nonumber
\end{equation}
Larger values of $\theta$ correspond to less 
vulnerable mechanisms, but this increased protection is achieved at the cost of 
reduced accuracy of the query response. The key 
properties of DP - composition \cite[Section 
3.5]{Dwork:2014:AFD:2693052.2693053} and 
post-processing \cite[Proposition 2.1]{Dwork:2014:AFD:2693052.2693053} - have 
motivated its adoption as a measure of privacy. In particular, the 
``post-processing'' property states that querying a 
DB sanitized via 
$\theta-$DP DSM is, irrespective of the query and the querying mechanism, 
at least as robust as a $\theta-$DP mechanism. In other words, sanitizing a DB 
via a DP mechanism provides an impermeable firewall against privacy 
breaches.

This architecture has been referred to in the literature as
\textit{non-interactive} mechanisms. We reduce the case of persistent querying 
to the non-interactive case by considering how the
\emph{entire} database can be sanitized and exported. We address the following 
central questions
that govern the same. Firstly, how does one quantify the amount of information
preserved in a DB sanitizing mechanism (DSM)? Any such metric must be
representative of the accuracy of responses provided to canonical DB queries. A
higher accuracy of responses must be reflected by a larger amount of information
preserved.
Secondly,
among all DSMs subject to a DP constraint $\theta \in (0,1)$, henceforth
referred to as a $\theta-$DP DSM, which of them is optimal, and how
much information is preserved?

Taking a cue from rate-distortion theory, we quantify the
information preserved between the \textit{information source} (original
DB) and its \textit{representation} (sanitized DBs) via a measure of
\textit{fidelity}. Most statistical, machine learning queries aim to glean at
correlations across attributes. The quintessential object of
interest is the histogram of the DB, referred to as
\textit{type} \cite[Chap. 2]{CK-IT2011},
\cite{199810TIT_Csi, 201608TIT_BarDudSzp}. We
therefore characterize fidelity between the original and sanitized DBs via a
distortion between
their corresponding histograms. Measures of divergence between probability
distributions
such as total variation (TV), Kullbach-Leibler, Csisz\'{a}r $f-$divergences
\cite{1967SSMH_Csi, 1978TSPCIT_Csi} serve as good choices
for measure of distortion. Here we focus on the TV
distance. Simple
and yet popular, this choice provides us with an elegant case to present
fundamental
connections between DP and discrete geometry, combinatorics.

Adhering to the information-theoretic flavor, we focus on
characterizing precisely the minimum expected distortion between
histograms of the original and sanitized DBs, of an optimal $\theta-$DP DSM, in
the asymptotic regime of large DBs. Section \ref{Sec:Preliminaries}
contains a mathematical formulation of this problem. The latter reduces to a
prohibitively complex optimization problem (Remark \ref{Rem:InvolvedLPProblem})
with
an exponential number of constraints. Seeking to identify the structure of
the optimal mechanism, we consider the $\mathbb{L}_{1}$ or TV divergence
measure, in which case the objective function is linear, thereby resulting in a
linear program (LP). We are thus confronted with the task of identifying the 
limit of
solutions to a sequence of LPs, each of which is subject to exponentially
many constraints (Remark \ref{Rem:InvolvedLPProblem}). One of our main
contributions is a precise characterization of this limit, and hence the minimum
expected $\mathbb{L}_{1}-$distortion of a $\theta-$DP DSM, in the limit of large
DBs.

Our solution is built on the fundamental connections we discover between DP
and \textit{Ehrhart theory} \cite{BeckRobins-2015}. Ehrhart theory concerns
integer-point enumeration of polytopes. The counts of the number of integer
points in the $t-$th dilation of a polytope (Fig.
\ref{Fig:LatticePointEnumeration}) - the \textit{Ehrhart polynomial} of the
polytope - and the associated generating function - the \textit{Ehrhart series}
of the polytope - are fundamental constructs in Ehrhart theory. As we describe
below, they will play a central role in characterizing the limit we
seek.
\begin{figure}\centering
\begin{minipage}{.49\textwidth}
 \centering
\includegraphics[width=2.7in,height=1.3in]{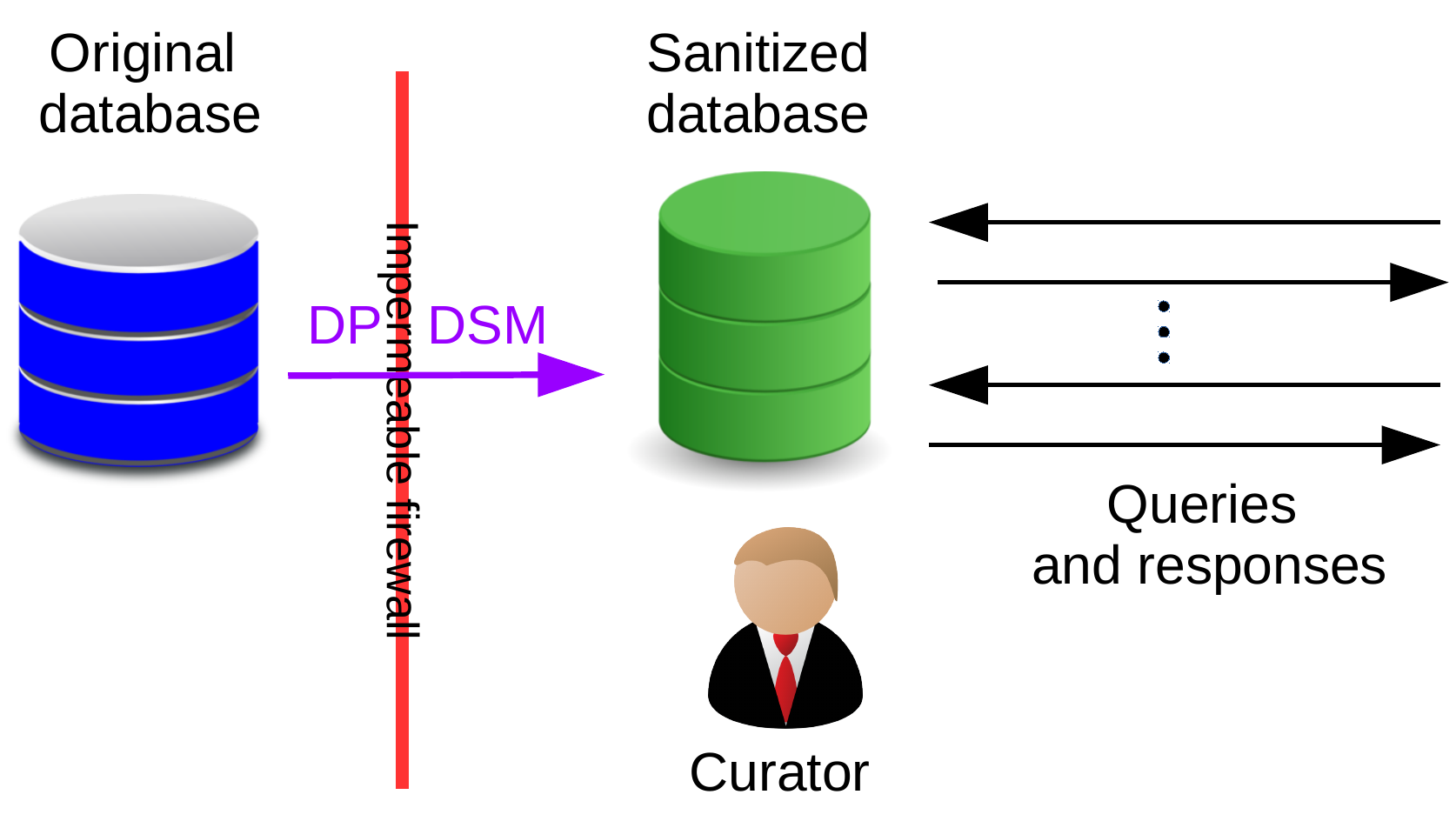}
\caption{Differentially Private Database Sanitizing Mechanism. 
The original database is sanitized and then destroyed. All subsequent querying, 
unlimited in any way, is subsequently performed only on the \textit{sanitized} 
database.}
\label{Fig:DPDSMArchitecture}
\end{minipage}~~~~\begin{minipage}{.49\textwidth}\centering
\includegraphics[width=1.8in]{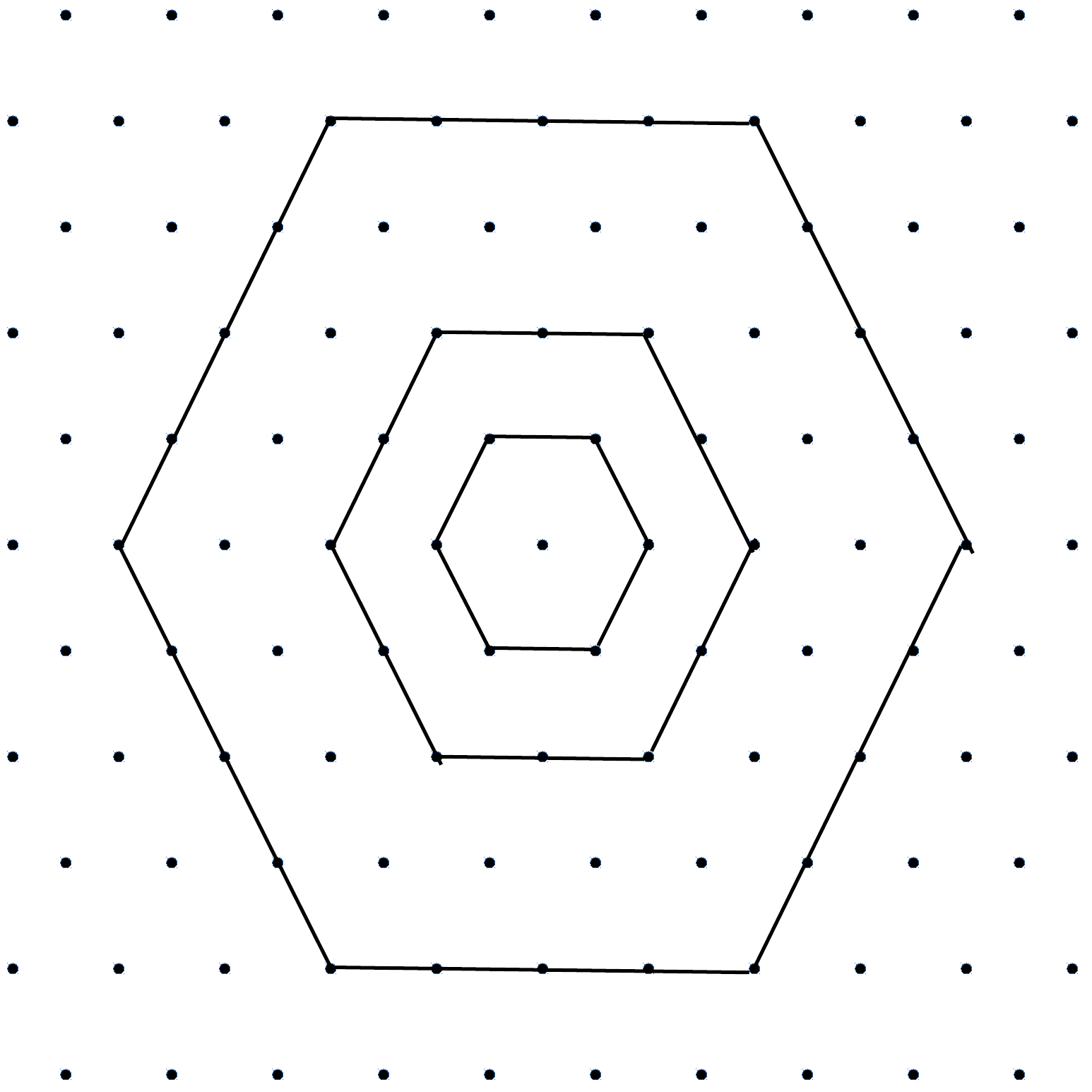}
\caption{Counts of the number of integer points in the $t-$th dilation of a 
polytope. The dots represent integer points. There are 6, 12 and 24 integer 
points in the 1st, 2nd and 4th dilation of the innermost convex polytope.}
\label{Fig:LatticePointEnumeration}
\end{minipage}
\end{figure}

Our crucial first step of visualizing the LP through a graph paves the way to
developing these connections with discrete geometry. In particular, we relate
the objective and constraints of the LP with the distance distribution of
vertices in this graph. This relationship enables us to glean the structure of
an optimal solution to our LP. Identifying symmetry properties of the
graph, we make the key observation that its distance distribution can be
obtained via the Ehrhart polynomial of a suitably defined convex polytope.
Leveraging these insights, we identify a sequence of truncated 
geometric $\theta-$DP mechanisms,
which are indeed feasible solutions to the sequence of LPs. We characterize
the limit of the corresponding sequence of expected $\mathbb{L}_{1}-$fidelities
through a simple functional of the Ehrhart series of the above mentioned convex
polytope, a significant finding.
We then employ tools from analytic combinatorics and provide a
simple computable closed form expression to the above functional,
thereby further characterizing explicitly the limit of the 
sequence of expected
$\mathbb{L}_{1}-$distortions.

The above mentioned expression is a limit of the objective values
corresponding to a sequence of feasible solutions, and hence serves
as an upper bound on the limit we seek. We leverage weak duality of LP to 
identify a lower bound. Note
that every feasible solution to the dual of the above LP evaluates to a lower
bound on the minimum expected distortion. We therefore consider the sequence of
dual LPs and identify a sequence of feasible solutions for the same. We prove
that these feasible solutions evaluate to, in the limit, the same functional
as obtained in the upper bound. This enables us to conclude that the Ehrhart 
series
of the above mentioned convex integral polytope yields the minimum expected
$\mathbb{L}_{1}-$distortion of a $\theta-$DP DSM, thereby establishing a
connection between objects of fundamental interest in the two disciplines/areas.

In addition to proving that the sequence of 
truncated geometric mechanisms is optimal in the limit, the findings 
highlight 
a useful and interesting property analogous to universal optimality 
\cite{2012SJC_GhoRouSun}. Given any distribution (pmf) on the set of records, we 
prove that this truncated geometric mechanism $\mathbb{W}^{n}(\cdot | \cdot)$ 
can be realized as a cascade of two mechanisms $\mathbb{U}^{n}(\cdot | 
\cdot), \mathbb{V}^{n}(\cdot | \cdot)$. See Figure 
\ref{Fig:CascadeMechanism}. The first mechanism $\mathbb{U}^{n}(\cdot| \cdot)$ 
is a pure $\theta-$DP geometric mechanism that is invariant with the 
distribution on the set of records. The second mechanism $\mathbb{V}^{n}(\cdot | 
\cdot)$ is a truncation that is centered at the histogram corresponding to the 
distribution. The invariance of $\mathbb{U}^{n}(\cdot | \cdot)$ lends utility to 
this cascade mechanism. Specifically, a data gatherer who is oblivious to the 
true distribution on the set of records can sanitize the original DB through 
$\mathbb{U}^{n}(\cdot | \cdot)$ and generate an intermediate DB that is 
guaranteed to protect privacy while not compromising on utility. Indeed, any 
entity or enterprise with an accurate knowledge of the underlying distribution 
can post-process the intermediate database with the corresponding mechanism 
$\mathbb{V}^{n}(\cdot | \cdot)$ to obtain a DB with least distortion. In 
essence, this property permits distributed implementation of an optimal 
mechanism. This leads us to the notion of universal optimality 
\cite{2012SJC_GhoRouSun}. Ghosh, Roughgarden and Sundararajan 
\cite{2012SJC_GhoRouSun} have studied the 
particular setting of a count query, i.e., a database whose records can take 
one among two possibilities. They prove that the truncated geometric mechanism 
is universally 
optimal for any size of the database for a fairly general class of utility 
functions. Brenner and Nissim \cite{2010FOCS_BreNis} 
prove that such universal optimal mechanisms do \textit{not} exist if the 
records can take more than two possibilities. Our findings bring to light a 
relaxed notion of universal optimality that is useful, and which circumvents 
the 
impossibility results proven in \cite{2010FOCS_BreNis}. Specifically, we seek 
optimality only for the family of multinomial distributions on the space of 
histograms. As the reader will note, this is sufficiently general. Secondly, we 
seek optimality in the limit of large databases. These two relaxations of 
universal optimality, both in the spirit of information theory, enable us prove 
positive existence results and are 
useful in the light of \cite{2010FOCS_BreNis}.

While DP \cite{Dwork:2014:AFD:2693052.2693053} has been a subject of intense 
research, the problem of identifying optimal mechanisms and characterizing the 
privacy-fidelity trade-off in the expected sense has received much less 
attention. This, as we state in Remarks \ref{Rem:InvolvedLPProblem} and 
\ref{Rem:PreciseAsymptoticsOfFidelity}, is due to the complexity of the 
resulting optimization problem. Ghosh Roughgarden and Sundararajan 
\cite{2012SJC_GhoRouSun} focus attention on a single count query and prove 
universal optimality of the geometric mechanism for a fairly general class of 
utility measures. It may be however noted that their finding only provides 
structural properties of an optimal mechanism leaving the precise 
characterization of an optimal mechanism and the maximum utility open. Our 
findings answer this question in the asymptotic limit of large databases, and 
moreover for a multi-dimensional count query. In our work, we provide a 
solution to the original optimization problem without resorting to relaxation 
or continuous extensions, in spite of its hardness. This is, in spirit similar 
to the work of Geng and Viswanath \cite{7353177, 7345591}, wherein staircase 
mechanisms \cite{7093132} are proven to be the optimal noise adding mechanisms 
for a general class of convex utility functions, albeit in the minimax setting. 
Specifically, \cite{7345591} employs functional analytic arguments to 
characterize the density function of an optimal noise adding mechanism.

Finally, we highlight certain additional aspects of our work. By considering an 
arbitrary
distribution for entries in the DB, we enable a generic information theoretic
study (Remark \ref{Rem:ProbabilityModellingOfDatabase}). Secondly, in
our general formulation, a standard geometric 
mechanism is not optimal; in fact it is non-trivial to identify an 
optimal one
(Remark \ref{Rem:NotPlainVanillaGeometricMechansism}). However, by identifying 
an optimal
sequence of mechanisms we also design an efficient shaping of the geometric
mechanism that renders it both feasible and optimal. Thirdly, we
prove this sequence of mechanisms to be asymptotically universally optimal
\cite{2012SJC_GhoRouSun}, thereby potentially supporting its adoption
(Remark \ref{Rem:AsymptoticallyUniversallyOptimal}).

 \section{Preliminaries: Notation, Problem Statement}
 \label{Sec:Preliminaries}
Notation will be introduced as and when necessary. A summary is provided in 
Table \ref{Table:Notation} in Appendix \ref{AppSec:Notation}.
\med
{\bf Problem Formulation : }Consider a DB with $n$ \textit{subjects}. Each
subject is 
identified with a \textit{record} which stores his or her preferences 
and/or characteristics. We let $\mathcal{R} = \{ a_{1},\cdots,a_{K}\}$ 
denote the set of possible records. $K$ can be arbitrary, but will remain 
fixed throughout our study. We let $\uliner \define (r_{1},\cdots,r_{n})
\in\databasespace$ denote a generic DB with $n$ records. 
 
\begin{example}
\label{Ex:DatabaseExample}
{\rm Consider the DB in Fig.~\ref{Fig:DatabaseExample} containing records of
$n=6$ subjects. Each records contains $5$
attributes - zip-code, ethnicity, 
income, health and average-monthly-expenditure. The database stores subject 
information with respect 
to $5$ attributes - zipcode, ethnicity, income, health and 
average-monthly-expenditure. Let $\mathcal{A}_{1} = \{ 47906, 47907, 77840, 
77841\}, \mathcal{A}_{2}=\{ \mbox{asian}, \mbox{caucasian}, \mbox{hispanic}\}, 
\mathcal{A}_{3} = \{ 50000, 55000, \cdots, 300000\}, 
\mathcal{A}_{4}=\{\mbox{heart-ailment}, \mbox{no-heart-ailment}\}, 
\mathcal{A}_{5}=\{500, 600,\cdots ,4000 \}$ denote the preferences 
corresponding 
to the attributes. The set of records is $\mathcal{R} = \mathcal{A}_{1}\times 
\cdots \mathcal{A}_{5}$, and $K = |\mathcal{R}| = 4\cdot 3 \cdot 51 \cdot 2 
\cdot 
36 = 44064$.}
\end{example}
\begin{figure}\centering
\includegraphics[width=3.6in]{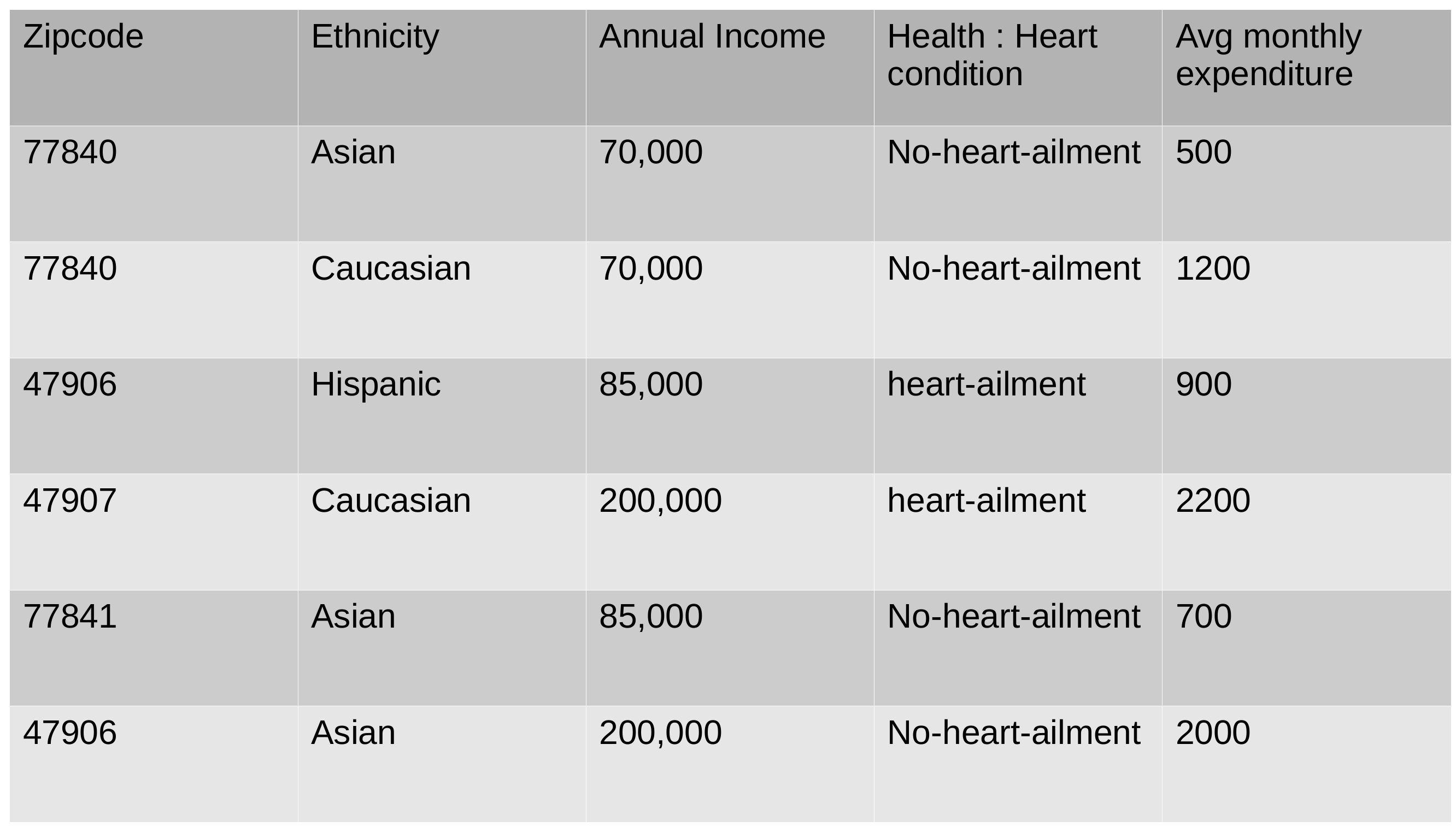}
\caption{The DB corresponding to Ex.~\ref{Ex:DatabaseExample}.}
\label{Fig:DatabaseExample}
\end{figure}

The \emph{histogram} of a DB plays a key role in our study. For a 
DB $\uliner \in \databasespace$ and a record $a_{k} \in 
\mathcal{R}$, we let $\histof(\uliner)_{k} = \sum_{i=1}^{n}
\mathds{1}_{\left\{ r_{i}=a_{k}\right\}}$ denote the number of subjects 
with record $a_{k}$, and $\histof(\uliner) \define 
(\histof(\uliner)_{1},\cdots, \histof(\uliner)_{K})$ denote 
the histogram corresponding to DB $\uliner \in \databasespace$. Let 
\begin{eqnarray}
\label{Eqn:HistogramModelling}
&\nhist \define \{ (h_{1}, \cdots,h_{K}) \in \mathbb{Z}^{K}: h_{i} 
\geq 0, \sum_{k=1}^{K}h_{k}=n \}&
\end{eqnarray}
denote the collection of histograms. When $K$ is set to a particular value, 
we let $\nhist_{K}$ denote $\nhist$.

We measure fidelity between a pair of histograms through a distortion
measure $\histdist\!:\! \nhist \times \nhist \rightarrow [0,\infty)$.
Typical distortion measures include $\mathbb{L}_{1}, \mathbb{L}_{2}-$norms,
divergence between probability distributions, such as Csisz\'{a}r
$f-$divergences \cite{1967SSMH_Csi}, Wasserstein distance etc.
For histograms $\ulines, \ulinet \in \nhist$, $\histdist(\ulines,\ulinet)$ is a
proxy for the useful information of $\ulines$ contained in $\ulinet$ and vice
versa.

In order to protect privacy, we employ a DP database {\it sanitizing mechanism}
(DSM) to output a random sanitized DB. A DP mechanism is a randomized 
algorithm and we introduce the necessary notation. A mechanism (randomized
algorithm) $M:\mathcal{A} \Rightarrow \mathcal{B}$ with set $\mathcal{A}$ of
inputs and set $\mathcal{B}$ of outputs is a map $\mathbb{W}_{M}:\mathcal{A}
\rightarrow \mathbb{P}(\mathcal{B})$ where $\mathbb{P}(\mathcal{B})$ is the set
of probability distributions on $\mathcal{B}$. When input $a \in \mathcal{A}$,
the mechanism ${M}$ produces the output $b \in \mathcal{B}$ with probability
$\mathbb{W}_{M}(b|a)$. Since $M:\mathcal{A} \Rightarrow \mathcal{B}$ is uniquely
characterized by the corresponding
collection $(\mathbb{W}_{M}(\cdot|a):a \in \mathcal{A}))$ of probability
distributions, we refer to it either as $\mathbb{W}_{M}:\mathcal{A}
\rightarrow \mathbb{P}(\mathcal{B})$ or $\mathbb{W}_{M}:\mathcal{A} \Rightarrow
\mathcal{B}$. 

A pair $\uliner,
\underline{\hat{r}} \in \databasespace$ of DBs is \textit{neighboring} if
$\uliner$ and $\underline{\hat{r}}$ differ in exactly one entry. Note that
$\uliner, \underline{\hat{r}} \in \databasespace$ are neighboring if and only if
$|\histof(\uliner)-\histof(\underline{\hat{r}})|_{1}=2$. We also say a
pair of histograms $\ulineh \in \nhist$ and $\underline{\hat{h}}\in \nhist$ is
neighboring if $|\ulineh - \underline{\hat{h}}|_{1}=2$.

\begin{definition}
 \label{Defn:DP}
 Consider the space $\databasespace$ of DBs with $n$ subjects. A DSM, $M
:\databasespace \Rightarrow \databasespace$ is $\theta-$DP ($0<\theta<1$) if for
every pair of neighboring DBs $\uliner, \underline{\hat{r}}$ and every
DB $\ulines \in \databasespace$, we have $
  \theta~\mathbb{W}_{M}(\ulines|\uliner)\leq\mathbb{W}_{M}(\ulines|\underline{\hat{r}})\leq \theta^{-1}~\mathbb{W}_{M}(\ulines|\uliner)$.
\end{definition}
We formulate the problem of characterizing the minimum 
\textit{expected} distortion of a $\theta-$DP DSM. Towards that end, 
we model a distribution on the space of DBs. For a 
record $a_{k} \in \mathcal{R}$, let $p(a_{k})>0$ denote the probability 
that a subject's record is $a_{k}$. The $n$ records that make 
up the DB are independently and identically distributed 
with pmf $\ulinep \define (p(a_{k}): a_{k} \in \mathcal{R})$. The probability 
of the gathered DB being $\uliner = (r_{1},\cdots,r_{n})$ is 
$\prod_{i=1}^{n}p(r_{i})$ where $r_{i}$ is the record of the $i$-th subject.
\begin{remark}
 \label{Rem:ProbabilityModellingOfDatabase} {\rm We do not assume any
restriction on $\ulinep$, allowing a generic information theoretic
study, as we further elaborate below by showing that the problem can be mapped 
into the class of histograms. In particular, since we do not assume $p_{k}$ 
factorizes across
attribute fields, as for example a uniform distribution would, the model
permits arbitrary correlation across attributes.}
\end{remark}

The \emph{expected distortion} of a DSM 
$(\mathbb{W}_{M}(\cdot|\uliner): 
\uliner \in \mathcal{R}^{n})$ is defined as
\begin{eqnarray}
 \label{Eqn:ExpectedFidelityOfMechanism}
  &\displaystyle D^n(\mathbb{W}_{M},\ulinep,\histdist) \define 
\Expectation_{M}\left\{ \histdist(\histof(\ulineR),
\histof(\underline{S}))\right\} \define \!\!\sum_{\uliner \in \databasespace} 
\sum_{\ulines \in \databasespace} \prod_{i=1}^{n}p(r_{i})
\mathbb{W}_{M}(\ulines |
\uliner)\histdist(\histof(\uliner),\histof(\ulines)).\nonumber\end{eqnarray}We
now provide a formulation of the problem: We seek to characterize
\begin{eqnarray}
\label{Eqn:PreliminaryProblemStatement}
\displaystyle D^{*}_{K}(\theta,\ulinep,\histdist) \define \lim_{n \rightarrow 
\infty} 
D_{*}^{n}(\theta,\ulinep,\histdist),\mbox{ where } 
D_{*}^{n}(\theta,\ulinep,\histdist)\overset{(a)}{\define} 
\min_{\substack{\mathbb{W}(\cdot|\cdot)\mbox{\scriptsize is a}\\\theta-
\mbox{\scriptsize DP DSM}}}D^n({\mathbb{W}},\ulinep,\histdist).&
\end{eqnarray}
$D_{*}^{n}(\theta,\ulinep,\histdist)$ is the minimum expected distortion 
corresponding 
to a DB with $n$ records. Characterizing $D^{*}_{K}(\theta,\ulinep,\histdist)$ 
precisely, as well as a sequence of optimal mechanisms is the
main goal of the study.

\section{Main Results : Precise characterization of 
$D^{*}_{K}(\theta,\ulinep,|\cdot|_{1})$ and Essential Universal Optimality}
\label{Sec:MainResult-PrecCharViaEhrhart}
First, we provide a simpler equivalent formulation of
problem (\ref{Eqn:PreliminaryProblemStatement}) with an exponentially smaller
number of decision variables. As we will note, even this simplified
formulation is quite involved.
\med
{\bf Equivalent formulation of $D_{*}^{n}(\theta,\ulinep,\histdist)$ via 
sufficiency of histogram
sanitization: }Viewing the DB through its histogram enables us to simplify
(\ref{Eqn:PreliminaryProblemStatement})(a). We make two observations. (i) The
distortion between the original and sanitized DBs is a function only of
their histograms, and (ii) the DP constraints are related only through the
histograms of the DBs. These observations enable us to restrict attention
to mechanisms that identically randomize DBs with the same histogram. For
such a mechanism $M$, we have $(\mathbb{W}_{M}(\ulines|\uliner): \ulines \in
\databasespace) = (\mathbb{W}_{M}(\ulines|\ulinetilder): \ulines \in
\databasespace)$ whenever $\histof(\uliner) = \histof(\ulinetilder)$. In
Appendix \ref{AppSec:RestrictingSearchForMechanisms}, we prove that this
restriction does \textit{not} entail 
any loss in optimality. The first observation enables us to go further. It lets 
us
conclude that the expected distortion of a mechanism does not depend on how it
distributes the probability among DBs with the same histogram. Formally,
the expected distortions of two DSMs 
$M,\tilde{M}$ are identical if $\sum_{\ulines \in
\mathcal{R}^{n}:\histof(\ulines)=\ulineh}\mathbb{W}_{M}(\ulines|\uliner) =
\sum_{\ulines \in
\mathcal{R}^{n}:\histof(\ulines)=\ulineh}\mathbb{W}_{\tilde{M}}
(\ulines|\uliner)\mbox{ for all }\ulineh \in \nhist\mbox{ and for all }\uliner
\in \databasespace.$ These enable us to shift our viewpoint from DB 
sanitization to
\textit{histogram sanitization}. We define a $\theta-$DP histogram sanitizing
mechanism (HSM) as follows:
\begin{definition}
 \label{Defn:ThetaDPHSM}
A pair $\ulineh, \ulinehath \in \nhist$ of histograms is neighboring if
$|\ulineh-\ulinehath|_{1}=2$. A histogram sanitizing mechanism (HSM) $M : \nhist
\Rightarrow \nhist$ is $\theta-$DP ($0<\theta<1$) if for every pair $\ulineh,
\underline{\hat{h}} \in \nhist$ of neighboring histograms and every histogram
$\ulineg \in \nhist$, we have $
  \theta~\mathbb{W}_{M}(\ulineg|\ulineh)\leq\mathbb{W}_{M}(\ulineg|\underline{\hat{h}})
\leq \theta^{-1}~\mathbb{W}_{M}(\ulineg|\ulineh)$.
\end{definition}
We now describe our problem (\ref{Eqn:PreliminaryProblemStatement}) from the
histogram sanitization viewpoint. A random DB $\ulineR \in \databasespace$ is
chosen with distribution as modeled earlier. Its histogram $\histof(\ulineR)$ is
input to a HSM $M: \nhist \Rightarrow \nhist$. Let $\ulineG \in \nhist$ denote
the random output histogram. Any DB $\ulineS \in \databasespace$, whose
histogram $\histof(\ulineS) = \ulineG$ can be considered as the sanitized
DB. Our goal is to find a $\theta-$DP HSM $M$ that minimizes
\begin{eqnarray}\Expectation_{M}\{ 
\histdist(\histof(\ulineR),\histof(\ulineS))\} =
\Expectation_{M}\{ \histdist(\histof(\ulineR),\ulineG)\} = \sum_{\ulineh\in
\nhist}\sum_{\ulineg \in \nhist}P(\histof(\ulineR) =
\ulineh)\mathbb{W}_{M}(\ulineg|\ulineh)\histdist(\ulineg,\ulineh).\nonumber
\end{eqnarray}
We note that the
distribution $P(\histof(\ulineR) = \ulineh)$ of the random histogram 
is given by $P(\ulineR = \uliner) = 
\prod_{i=1}^{n}p(r_{i})=\prod_{k=1}^{K}p(a_{k})^{
\histof(\uliner)_{k} }$. Henceforth, we let $p_{k}
 \define p(a_{k})$ and $\ulinep^{\ulineh} \define \prod_{k=1}^{K}p_{k}^{h_{k}}$. 
With these, we have $P(\ulineR = 
\uliner)=\ulinep^{\underline{\histof}(\uliner)}$. This leads to
\begin{eqnarray}
 \label{Eqn:ProbOfRandomDatabaseAndHistograms}
 P(\histof(\ulineR) = \ulineh)  =\sum_{\substack{\uliner \in \databasespace : \histof(\uliner) =\ulineh}}P(\ulineR=\uliner)=  \sum_{\substack{\uliner \in \databasespace : \histof(\uliner) =\ulineh}} \ulinep^{\underline{\histof}(\uliner)}
 ={n \choose \ulineh}\ulinep^{\ulineh},
\end{eqnarray}
where (\ref{Eqn:ProbOfRandomDatabaseAndHistograms}) follows from the 
fact that the number of DBs whose histogram is 
$\ulineh \in \nhist$ is the multinomial coefficient 
${n \choose \ulineh} \define {n \choose h_{1}\cdots h_{K}}$. We 
note that the multinomial distribution 
(\ref{Eqn:ProbOfRandomDatabaseAndHistograms}) with a generic 
distribution $\ulinep$ on the set of records is indeed the 
most generic distribution on the space of histograms. Throughout, we make no 
assumption on $\ulinep$, resulting in a fairly generic study.

Equation (\ref{Eqn:ProbOfRandomDatabaseAndHistograms}) lets us explicitly state
our equivalent simplified problem as follows. Given a privacy budget $\theta >
0$, our goal is to characterize $D^{*}_{K}(\theta,\ulinep,\histdist) \define 
\lim_{n \rightarrow
\infty} 
D_{*}^{n}(\theta,\ulinep,\histdist)$, where
\begin{eqnarray}
 \label{Eqn:GeneralDPProblemPosedAsALP}
 \begin{array}{lcccl}
  D_{*}^{n}(\theta,\ulinep,\histdist) & \define &\displaystyle
\min_{\substack{\mathbb{W}(\cdot|\cdot)}}
&\!\!\!\!\!\!\!\!\!\!\!\!\!\!\!
D^{n}(\mathbb{W},\ulinep,\histdist), 
&\!\!\!\!\!\!\!\!\!\!\!\!\!\!\!\!\mbox{ with
}D^{n}(\mathbb{W},\ulinep,\histdist)\define \displaystyle\sum_{\ulineh 
\in
\nhist}\sum_{\ulineg \in \nhist} {n\choose \ulineh}\ulinep^{\ulineh}
\mathbb{W}(\ulineg | \ulineh)\histdist(\ulineh,\ulineg), \\
  &
  &\!\!\!\!\!\!\!\!\!\!\!\!\!\!\!\!\!\!\!\!\!\!\!\!\!\!\!\!\!\!\!\!\!\!\!\!\!\!\!\!\!\!\!\!\!\!\!\!\!\!\!\!\!\!\!\!\!\mbox{subject to}
  &\!\!\!\!\!\!\!\!\!\!\!\!\!\!\!\!\!\!\!\!\!\!\!\!\!\!\!\!\!\!\!\!\!\!\!\displaystyle\mathbb{W}(\ulineg|\ulineh) \geq 0 
  & \mbox{for every pair }(\ulineg,\ulineh) \in \nhist \times \nhist,\\
  &
  & 
  &\!\!\!\!\!\!\!\!\!\!\!\!\!\!\!\!\!\!\!\!\!\!\!\!\!\!\!\!\!\!\!\!\!\!\!\displaystyle\sum_{\ulineg \in \nhist}\mathbb{W}(\ulineg|\ulineh) \overset{(a)}{=} 1 
  &\mbox{for every } \ulineh \in \nhist,\\
  &
  & 
  &\!\!\!\!\!\!\!\!\!\!\!\!\!\!\!\!\!\displaystyle\mathbb{W}(\ulineg|\ulineh) - \theta~\mathbb{W}(\ulineg|\underline{\hat{h}}) \overset{(b)}{\geq} 0 
  & \mbox{for every pair of histograms }(\underline{h},\underline{\hat{h}}) \in \nhist \times \nhist\\
  &
  & 
  & \!\!\!\!\!\!\!\!\!
  & \mbox{for which }|\underline{h}-\underline{\hat{h}}|_{1} = 2 \mbox{ and every }\ulineg \in \nhist.
 \end{array}
\end{eqnarray}
In going from (\ref{Eqn:PreliminaryProblemStatement}) to
(\ref{Eqn:GeneralDPProblemPosedAsALP}), we have replaced the collection
$(\mathbb{W}(\ulines|\uliner):\histof(\ulines)=\ulineh)$ by a single decision
variable $\mathbb{W}(\ulineh|\histof(\uliner))$ and set
$\mathbb{W}(\cdot|\uliner)=\mathbb{W}(\cdot|\ulinetilder)$ whenever
$\histof(\uliner)=\histof(\ulinetilder)$.
Constraints (\ref{Eqn:GeneralDPProblemPosedAsALP}) and
(\ref{Eqn:PreliminaryProblemStatement}) are specified by $|\nhist|^{2} = {n+K-1
\choose K-1}^{2} \sim (n+1)^{2K}$ and $K^{2n}$ decision variables, respectively.
With $K$ fixed, the former is exponentially smaller. This
simplification is not a result of any assumption. 
$D_{*}^{n}(\theta,\ulinep,\histdist)$ defined in
(\ref{Eqn:GeneralDPProblemPosedAsALP}) and
(\ref{Eqn:PreliminaryProblemStatement})(a) are proven to be equal in Appendix
\ref{AppSec:RestrictingSearchForMechanisms}.

\begin{remark}
 \label{Rem:InvolvedLPProblem}{\rm
The optimization problem (\ref{Eqn:GeneralDPProblemPosedAsALP}) has
$(n+1)^{2K}$ decision
variables. For every choice $(\ulineh, \underline{\hat{h}})$ of neighboring
histograms and every $\ulineg \in \nhist$, the LP imposes two types 
of constraints. There
are $\mathcal{O}(k^{2}|\nhist|^{2}) = \mathcal{O}(k^{2}(n+1)^{2(k-1)})$
constraints\footnote{For
every $\ulineh \in \nhist$ except those for which one or more of the 
coordinates are
$0$, we have $|\{\underline{\hat{t}} \in \nhist: |\ulineh -
\underline{\hat{t}}|_{1} =2 \}| = k(k-1)$. Also, $|\nhist| = {n+k-1 \choose k-1}
\sim (n+1)^{k-1}$ \cite[Lemma II.1]{199810TIT_Csi}, \cite[Chap 2,
Lemma1]{CK-IT2011}.} of the form (\ref{Eqn:GeneralDPProblemPosedAsALP})(b). For
any practical values of $K$ and $n$, it is intractable to obtain a solution via
computation. In fact, we are unaware of a solution of this LP even for the case
$K=2$. While \cite{2012SJC_GhoRouSun} proves the optimal mechanism can be
achieved by a post-processing remapping of the geometric mechanism, for any user
preference an optimal mechanism and the 
corresponding utility remain unknown.}
\end{remark}

Notwithstanding this difficulty, one can obtain a precise
characterization of $D_{K}^{*}(\theta,\ulinep,\histdist)$ by leveraging rich 
tools from discrete
geometry and LP theory.
\med
{\bf Statement of the Main Result :} We restate our problem in the context of 
the $\mathbb{L}_{1}-$distance measure. We aim to characterize 
$D^{*}_{K}(\theta,\ulinep,|\cdot|_{1}) \define \lim_{n \rightarrow \infty} 
D_{*}^{n}(\theta,\ulinep,|\cdot|_{1})$, where
\begin{eqnarray}
 \label{Eqn:DPProblemPosedAsALP}
\begin{array}{c}
  D_{*}^{n}(\theta,\ulinep,|\cdot|_{1}) \define \!\!\min 
D^{n}(\mathbb{W},\ulinep,|\cdot|_{1}) \mbox{ subject to the constraints in 
(\ref{Eqn:GeneralDPProblemPosedAsALP}), where }\\
D^{n}(\mathbb{W},\ulinep,|\cdot|_{1}) \define\!\!\! \displaystyle\sum_{\ulineh 
\in \nhist}
\sum_{\ulineg \in \nhist} \!\!\!{n\choose \ulineh}\ulinep^{\ulineh} 
\mathbb{W}(\ulineg | \ulineh)|\ulineh-\ulineg|_{1}.
 \end{array}
\end{eqnarray}
Since we restrict attention to $|\cdot|_{1}$, we let 
$D_{*}^{n}(\theta,\ulinep)$ and 
$D^{*}_{K}(\theta,\ulinep)$ denote 
$D_{*}^{n}(\theta,\ulinep,|\cdot|_{1})$ and 
$D^{*}_{K}(\theta,\ulinep,|\cdot|_{1})$ in the sequel. Theorem 
\ref{Thm:MainThmCharOfDKTheta} is our main result and provides a simple 
computable closed form expression for $D^{*}_{K}(\theta,\ulinep)$. In 
particular, we 
provide three characterizations of $D_{K}^{*}(\theta,\ulinep)$. The first one 
expresses 
$D_{K}^{*}(\theta,\ulinep)$ in terms of the Ehrhart series of a suitably defined 
convex 
polytope, thereby establishing connection between DP and Ehrhart theory. The 
second employs simple combinatorial arguments to characterize the resulting 
power series explicitly. The third exploits analytic combinatorial techniques 
to express this power series in terms of a \textit{hyper-geometric series}. The 
latter encapsulates the entire information from a power series and provides a 
computable expression. The result also shows that the limiting minimum 
distortion is not dependent on $\ulinep$.
\begin{thm}
\label{Thm:MainThmCharOfDKTheta} (a) The minimum expected 
$\mathbb{L}_{1}-$distortion of a $\theta-$DP HSM is given by \begin{eqnarray}
D_{K}^{*}(\theta,\ulinep) = \frac{2\theta}{\EhrSer_{\mathcal{P}}(\theta)}
\frac{d\EhrSer_{\mathcal{P}}(\theta)}{d\theta}-\frac{2\theta}{1-\theta}, \mbox{ 
where }\EhrSer_{\mathcal{P}}(z) \define 1+ \sum_{d=1}^{\infty} 
L_{\mathcal{P}}(d)z^{d} \label{Eqn:FirstCharacterization}
\end{eqnarray}
is the \textit{Ehrhart 
series} of the cross-polytope whose $d-th$ dilation is given by
\begin{eqnarray}
 \label{Eqn:Cross-Polytope-Early}
 &\mathcal{P}_{d} = \{ (x_{1},\cdots,x_{K}) \in \reals^{K}: \sum_{k=1}^{K}x_{k} 
= 0, \sum_{k=1}^{K}|x_{k}| \leq 2d \},&
\end{eqnarray}
and $L_{\mathcal{P}}(d)$ is the number of points in $\mathcal{P}_{d}$ with 
integer co-ordinates. $D_{K}^{*}(\theta,\ulinep)$ does not depend on 
$\ulinep$ and hence does not depend on the multinomial distribution. (b) We have
\begin{eqnarray}
 \label{Eqn:IntroOurEhrhartSeries}
 \EhrSer_{\mathcal{P}}(\theta) = \frac{1}{1-\theta} + 
\sum_{d=1}^{\infty}\left\{ \sum_{r=1}^{K-1}{K \choose r} { d+r-1 \choose 
r-1}{d-1 \choose K-r-1}\right\}\frac{\theta^{d}}{1-\theta}, \mbox{ and}
\end{eqnarray}
(c) \begin{eqnarray}
\label{Eqn:DK*OfTheta}
D^{*}_{K}(\theta,\ulinep) = 2\theta \left\{  \frac{K-1}{1-\theta}+ 
\frac{S_{K-1}'(\theta)}{S_{K-1}(\theta)}  \right\},
\mbox{ where }S_{K-1}(\theta) = \displaystyle\displaystyle\sum_{j=0}^{K-1} 
\theta^{j}\left[ {K-1 \choose j}\right]^{2}
\end{eqnarray}
with $S_{K-1}'(\theta) \define \frac{d}{d\theta}S_{K-1}(\theta)$. An optimal
HSM is obtained as a truncation of a geometric mechanism
$\mathbb{W}^{*}(\ulineg|\ulineh)=(1-\theta)^{-1}\EhrSer_{\mathcal{P}}(\theta)^{
-1 } \theta^{\frac{ |\ulineg-\ulineh|_ { 1 } } { 2 } } $, where
$\EhrSer_{\mathcal{P}}(\theta)$ is defined in 
(\ref{Eqn:IntroOurEhrhartSeries}).
\end{thm}

Below, we express $D_{K}^{*}(\theta,\ulinep)$ in terms of another important 
construct in analysis - the \textit{Legendre polynomial}. We note that 
$S_{K-1}(\theta) = (1-\theta)^{K-1}L_{K-1}(\frac{1+\theta
}{1-\theta})$ \cite[Pg. 86, Prob. 85]{PolSzeII-1976}, where $L_{n}(x)\define \frac{1}{2^{n}n!}\frac{d^{n}}{dx^{n}}(x^{2}-1)^{n}$ is the Legendre polynomial of degree $n$ defined
 in \cite[Pg. 147, Prob. 219]{PolSzeI-1976}. This leads to the following important
characterization.

\begin{corollary}
 \label{Cor:ExtractingMainTerms}
The minimum expected $\mathbb{L}_{1}-$distortion of a $\theta-$DP HSM is given 
by
\begin{eqnarray}
 \label{Eqn:SimplifiedExpression}
 D_{K}^{*}(\theta,\ulinep) = K \left\{ \frac{1+\theta}{1-\theta}+
\frac{L_{K}(y)}{L_{K-1}(y)} \right\}, \mbox{ where }y = \frac{1+\theta}{1-\theta}.
\end{eqnarray}
In particular for $K=2$, the limit $D_{2}^{*}(\theta,\ulinep) = \lim_{n 
\rightarrow \infty} D^{n}_{*} 
(\theta,\ulinep) = \frac{4\theta}{1-\theta^{2}}$.
\end{corollary}
The proof is based on the identity $S_{K-1}(\theta)=(1-\theta)^{K-1}L_{K-1}(y)$ 
and the recurrence relation $(1-y^{2})L_{K-1}'(y) = KyL_{K-1}(y) - K L_{K}(y)$. 
We provide the details in Appendix \ref{AppSec:ProofOfCorollary}.

\begin{remark}
 \label{Rem:PreciseAsymptoticsOfFidelity}
{\rm We emphasize that (\ref{Eqn:DK*OfTheta}) and
(\ref{Eqn:SimplifiedExpression}) provide an \textit{exact computable
closed form expression} for $D_{K}^{*}(\theta,\ulinep)$. Owing to the complexity 
of the
resulting optimization problem, study of the privacy-distortion trade-off for 
the
\textit{expected} distortion, which is the common object of interest in 
information theory, is very limited. While a lot more is known in the minimax
setting, most of these results are only up to an order. The
reader will note that the tools we employ in proving Thm.
\ref{Thm:MainThmCharOfDKTheta} are also
applicable for the minimax setting. A similar analysis can throw more light on
the latter setting. In the interest of brevity, we reserve this for future
work.
}
\end{remark}
\begin{remark}
\label{Rem:AboutThisProblem}
{\rm One may recover problem formulations 
studied in \cite{2012SJC_GhoRouSun, 201609TIT_WanYinZha}, among others, by an 
appropriate choice of the distortion measure 
$\mathcal{F}(\cdot,\cdot)$ in (\ref{Eqn:GeneralDPProblemPosedAsALP}). In 
particular Ghosh, Roughgarden and Sundararajan \cite{2012SJC_GhoRouSun} study 
the $K=2$ case for a fairly generic distortion measure, and prove structural 
properties of an optimal mechanism. While these hold for each $n$, 
they do not pin down an optimal mechanism, leaving 
$D_{2}^{*}(\theta,\ulinep)$ unknown. On the one hand, \cite{2010STOC_HarTal} 
studies a 
min-max problem setting. Secondly, their continuous extension results in a 
larger constraint set, lending the lower bounds developed therein invalid for 
the original discrete problem setting.}\end{remark}

A striking aspect of (\ref{Eqn:FirstCharacterization}) is the invariance of 
$D_{K}^{*}(\theta,\ulinep)$ with $\ulinep$ as noted above. Why is this true? 
For 
large $n$, ${n 
\choose \ulineh}\ulinep^{\ulineh}$ approximates a pmf that is `relatively flat'
\cite{201107TIT_SzpVer} on the set of histograms 
within an $\mathbb{L}_{1}-$ball of radius $\mathcal{O}(\sqrt{n})$ centered at 
$(np_{1},\cdots,np_{K})$. This radius being sub-linear, for any $\ulinep$ with 
positive entries, the $\mathbb{L}_{1}-$ball that contains most of the mass is 
eventually supported on the set of histograms. Since we are
concerned only in the eventual limit, the effect of $\ulinep$ is only a shift of 
the center of this $\mathbb{L}_{1}-$ball containing a `relatively flat' pmf. 
This 
leads us to the following question. Can we design a sequence $\mathbb{W}^{n} : 
\nhist \Rightarrow \nhist : n \in \naturals$ of mechanisms that is in the limit 
optimal, where each $\mathbb{W}^{n}$ can be realized as a cascade of 
$\mathbb{U}^{n}: \nhist \Rightarrow \mathcal{Y}$ and $\mathbb{V}^{n} : 
\mathcal{Y} \Rightarrow \nhist$, where $\mathbb{U}^{n}$ is $\theta-$DP and is 
invariant with $\ulinep$? As the informed reader will recognize, this is 
related to the notion of universal optimality \cite{2012SJC_GhoRouSun}. We 
define the related notion of essential universal optimality.

\begin{definition}
\label{Defn:EssUnivOpt}
{\rm A sequence $\mathbb{W}^{n}: \nhist \Rightarrow \nhist: n \in \naturals$ of 
$\theta-$DP HSMs are \textit{essentially universally optimal} (Ess-Univ-Opt) if 
for 
each $n \in \naturals$, $\mathbb{W}^{n}$ can be realized as a cascade 
$\mathbb{U}^{n}:\nhist \rightarrow \nhistext$, $\mathbb{V}^{n}: \nhistext 
\rightarrow \nhist$, i.e. (see Figure \ref{Fig:CascadeMechanism}), 
$\mathbb{W}^{n}(\ulineg|\ulineh) = \sum_{\ulineb \in 
\nhistext} \mathbb{U}^{n}(\ulineb | \ulineh)\mathbb{V}^{n}(\ulineg|\ulineb)$ 
for every $\ulineg,\ulineh \in \nhist$, where $\nhistext$ is any 
(not necessarily finite) set, such that (i) $\lim_{n \rightarrow \infty} 
D^{n}(\mathbb{W}^{n},\ulinep) = D_{K}^{*}(\theta,\ulinep)$ for every 
pmf $\ulinep$ on a set of $K$ elements, and (ii) $\mathbb{U}^{n}:\nhist 
\rightarrow \nhistext$ is $\theta-$DP and invariant with $\ulinep$.}
\end{definition}

\begin{figure}\centering
\begin{minipage}{.43\textwidth}
 \centering
\includegraphics[width=3in]{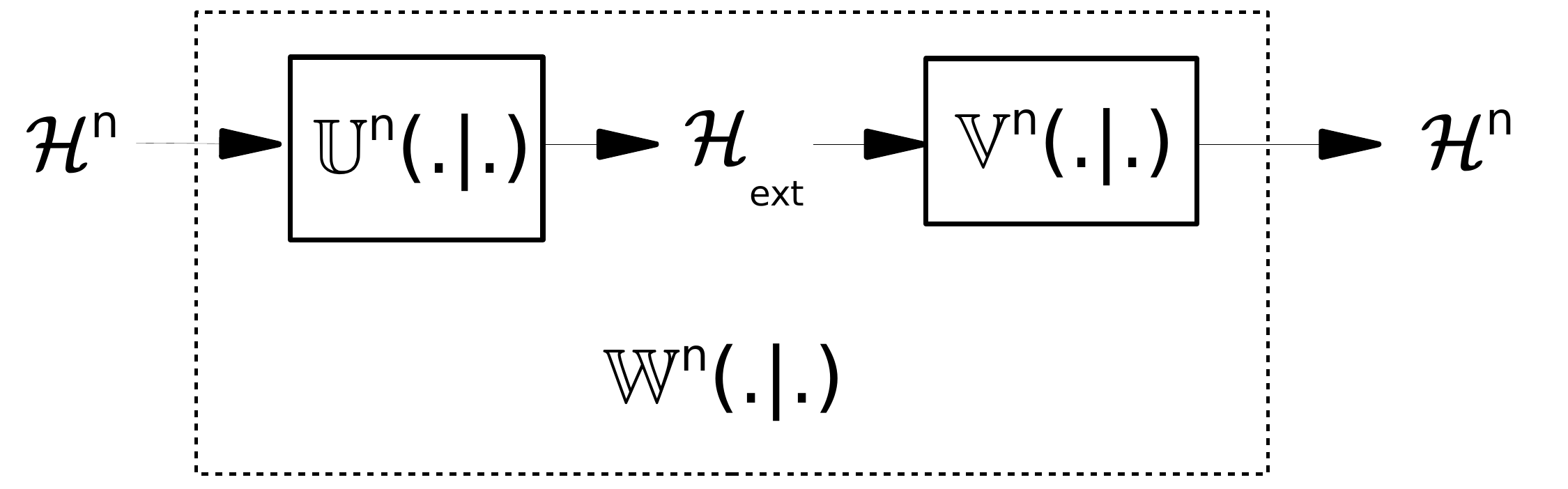}
\caption{$\mathbb{W}^{n}(\cdot | \cdot)$ realized as a cascade mechanism.}
\label{Fig:CascadeMechanism}
\end{minipage}\begin{minipage}{.57\textwidth}\centering
\includegraphics[width=3.5in]{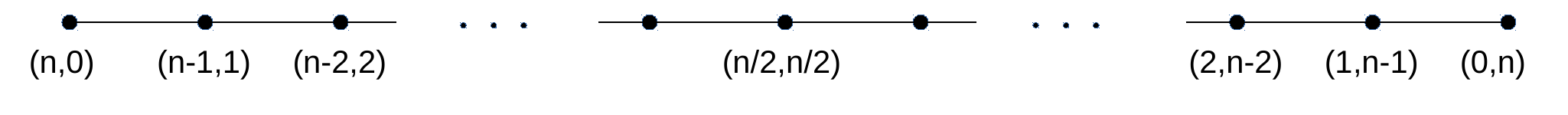}
\caption{Privacy-constraint graph for $K=2$ and general $n$. The vertices are
labeled by the corresponding histogram. Two vertices are connected by an edge if
their corresponding histograms are at an $\mathbb{L}_{1}-$distance $2$.}
\label{Fig:PCGraphFork=2}
\end{minipage}
\end{figure}

\begin{remark}
\label{Rem:AsymUnivOpt}
{\rm 
Ess-Univ-Opt is a \textit{relaxed/weaker} form 
of universal optimality \cite{2012SJC_GhoRouSun} in two respects. Firstly, we 
\textit{restrict} the class of pmfs on histograms to 
\textit{multinomial} pmfs. Indeed, our definition of 
$D^{n}_{*}(\theta,\ulinep)$ in (\ref{Eqn:DPProblemPosedAsALP}) is wrt ${n 
\choose \ulineh}\ulinep^{\ulineh}$. Secondly, we only ask for 
\textit{asymptotic} optimality 
of the sequence of mechanisms. This relaxed notion is of interest for
the following reasons. Firstly, we operate with large databases. For 
sufficiently large $n$ the distortion of an Ess-Univ-Opt 
sequence of mechanisms might be sufficiently close to the true optimum for that 
$n$. Secondly, as the reader will note, it suffices to consider multinomial 
pmfs on $\nhist$. In the light of non-existence of `strict' universally optimal 
mechanisms \cite{2010FOCS_BreNis}, it is worth pursuing this relaxed notion.
}\end{remark}
As mentioned in \cite{2012SJC_GhoRouSun}, the existence of Ess-Univ-Opt is 
noteworthy. The proof of 
our main result will bring to light a sequence of Ess-Univ-Opt mechanisms.
\begin{thm}
 \label{Thm:AsymUnivOtimalExists}
 Ess-Univ-Opt mechanisms for histogram sanitization wrt 
$\mathbb{L}_{1}-$distortion exist.
\end{thm}
The proof of Thm.~\ref{Thm:AsymUnivOtimalExists} follows from the proof of 
Thm.~\ref{Thm:MainThmCharOfDKTheta} wherein a sequence of truncated geometric 
mechanisms are proven to be Ess-Univ-Opt. The following section details 
the proof 
of Theorem \ref{Thm:MainThmCharOfDKTheta}.

\section{Analysis and Proofs}
\label{Sec:ToolsAndSketchUpperBound}
The proof of Theorem 
\ref{Thm:MainThmCharOfDKTheta} involves two parts - establishing the upper 
bound and the lower bound. The lower bound is via the weak duality theorem and 
is 
detailed in Section \ref{SubSec:LowerBound}. The 
upper bound leverages tools from Ehrhart theory and is provided in Section 
\ref{SubSec:UpperBound}. Before we provide a proof of the upper bound, we 
introduce the 
necessary constructs from Ehrhart theory and describe how and why they are 
related to $D_{K}^{*}(\theta,\ulinep)$ and the LP 
(\ref{Eqn:DPProblemPosedAsALP}) studied here. The following description
serves as a road map of the proof.

$D_{K}^{*}(\theta,\ulinep)$ is the limit of solutions to a sequence of LPs
(\ref{Eqn:DPProblemPosedAsALP}). These LPs are involved. We begin with the
privacy-constraint (PC) graph \cite{2010FOCS_BreNis} which greatly aids in
visualization and naturally leads us into Ehrhart theory. Consider a
graph $G=(V,E)$ with vertex set $V =\nhist$ and an edge set $E = \left\{
(\ulineh, \underline{\hat{h}}) \in \nhist \times \nhist:
|\ulineh-\underline{\hat{h}}|_{1} = 2\right\}$. Figures \ref{Fig:PCGraphFork=2},
\ref{Fig:PCGraphFork=3n=5} provide the PC graph for $(K=2,n)$, $(K=3,n=5)$
respectively. For every vertex $\ulineh \in V$, visualize the
sub-collection $(\mathbb{W}(\ulineg|\ulineh): \ulineg \in \nhist)$ of decision
variables as a function of $V$, i.e., as values lying on $V$, corresponding to
$\ulineh \in V$ (See Fig. \ref{Fig:PCGraphFork=3n=3WithMechanisms}). The values
$(\mathbb{W}(\ulineg|\ulineh): \ulineg \in \nhist)$ and
$(\mathbb{W}(\ulineg|\ulinehath): \ulineg \in \nhist)$ corresponding 
to two neighboring vertices $\ulineh, \ulinehath$ have to be within $\theta$ and 
$\frac{1}{\theta}$ of each other everywhere, i.e., at every $\ulineg$ (see Fig. 
\ref{Fig:PCGraphFork=3n=3WithMechanisms}). In addition, the values corresponding 
to any node must be non-negative and sum to $1$. The PC graph also provides a 
visualization of the objective function. $|\ulineg - \ulineh|_{1}$ is exactly 
twice $d_{G}(\ulineg,\ulineh)$ (proof in Lemma \ref{Lem:PropertiesPCGraph}(ii), 
Appendix \ref{AppSec:L1DistanceAndGraphDistance}). Two useful consequences 
follow. Firstly, the values corresponding to a node, say $\ulineh$, that are 
equidistant from $\ulineh$, are multiplied by identical coefficients in the 
objective function. Formally, ${n \choose 
\ulineh}|\underline{\tilde{g}}-\ulineh|_{1} = {n \choose 
\ulineh}|\underline{{g}}-\ulineh|_{1}$ iff 
$d_{G}(\underline{\tilde{g}},\ulineh)=d_{G}(\ulineg,\ulineh)$. Here and 
henceforth, $d_{G}(v_{1},v_{2})$ denotes the length of a shortest path from 
$v_{1} \in V$ to $v_{2} \in V$ in graph $G=(V,E)$. Secondly, 
coefficients associated with the values increase with their distance from 
$\ulineh$. Formally, if $d_{G}
(\underline{\tilde{g}},\underline{h}) > d_{G}(\underline{g},\underline{h})$, 
then ${n \choose \ulineh}|\underline{\tilde{g}}-\ulineh|_{1} > {n \choose 
\ulineh}|\underline{{g}}-\ulineh|_{1}$. These observations let us restate our 
objective function (\ref{Eqn:DPProblemPosedAsALP}) as
\begin{eqnarray}
 \label{Eqn:ObjectiveReformulated1}
D^{n}(\mathbb{W},\ulinep)\overset{(a)}{=} \sum_{\ulineh \in 
\nhist}\sum_{d=1}^{n}
\sum_{\substack{\ulineg \in \nhist:\\|\ulineg-\ulineh|_{1}=2d}}
{n\choose \ulineh}\ulinep^{\ulineh} \mathbb{W}(\ulineg | \ulineh)2d =
\sum_{\ulineh \in \nhist}{n\choose \ulineh}\ulinep^{\ulineh}
\sum_{d=1}^{n}2d \sum_{\substack{\ulineg \in 
\nhist:\\d_{G}(\ulineg,\ulineh)=d}} \mathbb{W}(\ulineg | \ulineh).
\end{eqnarray}

\begin{figure}\centering
\includegraphics[width=2in]{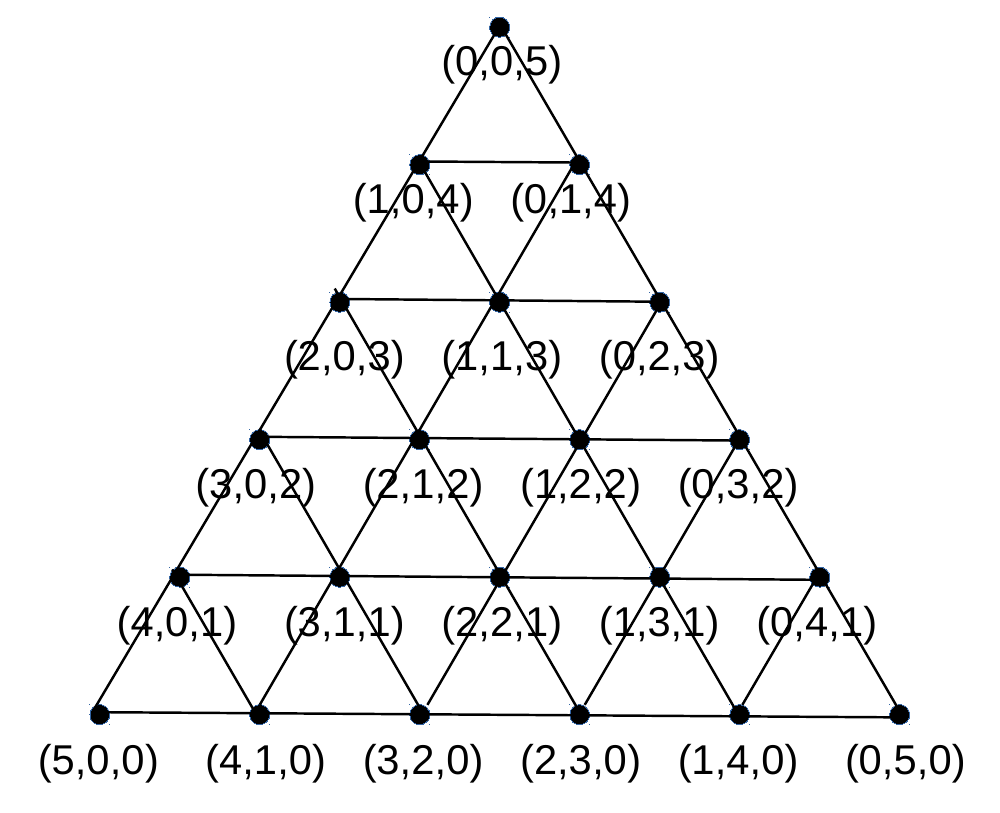}
\caption{Privacy-constraint graph for $k=3$, $n=5$.}
\label{Fig:PCGraphFork=3n=5}\end{figure}

In arriving at (\ref{Eqn:ObjectiveReformulated1})(a), we used the fact that for
any $\ulineg,\ulineh \in \nhist$, we have $|\ulineg-\ulineh|_{1}$ is an even
integer and at most $2n$. This is proven in Lemma
\ref{Lem:PropertiesPCGraph}(i), Appendix
\ref{AppSec:L1DistanceAndGraphDistance}. Consider a HSM $M : \nhist \Rightarrow
\nhist$ for which $\mathbb{W}(\ulineg|\ulineh) = f(\ulineh,
|\ulineg-\ulineh|_{1})$ is a function only of the distance between the vertices.
In the sequel, we will prove this sub-collection contains a mechanism that is
optimal in the limit $n \rightarrow \infty$. For such a HSM,
(\ref{Eqn:ObjectiveReformulated1}) reduces to
\begin{figure}
\centering
\includegraphics[width=6in]{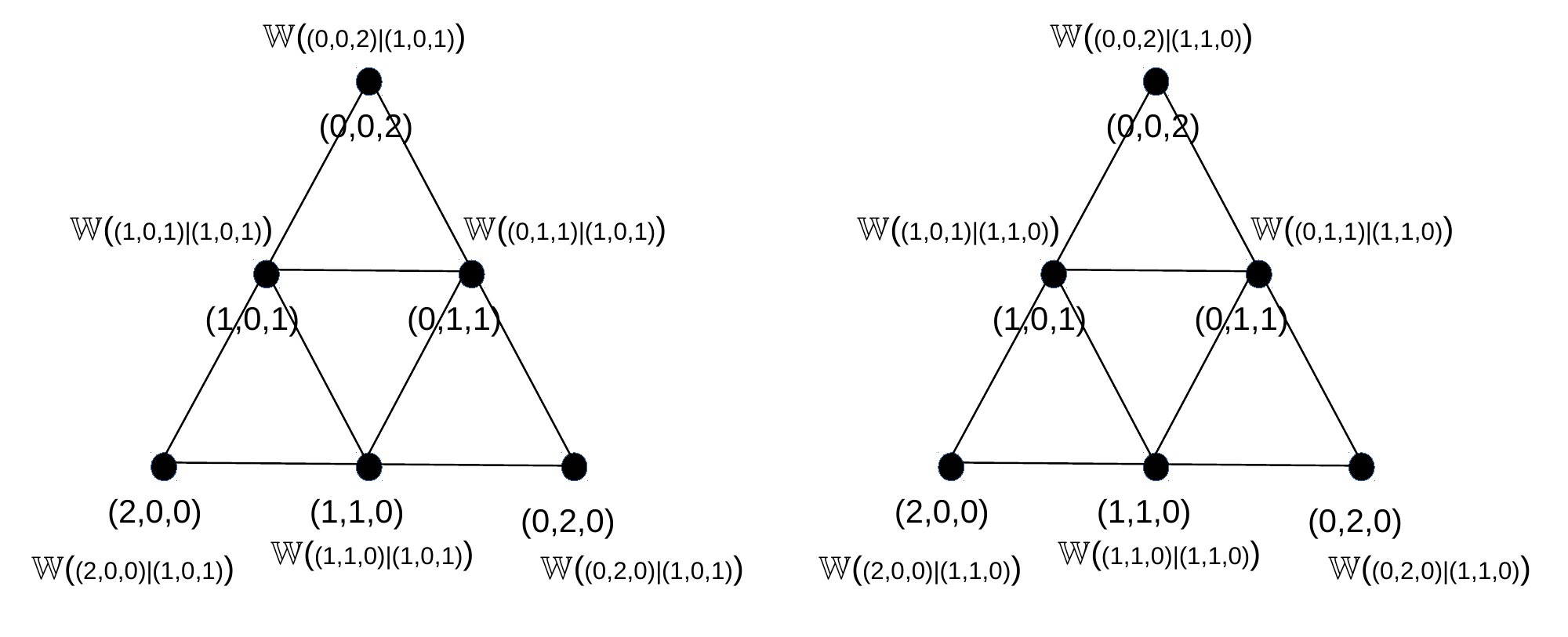}
\caption{The PC graphs for $K=3, N=2$ are depicted. The decision variables 
$(\mathbb{W}(\ulineg|(1,0,1)):\ulineg \in \hist_{3}^{2})$ are associated with 
the nodes of the graph on the left. On the right, the decision variables 
$(\mathbb{W}(\ulineg|(1,1,0)):\ulineg \in \hist_{3}^{2})$ are associated with 
the nodes of the graph. Since $(1,1,0)$ and $(1,0,1)$ are neighbors, at every 
node, the two values have to be within $\theta$ and $\frac{1}{\theta}$ of each 
other.}
\label{Fig:PCGraphFork=3n=3WithMechanisms}
\end{figure}
\begin{eqnarray}
\label{Eqn:GrossEstimateOfObjectiveFunction}D^{n}(\mathbb{W},\ulinep)=\sum_{
\ulineh \in \nhist}{n \choose \ulineh}\ulinep^{\ulineh}\sum_{d=1}^{n}2d
N_{d}(\ulineh)f(\ulineh, 2d), \mbox{ where } N_{d}(\ulineh) = |\left\{ \ulineg
\in \nhist: d_{G}(\ulineg,\ulineh)=d\right\}|\end{eqnarray}is the number of
vertices at graph distance $d$ from $\ulineh$. To evaluate the RHS of
$D^{n}(\mathbb{W},\ulinep)$ above, we will need to
characterize the sum $\sum_{d=1}^{n}d N_{d}(\ulineh)f(\ulineh, d)$. Let us
consider the sequence $N_{1}(\ulineh),N_{2}(\ulineh),\cdots, N_{n}(\ulineh)$
which may be regarded as the distance distribution of the vertex $\ulineh \in V
= \nhist$. Consider Fig. \ref{Fig:DistanceDistributionHardens} and two sequences
$(N_{d}(\ulineh):d=1,2,\cdots)$ and
$(N_{d}(\underline{\tilde{h}}):d=1,2,\cdots)$ for any pair $\ulineh,
\underline{\tilde{h}} \in V$ within the dotted circle. These sequences agree on
the initial few terms, henceforth referred to as the \textit{head}, and disagree
in a few subsequent terms due to the presence of the boundary. As the boundary
recedes (i.e., $n \rightarrow \infty$), the first term of disagreement recedes,
and the head elongates. Alternatively stated, the heads of the sequences
$(N_{d}(\ulineh):d=1,2,\cdots)$ for $\ulineh$ within the dotted circle become
invariant with $\ulineh$. Formally, there exists a distance $r 
\in \naturals$ such that, for every $\ulineh$ in the dotted circle,
$N_{d}(\ulineh) \rightarrow N_{d}$ for all $d=1,2,\cdots,r-1$. Moreover $r
\rightarrow \infty$ as the boundary recedes, i.e., $n \rightarrow \infty$. We
characterize $N_{d}$ by considering $\ulinec \define
n\ulinep$. Observe that
\begin{figure}
\centering
\includegraphics[width=9in,angle=90]{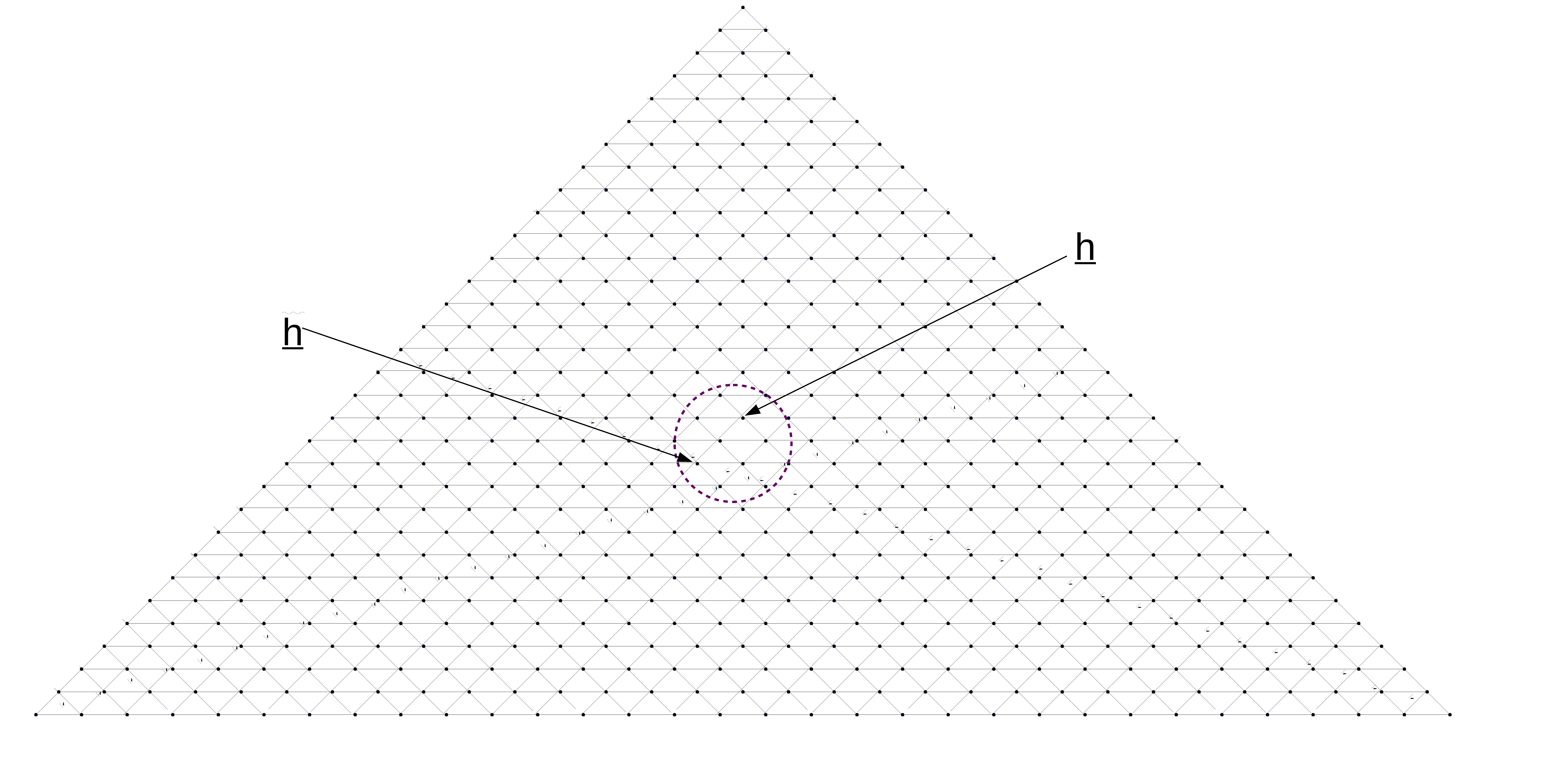}
\caption{The dotted circle within which the distance distribution of 
nodes is considered.}
\label{Fig:DistanceDistributionHardens}
\end{figure}
\begin{eqnarray}
 \label{Eqn:SetsInCorrespondence}
N_{d}(\ulinec) = |\left\{ \ulineg \in \nhist: d_{G}(\ulineg,\ulinec)=d\right\}|= |\left\{ \ulinez \in \integers^{K}: \ulinec + \ulinez \in \nhist, |\ulinez|_{1}=2d \right\}| ~~~~~~~~~~~~~~~~~~~\nonumber\\
= |\{ \ulinez \in \integers^{K} \!\!: c_{i}+z_{i} \geq 0,
\sum_{i=1}^{K}\!c_{i}+z_{i}=n, |\ulinez|_{1}=2d \}|\nonumber= |\{ \ulinez \in
\integers^{K}\!\! : z_{i} \geq -np_{i}, \sum_{i=1}^{K}z_{i}=0, |\ulinez|_{1}=2d
\}|. \nonumber
\end{eqnarray}
As $n \rightarrow \infty$, the lower bound on $z_{i}$ vanishes (becomes 
redundant), and we have
\begin{eqnarray}
 \label{Eqn:InfiniteGraph}
 &N_{d}(\ulinec) \rightarrow N_{d} \define \left|\{\ulinez \in \mathbb{Z}^{k}: \sum_{k=1}^{K}z_{k}=0,|\ulinez|_{1} = 2d \}\right|.&
\end{eqnarray}
$N_{d}$ is the number of \textit{integer} points on the \textit{face} of the \textit{integral convex polytope}
\begin{eqnarray}
 \label{Eqn:Cross-Polytope}
 &\mathcal{P}_{d} = \{ (x_{1},\cdots,x_{K}) \in \reals^{K}: \sum_{k=1}^{K}x_{k} = 0, \sum_{k=1}^{K}|x_{k}| \leq 2d \}.&
\end{eqnarray}

Indeed, if $L_{\mathcal{P}}(d) \define |\integers^{K} \cap \mathcal{P}_{d}|$, 
then $N_{d} = L_{\mathcal{P}}(d)-L_{\mathcal{P}}(d-1)$.\footnote{If
$(x_{1},\cdots,x_{K}) \in \integers^{K}$ and $\sum_{k=1}^{K}x_{k} = 0$, then
$\sum_{k=1}^{K}|x_{k}|$ is an even integer. This follows in a straightforward
manner from Lemma \ref{Lem:PropertiesPCGraph}(i), Appendix
\ref{AppSec:L1DistanceAndGraphDistance}. Therefore, if $(x_{1},\cdots,x_{K}) \in
\integers^{K}\cap\mathcal{P}_{d}$ and $\sum_{k=1}^{K}|x_{k}|< 2d$, then
$\sum_{k=1}^{K}|x_{k}| \leq 2(d-1)$.}
Notice that $L_{\mathcal{P}}(d)$ is the number of integral points in the $d-$th dilation of 
the integral convex polytope $\mathcal{P} \define \mathcal{P}_{1}$. $L_{\mathcal{P}}(d)$ and 
its generating function play a central role in this paper. Ehrhart
theory concerns 
the enumeration of integer points in a integral convex polytope and the objects 
associated with these counts. We present the foundational results in Ehrhart theory 
that we will have opportunity to use. The 
reader is referred to \cite{BeckRobins-2015} for a beautiful exposition of 
Ehrhart theory.

A convex $l-$polytope is a convex polytope of dimension $l$. A convex $l-$polytope whose 
vertices have integral co-ordinates is an  integral convex $l-$polytope. 
$L_{\mathcal{P}}(d)$ is the number of integral points in the $d-$th dilation of the 
integral convex $l-$polytope (Fig. \ref{Fig:LatticePointEnumeration}). Our pursuit of $L_{\mathcal{P}}(d)$ and the associated 
objects is aided by the following fundamental theorem of Ehrhart.
Ehrhart's theorem states that if $\mathcal{P}$ is an integral convex $l-$polytope, 
then $L_{\mathcal{P}}(d)$ is a polynomial in $d$ of degree $l$. We refer to
$L_{\mathcal{P}}(d)$ as \textit{Ehrhart's polynomial}. We will identify $N_{d}$,
and hence $L_{\mathcal{P}}(d)$, precisely in our proof. As evidenced by
(\ref{Eqn:FirstCharacterization}), we will have opportunity to study the
generating function of the  
counts $L_{\mathcal{P}}(d):d \in \naturals$. We refer to the formal power series
\begin{eqnarray}
 \label{Eqn:EhrharSeries}\!\!\!\!
\!\!&\EhrSer_{\mathcal{P}}(z) = 1 + 
\sum_{d=1}^{\infty}L_{\mathcal{P}}(d)z^{d}\mbox{ as the \textit{Ehrhart series} 
of $\mathcal{P}$, and let }\EhrFaceSer(z)\define(1-z)\EhrSer(z).&
\end{eqnarray}
Since $N_{d} = L_{\mathcal{P}}(d)-L_{\mathcal{P}}(d-1)$, we have 
$\EhrFaceSer(\theta) = (1-\theta)\EhrSer_{\mathcal{P}}
(\theta) \overset{}{=} 
1+\displaystyle\sum_{d=1}^{\infty}N_{d}\theta^{d}$. 

Having introduced the tools, we now sketch the main elements of the proof. In 
this section, we first argue that the RHS of 
(\ref{Eqn:FirstCharacterization}) is an upper bound on 
$D_{K}^{*}(\theta,\ulinep)$.

\subsection{Upper bound}
\label{SubSec:UpperBound}
Suppose one were to consider the popular Laplace/geometric/staircase mechanism
$\mathscr{G}:\nhist \Rightarrow \nhist$ and characterize its distortion. In that
case,
\begin{equation}
\label{Eqn:AttemptAtGeometricMechanism}
\WOfGeomMech(\ulineg | \ulineh) \propto
\theta^{\frac{|\ulineg-\ulineh|_{1}}{2}} \mbox{ and hence }\WOfGeomMech(\ulineg
| \ulineh) = \frac{\theta^{d_{G}(\ulineg,\ulineh)}}{E_{\ulineh}(\theta)}, \mbox{
where }E_{\ulineh}(\theta) =
1+\sum_{d=1}^{n}N_{d}(\ulineh)\theta^{d}\end{equation}
is a normalizing constant chosen to ensure $\sum_{\ulineg \in
\nhist}\WOfGeomMech(\ulineg|\ulineh) = 1$. It will be apparent that
$\WOfGeomMech(\cdot| \ulineh)$ is $\theta-$DP only if $E_{\ulineh}(\theta)$ is
invariant with $\ulineh$. For any (finite) $n \in \naturals$ this is not true,
leading to obstacles in defining a feasible $\theta-$DP HSM analog to the
geometric mechanism. We overcome this by considering a cascade mechanism. See 
Figure \ref{Fig:CascadeMechanism}. $\mathbb{U}^{n}$ is analogous to the 
geometric mechanism $\mathbb{W}_{\mathcal{G}}$ and outputs a `histogram' in an 
`extended set of histograms'. This overcomes the issue of $E_{\ulineh}(\theta)$ 
being variant with $\ulineh$. An `extended histogram' is then remapped back to 
a histogram $\ulineh \in \nhist$ via the truncation mechanism $\mathbb{V}^{n}$. 
$\mathbb{V}^{n}(\cdot | 
\cdot)$ is so chosen such that effective expected $\mathbb{L}_{1}-$distortion 
does not increase, in the limit. Reserving these elements to the proof, we put 
forth a heuristic limiting argument that explains the effective distortion of 
the cascade mechanism in Figure \ref{Fig:CascadeMechanism}. As $n \rightarrow 
\infty$, we noted that
$N_{d}(\ulineh) \rightarrow N_{d}$ and becomes invariant with $\ulineh$, and
hence it is plausible that (i) $E_{\ulineh}(\theta) \rightarrow
\EhrFaceSer(\theta)$, where $\EhrFaceSer(\theta) 
\define 1+\displaystyle\sum_{d=1}^{\infty}N_{d}\theta^{d} 
 = (1-\theta)\EhrSer_{\mathcal{P}}$, and (ii) $\WOfGeomMech(\ulineg|\ulineh) 
\rightarrow
(\EhrFaceSer(\theta))^{-1}\theta^{d_{G}(\ulineg,\ulineh)}$. We substitute this
in the RHS of (\ref{Eqn:ObjectiveReformulated1}), to obtain

\begin{eqnarray}
\lefteqn{\lim_{n
\rightarrow \infty}D^{n}(\WOfGeomMech,\ulinep)=
\lim_{n \rightarrow \infty}\sum_{\ulineh \in \nhist}{n\choose 
\ulineh}\ulinep^{\ulineh}
\sum_{d\geq 1}2d \sum_{\substack{\ulineg \in 
\nhist:\\d_{G}(\ulineg,\ulineh)=d}}\!\!\!\!\!\! 
\frac{\theta^{d_{G}(\ulineg,\ulineh)}}{\EhrFaceSer(\theta)}}\nonumber\\ 
&=& 
\label{Eqn:FidelityOfLimitingGeometricMechanism}
\lim_{n \rightarrow \infty}\!\sum_{\ulineh \in \nhist}\!\!{n\choose 
\ulineh}\ulinep^{\ulineh}
\sum_{d\geq 1}\frac{ 2dN_{d}\theta^{d}}{\EhrFaceSer(\theta)} \!=\lim_{n 
\rightarrow \infty}\!\sum_{\ulineh \in \nhist}\!\!{n\choose 
\ulineh}\ulinep^{\ulineh} \frac{2\theta}{\EhrFaceSer(\theta)}
\frac{d\EhrFaceSer(\theta)}{d\theta} = \lim_{n \rightarrow 
\infty} \frac{2\theta}{\EhrFaceSer(\theta)}
\frac{d\EhrFaceSer(\theta)}{d\theta}, \\
&=& \frac{2\theta}{\EhrFaceSer(\theta)}
\frac{d\EhrFaceSer(\theta)}{d\theta} =  
\frac{2\theta}{\EhrSer_{\mathcal{P}}(\theta)}
\frac{d\EhrSer_{\mathcal{P}}(\theta)}{d\theta}-\frac{2\theta}{1-\theta}, 
\label{Eqn:TheDiffEhrSeries}
\end{eqnarray}
and the latter quantity is invariant with $n$, enabling us conclude that
\begin{eqnarray} 
\label{Eqn:BothSeries}
\lim_{n
\rightarrow
\infty}D^{n}(\WOfGeomMech,\ulinep)=\frac{2\theta}{\EhrFaceSer(\theta)}
\frac{d\EhrFaceSer(\theta)}{d\theta} =
\frac{2\theta}{\EhrSer_{\mathcal{P}}(\theta)}
\frac{d\EhrSer_{\mathcal{P}}(\theta)}{d\theta}-\frac{2\theta}{1-\theta}. 
\end{eqnarray}
In arriving at
(\ref{Eqn:FidelityOfLimitingGeometricMechanism}), we used the fact that
$\frac{d\EhrFaceSer(\theta)}{d\theta} = \displaystyle \sum_{d\geq 
1}dN_{d}\theta^{d-1}$, and in arriving at (\ref{Eqn:TheDiffEhrSeries}) we used 
$\EhrFaceSer(\theta) 
\define 1+\displaystyle\sum_{d=1}^{\infty}N_{d}\theta^{d} 
 = (1-\theta)\EhrSer_{\mathcal{P}}$. These informal arguments provide a
heuristic explanation for (\ref{Eqn:FirstCharacterization}) and leaves certain
interesting and non-trivial elements, that are addressed in 
Section \ref{SubSec:UpperBound}.
\begin{remark}
 \label{Rem:NotPlainVanillaGeometricMechansism}{\rm We side-stepped
the question of identifying a $\theta-$DP mechanism for any $n \in
\naturals$. Characterizing a (truncated) geometric $\theta-$DP mechanism for a
general $K$ is non-trivial owing to the presence of multiple boundary vertices,
the involved geometry of the PC graph, and lack of an expression for
the `tail' sum.\footnote{The reader is encouraged to construct, via a truncation
or otherwise, a $\theta-$DP mechanism analogous to the geometric 
mechanism, for the case of $K=3$ and $n=5$ depicted in 
Fig.~\ref{Fig:PCGraphFork=3n=5}, to
recognize the non-triviality.} It is also worth noting that the often used
technique of enlarging the output space to be continuous followed by a heuristic
map does not permit a precise
performance characterization. Moreover, as we note in the
following proof, we are required to shape the geometric mechanism
appropriately to minimize expected distortion.}\end{remark}

Next, we show (\ref{Eqn:IntroOurEhrhartSeries}). Towards that end, we 
characterize $N_{d}$ explicitly. 
We recognize that an explicit characterization for $N_{d}$ or 
$L_{\mathcal{P}}(d)$ will enable us express the power 
series in (\ref{Eqn:BothSeries}). In general, characterizing the Ehrhart 
polynomial of a convex polytope is involved (see \cite{BeckRobins-2015}). 
However, in our case we are able to characterize $N_{d}$ for the cross-polytope 
$\mathcal{P}_{d}$ in (\ref{Eqn:Cross-Polytope-Early}). Recall $N_{d} = 
|\mathcal{S}_{d}|$, where 
\begin{eqnarray}\label{Eqn:CrossPolytopeIntegerCo-Ordinates}
\mathcal{S}_{d} \define \mathbb{Z}^{K}\cap\mathcal{P}_{d} =  \left\{ 
(x_{1},\cdots,x_{K}) \in \integers^{K}: \sum_{k=1}^{K}x_{k} = 0, 
\sum_{k=1}^{K}|x_{k}| \leq 2d\right\}.\nonumber
\end{eqnarray}
$\mathcal{S}_{d}$ can be partitioned into \textit{disjoint} sets based on the 
coordinates (in set $A_{|P|}$ below) corresponding to its non-negative indices. 
Let 
\begin{eqnarray}
 \label{Eqn:SetWithNonNegativeAndNegativeCo-ords}
 A_{n} &\define& \left\{  (a_{1},\cdots,a_{n}) \in \mathbb{Z}^{n} : a_{i}\geq 0,
\sum_{i=1}^{n}a_{i}=d \right\},\nonumber\\ B_{m} &=& \left\{ 
(b_{1},\cdots,b_{m}) \in \mathbb{Z}^{m} : b_{j}< 0,~ -\sum_{j=1}^{m}b_{i} = d
\right\} =\left\{  (b_{1},\cdots,b_{m}) \in \mathbb{Z}^{m} : b_{j}> 0 
\sum_{j=1}^{m}b_{i} = d \right\}. \nonumber
\end{eqnarray}
It can be verified that,
\begin{eqnarray}
 \label{Eqn:CountForNd}
 \mathcal{S}_{d} = \bigcup_{P \subseteq [K]} A_{|P|}\times B_{K-|P|} =
\bigcup_{P \subseteq [K]} A_{K-|P|}\times B_{|P|}. \nonumber
\end{eqnarray}
We can now compute $|A_{|P|}|$ and $|B_{|P|}|$. Since
\begin{eqnarray}
 \label{Eqn:TheSizeOfSetsInTheCartesianProduct}
 |A_{n}| = { d+n-1 \choose n-1}, |B_{m}| = {d-1 \choose m-1}, \mbox{ we have }
N_{d} &=& \sum_{r=1}^{K-1}{K \choose r} { d+r-1 \choose r-1}{d-1 \choose K-r-1}
\nonumber\\&=&  \sum_{r=1}^{K-1}{K \choose r} { d+K-r-1 \choose K-r-1}{d-1
\choose r-1},\nonumber
\end{eqnarray}
where the running variable $r$ denotes the cardinality of the (running set) $P 
\subseteq [K]$. An alternate count can be obtained by explicitly considering 
the set of zero coordinates. Suppose $0\leq z \leq K-1$ denotes the number of 
$0-$coordinates, and $p$ the number of positive co-ordinates, then, for $d 
\geq 
1$, it can be verified that
\begin{eqnarray}
 \label{Eqn:UsingZero-Coordinates}
 \mathcal{S}_{d} &=& \bigcup_{\substack{Z \subseteq [K]:|Z|\\\leq K-2}} 
\bigcup_{\substack{P \subseteq [K]\setminus Z:\\1\leq|P|\\\leq
K-{Z}-1}}\!\!\!\!\!\!\!B_{|P|}\times B_{K-|P|-|Z|}, \mbox{ and
hence }N_{d}= \sum_{z=0}^{K-2}\sum_{p=1}^{K-z-1}{K \choose z} {K-z
\choose p}{ d-1 \choose p-1}{d-1 \choose K-z-p-1}. \nonumber
\end{eqnarray}
So, we conclude
\begin{eqnarray}
 \label{Eqn:E_PfOfThetaInTermsOfNd}
 \EhrFaceSer(\theta) &=& 1 + \sum_{d=1}^{\infty}\left\{ \sum_{r=1}^{K-1}{K
\choose r} { d+r-1 \choose r-1}{d-1 \choose K-r-1}\right\}\theta^{d}
\\&=& 1 + \sum_{d=1}^{\infty}\left\{ \sum_{r=1}^{K-1}{K \choose r} {
d+K-r-1 \choose K-r-1}{d-1 \choose r-1}\right\}\theta^{d} \nonumber\\
 &=& 1 + \sum_{d=1}^{\infty}\left\{  \sum_{z=0}^{K-2}\sum_{p=1}^{K-z-1}{K
\choose z} {K-z \choose p}{ d-1 \choose p-1}{d-1 \choose K-z-p-1}
\right\}\theta^{d}. \nonumber
\end{eqnarray}
Finally, we show (\ref{Eqn:DK*OfTheta}). We refer to \cite[Eqn 
(3.8)]{1997TRS_ConSlo} for an alternate characterization for $N_{d}$. It may be 
verified that points on the root lattice $A_{K-1}$ at fractional height $d$ in 
\cite{1997TRS_ConSlo} correspond to vertices on the face of $\mathcal{P}_{d}$ 
in (\ref{Eqn:Cross-Polytope}). \cite{1997TRS_ConSlo} also refers to these 
vertices as being at a distance $d$ or $d$ bonds away. From \cite[Eqn 
(3.8)]{1997TRS_ConSlo}, we have
\begin{eqnarray}
N_{d} &=& \sum_{r=1}^{K-1}{K \choose r} { d+r-1 \choose r-1}{d-1 \choose K-r-1} 
= \sum_{j=0}^{K-1} \left[{K-1 \choose j}\right]^{2}{d+K-j-2 \choose K-2} 
\nonumber \\ &=&  
\label{Eqn:ConwaySloane}
\sum_{j=0}^{K-1} \left[{K-1 \choose 
j}\right]^{2}{d+K-j-2 \choose d-j}. 
\end{eqnarray}
We now use RHS of (\ref{Eqn:ConwaySloane}) to conclude
\begin{eqnarray}
 \EhrFaceSer(\theta)\!\!\!\! &=&\!\!\!\! 1+\sum_{d \geq 1}{\theta}^{d} \left\{  
\sum_{j=1}^{K-1} {K \choose j} {d+j-1 \choose j-1} {d-1 \choose K-j-1} \right\} 
= \sum_{l\geq 0}{\theta}^{l}\left\{ \sum_{j=0}^{k-1}{l-j+K-2 \choose l-j} 
\left[ 
{K-1 \choose j}\right]^{2} \right\} 
\nonumber\\
 &=&\sum_{l\geq 0}{l+K-2 \choose l}\left\{ \sum_{j=0}^{K-1}\left[ {K-1 \choose 
j}\right]^{2}{\theta}^{j+l} \right\}=\frac{\sum_{j=0}^{K-1} 
\left[ {K-1 \choose j}\right]^{2}{\theta}^{j}}{(1-{\theta})^{K-1}} = 
\frac{S_{K-1}({\theta})}{(1-{\theta})^{K-1}} \label{Eqn:UsingConway}.
\end{eqnarray}
Substituting (\ref{Eqn:UsingConway}) in (\ref{Eqn:BothSeries}), we obtain 
(\ref{Eqn:DK*OfTheta}).

We identify a sequence of upper bounds $D_{n}^{u}(\theta) \!\geq\!
D^{n}_{*}(\theta,\ulinep)\!:\! n \in \naturals$ and 
characterize the corresponding limit $\lim_{n \rightarrow \infty} 
D_{n}^{u}(\theta) $ to obtain an upper bound on $D_{K}^{*}(\theta,\ulinep)$. 
For this, we identify a sequence $\mathbb{W}^{n}\!:\! \nhist \Rightarrow
\nhist:\!
n \in \naturals$ of $\theta-$DP HSMs and let $D^{u}_{n}(\theta) \define
D(\mathbb{W}^{n},\ulinep)$.

In view of Remark \ref{Rem:NotPlainVanillaGeometricMechansism}, we
propose $\mathbb{W}^{n}: \nhist \Rightarrow \nhist$ as a cascade of mechanisms
$\mathbb{U}^{n}:\nhist \Rightarrow \nhistext$ and $\mathbb{V}^{n}:\nhistext
\Rightarrow \nhist$. See Figure \ref{Fig:CascadeMechanism}. $\mathbb{U}^{n}$ is 
a geometric mechanism and outputs
`histograms' from an `enlarged set of histograms'. This overcomes technical
obstacles mentioned in Remark \ref{Rem:NotPlainVanillaGeometricMechansism}.
$\mathbb{V}^{n}$ takes as input only the output of $\mathbb{U}^{n}$, and remaps
$\nhistext$ to $\nhist$. More importantly, it shapes the joint
distribution to minimize the expected distortion. Since a geometric mechanism 
is,
in general, optimal in most DP settings, and $\mathbb{V}^{n}$ is carefully
shaped, we obtain a reasonably good sequence $\mathbb{W}^{n}$ of mechanisms that
is, in the limit, optimal.

In establishing the upper bound, we first specify mechanisms
$\mathbb{U}^{n}$, $\mathbb{V}^{n}$ and characterize the distortion 
$D(\mathbb{U}^{n})$ of $\mathbb{U}^{n}$. Next, we relate $ 
D(\mathbb{W}^{n},\ulinep) (=
D^{u}_{n}(\theta))$ to $D(\mathbb{U}^{n})$ and thereby characterize the former 
as an upper bound.

We take a clue from
(\ref{Eqn:AttemptAtGeometricMechanism}) and
Remark \ref{Rem:NotPlainVanillaGeometricMechansism}. The normalizing terms
$E_{\ulineh}(\theta)$, $E_{\underline{\tilde{h}}}(\theta)$ differ because the
tails of the sequences $N_{d}(\ulineh): d \geq 1$ and $N_{d}(\ulinetildeh):d\geq
1$ differ. The latter is due to the presence of the boundary of $\nhist$ (or the
PC graph). We enlarge $\nhist$ to eliminate the boundary. This we do by getting
rid of the non-negativity constraint in (\ref{Eqn:HistogramModelling}). The
enlarged `set of histograms' is therefore $\nhistext \define \{
(h_{1},\cdots,h_{K}) \in \mathbb{Z}^{K}: \sum_{k=1}^{K}h_{k} = n  \}$.
$\nhistext$ is isomorphic to $\{ \ulinez \in \integers^{K}: \sum_{k=1}^{K}z_{k}
= 0\}$ and 
\begin{eqnarray}
 \label{Eqn:InfiniteGraphCopy}
 N_{d} \define \left|\{\ulinez \in \mathbb{Z}^{k}: 
\sum_{k=1}^{K}z_{k}=0,|\ulinez|_{1} = 2d \}\right|,&
\end{eqnarray}
defined identical to (\ref{Eqn:InfiniteGraph}), is the number of
`extended histograms' at an $\mathbb{L}_{1}$ distance of $2d$ from \textit{any}
element in $\nhistext$. $N_{d}$ being invariant with $\ulineh$, we define a
$\theta-$DP mechanism $\mathbb{U}^{
n} : \nhist \Rightarrow \nhist_{\mbox{\tiny ext}}$ analogous to the geometric 
mechanism in (\ref{Eqn:AttemptAtGeometricMechanism}) as
\begin{eqnarray}
\label{Eqn:DefnOfUn}\mathbb{U}^{n}(\ulineg|\ulineh) = 
\left(\EhrFaceSer(\theta)\right)^{-1}\theta^{\frac{|\ulineg-\ulineh|_{1}}{2}},
\end{eqnarray}
where $\mathcal{P}$ is the convex polytope whose $d^{th}-$dilation is 
\begin{eqnarray}\mathcal{P}_{d} = \{ (x_{1},\cdots,x_{K}) \in \reals^{K}: 
\sum_{k=1}^{K}x_{k} = 0, \sum_{k=1}^{K}|x_{k}| \leq 2d \}. \nonumber
\end{eqnarray}
In order to prove $\mathbb{W}^{n}$ is 
$\theta-$DP, it suffices to prove $\mathbb{U}^{n}$ is $\theta-$DP. Indeed, by 
the post-processing theorem of DP, so long
as $\mathbb{V}^{n}:\nhistext \Rightarrow \nhist$ takes only the output of
$\mathbb{U}^{n}$ as input, the cascade mechanism $\mathbb{W}^{n}$ is 
$\theta-$DP. It is straightforward to prove that $\mathbb{U}^{n}$ is 
$\theta-$DP, 
and the steps are provided in Appendix \ref{AppSec:UIsThetaDP}.

Before we identify $\mathbb{V}^{n}(\cdot | \cdot)$, let us characterize the 
distortion of $\mathbb{U}^{n}$. Let
\begin{eqnarray}
D(\mathbb{U}^{n}) \define \sum_{\ulineh \in \nhist} \sum_{\ulineg \in 
\nhistext} {n \choose \ulineh } 
\ulinep^{\ulineh}|\mathbb{U}^{n}(\ulineg|\ulineh)\ulineg -\ulineh|_{1}  
\label{Eqn:DistOfUn}
\end{eqnarray}
denote the distortion of $\mathbb{U}^{n}$. From (\ref{Eqn:DefnOfUn}), 
(\ref{Eqn:DistOfUn}), 
we have
\begin{eqnarray}
 \label{Eqn:DistortionWrtUChannel}
 D({\mathbb{U}}^{n}) &=& \sum_{\ulineh \in \nhist}\sum_{\ulineg \in 
\nhist_{\mbox{\tiny ext}}} 
{n\choose \ulineh} \ulinep^{\ulineh} \mathbb{U}^{n}(\ulineg | 
\ulineh)|\ulineg-\ulineh|_{1} 
= \sum_{\ulineh \in \nhist}{n\choose \ulineh}\ulinep^{\ulineh}\sum_{\ulineg \in 
\nhist_{\mbox{\tiny ext}}}  
\frac{1}{\EhrFaceSer(\theta)}\theta^{\frac{|\ulineg-\ulineh|_{1}}{2}}
|\ulineg-\ulineh|_{1}\nonumber\\
 &=& \sum_{\ulineh \in \nhist}{n\choose \ulineh}\ulinep^{\ulineh}
\frac{1}{\EhrFaceSer(\theta)} \sum_{d\geq 1} \sum_{\substack{ \ulineg \in 
\nhistext\\|\ulineg-\ulineh|_{1}=2d }} \!\!\!\!\! 2d\theta^{d} = \sum_{\ulineh 
\in \nhist}{n\choose \ulineh}\ulinep^{\ulineh}
\frac{1}{\EhrFaceSer(\theta)} \sum_{d\geq 1}  2dN_{d}\theta^{d} 
\nonumber\\&=&\sum_{\ulineh \in \nhist}{n\choose \ulineh}\ulinep^{\ulineh}
\frac{2\theta}{\EhrFaceSer(\theta)}  \frac{d\EhrFaceSer(\theta)}{d\theta}
\label{Eqn:DerivativeOfEhrhartPolynomialAsTheFidelity} =
\frac{2\theta}{\EhrFaceSer(\theta)} \frac{d\EhrFaceSer(\theta)}{d\theta},
 \end{eqnarray}
where (\ref{Eqn:DerivativeOfEhrhartPolynomialAsTheFidelity}) follows from steps 
identical to those that lead to (\ref{Eqn:TheDiffEhrSeries}). 

The choice of $\mathbb{V}^{n}$ is based on the fact that the DBs whose
histograms differ widely from the mean histogram $n\ulinep$ contribute
an exponentially (in $n$) small amount to the expected value.
$\mathbb{V}^{n}$ maps the histogram outside the $\mathbb{L}_{1}-$ball of radius
$Rn^{\frac{2}{3}}$ centered 
at $n\ulinep$ to the histogram $n\ulinep$. The histograms within radius $Rn^{\frac{2}{3}}$ of $n\ulinep$ 
remain unchanged. Formally, let 
\begin{eqnarray}
 \label{Eqn:DefnOfVnExplicitlyStated}
\mathbb{V}^{n}(\ulineg|\ulineh) = 1\mbox{ if }\ulineg=\ulineh,|\ulineh -
n\ulinep|_{1} \leq Rn^{\frac{2}{3}},~~ \mathbb{V}^{n}(\ulineg|\ulineh) = 1
\mbox{ if }\ulineg=n\ulinep, |\ulineh - n\ulinep|_{1} >
Rn^{\frac{2}{3}},\nonumber
\end{eqnarray}
and $\mathbb{V}^{n}(\ulineg|\ulineh)=0$ otherwise. For completeness, we also note $\mathbb{W}^{n}(\ulineg|\ulineh) = \sum_{\ulineb \in \nhistext}\mathbb{V}^{n}(\ulineg|\ulineb)\mathbb{U}^{n}(\ulineb|\ulineh)$.

Does $\mathbb{V}^{n}$ output a histogram in $\nhist$? The output of
$\mathbb{V}^{n}$ is contained within a $\mathbb{L}_{1}-$ball of radius
$\alpha_{n} = Rn^{\frac{2}{3}}$ centered at $n\ulinep \in \nhist$. The boundary
of $\nhist$ is at a $\mathbb{L}_{1}-$distance of at least $\beta_{n}=
\min_{k=1,\cdots,K}np_{k}$ from $n\ulinep \in \nhist$. Since $p_{k}>0$ for all
$k \in [K]$, as $n \rightarrow \infty$, $\alpha_{n} \leq \beta_{n}$, and the
range of $\mathbb{V}^{n}$ is contained within $\nhist$. The output of mechanism
$\mathbb{V}^{n}$ is indeed a histogram. We provide a formal proof below.

We recall $\mathbb{V}^{n}: \nhistext \rightarrow \nhist$ is defined as
\begin{eqnarray}
 \label{Eqn:DefnOfVnExplicitlyStated}
 \mathbb{V}^{n}(\ulineg|\ulineh) = \left\{\begin{array}{ll}
1&\mbox{if }\ulineg=\ulineh,|\ulineh - n\ulinep|_{1} \leq Rn^{\frac{2}{3}}\\
1&\mbox{if }\ulineg=n\ulinep, |\ulineh - n\ulinep|_{1} > Rn^{\frac{2}{3}},\\
0&\mbox{otherwise,}
\end{array} \right.\mbox{ and }\mathbb{W}^{n}(\ulineg|\ulineh) = \sum_{\ulineb 
\in 
\nhist}\mathbb{V}^{n}(\ulineg|\ulineb)\mathbb{U}^{n}(\ulineb|\ulineh),\nonumber
\end{eqnarray}
where $R > 0$ is any constant invariant with $n$. Since $\mathbb{V}^{n}$ is a 
deterministic map, it can also be defined through the map 
$f_{\mathbb{V}^{n}}:\nhistext \Rightarrow \nhist$ where
\begin{eqnarray}
 \label{Eqn:DefnOfVn}
f_{\mathbb{V}^{n}}(\ulineh) = \left\{\begin{array}{ll}
\ulineh&\mbox{if }|\ulineh - n\ulinep|_{1} \leq Rn^{\frac{2}{3}}\\
n\ulinep&\mbox{otherwise, i.e., }|\ulineh - n\ulinep|_{1} > Rn^{\frac{2}{3}},
\end{array} \right.\mbox{ and }\mathbb{V}^{n}(\ulineg|\ulineh) = 
\mathds{1}_{\left\{ \ulineg = f_{\mathbb{V}^{n}}(\ulineh)\right\}},\nonumber
\end{eqnarray}
where $R>0$ is a constant, invariant with $n$.
Let us analyze what `extended histograms' are within the range of 
$f_{\mathbb{V}^{n}}$. $\ulineh \in \nhistext$ falls in the range of 
$f_{\mathbb{V}^{n}}$, or in other words, is output by mechanism $\mathbb{V}^{n}$ 
only if $|\ulineh-n\ulinep|\leq Rn^{\frac{2}{3}}$, which is true only if 
$|h_{k}-np_{k}|\leq Rn^{\frac{2}{3}}$. The latter is equivalent to 
$np_{k}-Rn^{\frac{2}{3}} \leq h_{k} \leq np_{k}+Rn^{\frac{2}{3}}$ for every $ k 
\in [K]$. Observe that, since we assumed $p_{k}>0$ for all $k \in [K]$, the 
lower bound $np_{k}-Rn^{\frac{2}{3}}> 0$ for any $R>0$ and sufficiently large 
$n$. For sufficiently large $n$, $\mathbb{V}^{n}$ outputs an extended histogram 
whose coordinates are non-negative. From (\ref{Eqn:HistogramModelling}), and 
the definition $\nhistext$, the output of $\mathbb{V}^{n}$ is indeed a histogram 
from $\nhist$. Observe that, since we assumed $p_{k}>0$ for all 
$k=1,2\cdots,K$, we have $np_{i}-Rn^{\frac{2}{3}}> 0$ for any $R$ and 
sufficiently large $n$. Hence, for sufficiently large $n$, 
the output of mechanism $\mathbb{V}^{n}$ is indeed a histogram.

We now prove that $\lim_{n \rightarrow \infty}D(\mathbb{W}^{n},\ulinep) \leq
\lim_{n \rightarrow \infty}D(\mathbb{U}^{n})$. We describe the 
arguments before 
we provide the mathematical steps.
Let $D_{\ulineh}(\mathbb{W}^{n}) = \sum_{\ulineg \in
\nhist}\mathbb{W}^{n}(\ulineg|\ulineh)|\ulineg - \ulineh|_{1}$,
$D_{\ulineh}(\mathbb{U}^{n}) = \sum_{\ulineg \in
\nhistext}\mathbb{U}^{n}(\ulineg|\ulineh)|\ulineg - \ulineh|_{1}$ denote
(unweighted) contributions of $\ulineh$ to $D(\mathbb{W}^{n},\ulinep)$ and $
D(\mathbb{U}^{n})$ respectively. Refer to Fig.~\ref{Fig:ComparingFidelities}.
Let $B(\frac{1}{2})$ and $B(1)$ be the $\mathbb{L}_{1}-$balls centered at
$n\ulinep$ of radii $\frac{R}{2}n^{\frac{2}{3}}$ and ${R}n^{\frac{2}{3}}$
respectively. Let $B^{c}(1) := \nhistext \setminus B(1)$. For each $\ulineh
\in B(\frac{1}{2})$, the 
mechanism $\mathbb{V}^{n}$ has the effect of 
decreasing $\ulineh$'s contribution. In other words, for any $\ulineh \in
B(\frac{1}{2})$, $D_{\ulineh}(\mathbb{W}^{n}) \leq 
D_{\ulineh}(\mathbb{U}^{n})$.
\begin{figure}
\centering
\includegraphics[width=6in]{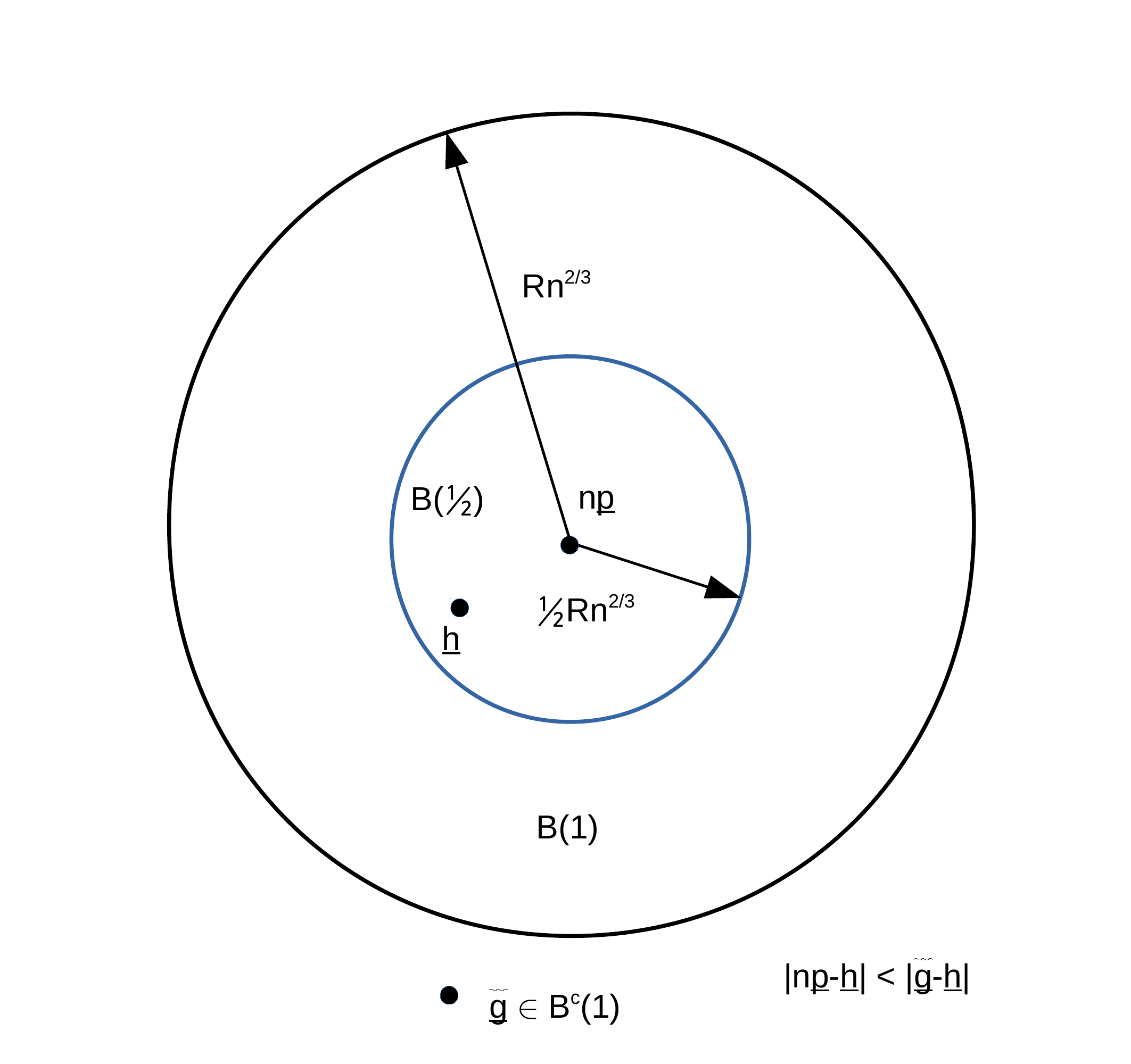}
\caption{}
\label{Fig:ComparingFidelities}
\end{figure}
This is because (i) $\mathbb{V}^{n}$ transfers mass placed on $\tilde{g} \in
B^{c}(1)$ - an element farther from $n\ulinep$ - to $n\ulinep$, and (ii)
$\mathbb{V}^{n}$ does not alter the mass placed on elements $\ulineg \in B(1)$
(other than $n\ulinep$).\footnote{This is made precise in the sequence of steps
(\ref{Eqn:SecondTermIsZero}) - (\ref{Eqn:FourthEquality}) below.}
What about for $\ulineh \in B^{c}(\frac{1}{2})$? The weights ${n \choose
\ulineh}\ulinep^{\ulineh}$ associated with these elements, when summed
up, contribute an exponentially small amount. Formally, $\sum_{\ulineh \in
B^{c}(\frac{1}{2})}{n \choose \ulineh}\ulinep^{\ulineh} \leq \exp\{-n\alpha \}$
for some $\alpha > 0$. Since $|\ulineg-\ulineh|_{1} \leq 2n$ whenever
$\ulineh,\ulineg \
in \nhist$, we have $D(\mathbb{W}^{n},\ulineh) \leq 2n\exp\{ -\alpha n\}$ and 
hence $\sum_{\ulineh \in B^{c}(\frac{1}{2})}{n \choose 
\ulineh}\ulinep^{\ulineh}D(\mathbb{W}^{n},\ulineh) \rightarrow 0$ as 
$n\rightarrow 0$. We flesh out these details in Appendix 
\ref{AppSec:DUnIs=DWn}.

From (\ref{Eqn:DerivativeOfEhrhartPolynomialAsTheFidelity}), 
(\ref{Eqn:TheDiffEhrSeries}) it suffices to characterize either the
Ehrhart series $\EhrSer_{\mathcal{P}}(\theta)$ or $\EhrFaceSer(\theta)$, of
$\mathcal{P} = \mathcal{P}_{1}$, where $\mathcal{P}_{d}$ is the polytope
characterized in (\ref{Eqn:Cross-Polytope}). From (\ref{Eqn:UsingConway}), we 
conclude
\begin{eqnarray}
 \label{Eqn:ProofOfUpperBound}
 D_{K}^{*}(\theta,\ulinep) \leq \frac{2\theta}{\EhrFaceSer(\theta)} 
\frac{d\EhrFaceSer(\theta)}{d\theta} = \AnalExpDis(\theta) \define 2\theta 
\left\{ \frac{K-1}{1-\theta} + 
\frac{S_{K-1}'(\theta)}{S_{K-1}(\theta)}  \right\} .
\end{eqnarray}

\begin{remark}
 \label{Rem:AsymptoticallyUniversallyOptimal}{\rm Observe that $\mathbb{U}^{n}$
is invariant with $\ulinep$, and $\mathbb{V}^{n}$ is a remapping mechanism that
depends on $\ulinep$ and reduces the expected distortion for histograms of high
probability. In order to prove Theorem \ref{Thm:AsymUnivOtimalExists}, it 
suffices to prove the lower bound, i.e., the reverse inequality in 
(\ref{Eqn:ProofOfUpperBound})}.
\end{remark}
\subsection{Lower Bound}
\label{SubSec:LowerBound}
Our proof of the lower bound is via the weak duality theorem. The weak duality 
theorem states that every feasible solution to the dual LP evaluates to a lower 
bound on the primal optimal. The reader is referred to Appendix 
\ref{AppSec:WDT} for precise statement of the WDT in the context of our 
problem. Consider the dual of the LP in (\ref{Eqn:DPProblemPosedAsALP}). If we 
can identify a dual feasible solution whose objective value evaluates to 
$C_{*}^{n}$ and $\lim_{n \rightarrow \infty} C_{*}^{n} = 
\AnalExpDis(\theta)$ defined in (\ref{Eqn:ProofOfUpperBound}), then we would 
have proved Theorem \ref{Thm:MainThmCharOfDKTheta}. This is our approach. 
Towards, this end we begin by identifying the dual of the LP in 
(\ref{Eqn:DPProblemPosedAsALP}).

Associated to each DP constraint (\ref{Eqn:GeneralDPProblemPosedAsALP}(b)), we 
have a 
non-negative dual variable
$\lambda_{\ulineg|(\ulineh,\ulinehath)}$. Note that $\lambda_{\ulineg|(\ulineh,
\ulinehath)}$ and $\lambda_{\ulineg|(\ulinehath,\ulineh)}$ are distinct dual
variables. Associated to each sum constraint
(\ref{Eqn:GeneralDPProblemPosedAsALP}) we have a free dual variable
$\mu_{\ulineh}$. It can be verified that the dual of
(\ref{Eqn:DPProblemPosedAsALP}) is \vspace{-0.1in}
\begin{eqnarray}
\label{Eqn:DualLPGeneral}
\lefteqn{\!S_{*}^{n}(\theta) \define \!\max\!\! \displaystyle \sum_{\ulineh
\in
\nhist}\!\!\!\mu_{\ulineh}  \mbox{ subject to (i) } \displaystyle\mu_{\ulineh}
\leq {n \choose \ulineh}\ulinep^{\ulineh}|\ulineh - \ulineg|_{1} + \theta\!\!\!
\sum_{\substack{\underline{\hat{h}} \in
\mathcal{N}(\ulineh)}}\!\!\!\!\lambda_{\ulineg|(\underline{\hat{h}},\ulineh)}
\!-\!\!\!\!\sum_{
\substack{\underline{\tilde{h}}\in
\mathcal{N}(\ulineh)}}\!\!\!\!\!\lambda_{\ulineg|(\ulineh,\underline{\tilde{h}})
}\mbox{ for }(\ulineh,\ulineg) \in}\nonumber\\
\!\!&\!\!\!&\!\!\!\!\nhist \times \nhist \mbox{ and  (ii)
}\lambda_{\ulineg|(\underline{\hat{h}},\ulineh)} \geq 0\mbox{ for }\ulineg \in
\nhist\mbox{ and }(\underline{\hat{h}},\ulineh) \in \nhist \times \nhist \mbox{
satisfying }|\ulineh - \underline{\hat{h}}|_{1}=2,~~~~~~
\end{eqnarray}
where
$\neighbor(\ulineh) \define \{ \ulinehath \in \nhist :
|\ulineh-\ulinehath|_{1}=2\}$ is the set of neighbors of $\ulineh \in \nhist$.
We let $C^{n}(\ulinelambda,\ulinemu) = \sum_{\ulineh \in \nhist}\mu_{\ulineh}$
denote the objective value corresponding to a feasible solution
$\ulinelambda,\ulinemu$, where $\ulinelambda$ and $\ulinemu$ represent the 
aggregate of $\lambda_{\ulineg|(\ulinehath,\ulineh)}$ and $\mu_{\ulineh}$ 
variables respectively.

The reader will note that each constraint in the primal LP 
(\ref{Eqn:GeneralDPProblemPosedAsALP}) has translated to a variable in the dual 
LP (\ref{Eqn:DualLPGeneral}) and vice versa. We therefore have at least 
$\mathcal{O}(k^{2}|\nhist|^{2}) = \mathcal{O}(k^{2}(n+1)^{2(k-1)})$ variables 
(Remark \ref{Rem:InvolvedLPProblem}). In order to describe the methodology 
behind the assignment of dual variables and the evaluation of its objective 
value, we first focus on the $K=2$ case. For this case, we provide a 
complete solution, i.e., identify a pair of primal and dual feasible 
solutions that satisfy complementary slackness conditions. This enables us to
glean the structure of an optimal dual feasible solution. We leverage this 
structure in providing an assignment for the general $K$ case. Specifically, we 
provide an interpretation of the dual feasible assignment via shadow prices 
(Appendix \ref{AppSec:ShadowPriceForDualVars}) which naturally leads us to the 
assignment for the general $K$ case.
\med
\textbf{The $K=2$ case}: We identify the histogram $(i,n-i) \in
\nhist_{2}$ with just its
first co-ordinate. We also let $\mathbb{W}(n-j|i)$ denote 
$\mathbb{W}((n-j,j)|(i,n-i))$, $\lambda_{j|(i-1,i)}$ denote 
$\lambda_{(j,n-j)|((i-1,n-i+1),(i,n-i))}$, and
so on. With this notational 
simplification, we state below the primal and dual LPs 
for $K=2$.\vspace{-0.08in}
\begin{eqnarray}
\begin{array}{lcclc}
&\mbox{Primal LP}&|&&\mbox{Dual LP}\\
\min\!\!\!\!\!\!\!\!\!&\displaystyle\sum_{i=0}^{n}\sum_{j=0}^{n}\mathscr{
C}_{i}^{n}\mathbb { W } (j|i)2|j-i|&|&\displaystyle\max
&\displaystyle\sum_{i=0}^{n}\mu_{i}\\
\mbox{subject to}\!\!\!\!\!\!\!\!\!&\mathbb{W}(j|i) \geq 0,\mbox{ for all }0\leq
i,j\leq
n&|&\mbox{subject to}&\displaystyle\mu_{i} \leq \mathscr{
C}_{i}^{n}2|j-i|\\
\!\!\!\!\!\!\!\!\!&&|&&+\theta\lambda_{j|(i-1,i)}+\theta\lambda_{j|(i+1,i)}\\
\!\!\!\!\!\!\!\!\!&&|&&-\lambda_{j|(i,i-1)} -\lambda_{ j|(i,i+1)}\\
\!\!\!\!\!\!\!\!\!&\sum_{j=0}^{n}\mathbb{W}(j|i) = 1\mbox{ for all
}0\leq i \leq
n&|&&\mu_{i}\mbox{ is free},\\
\!\!\!\!\!\!\!\!\!&\mathbb{W}(j|i-1)-\theta\mathbb{W}(j|i) \geq 0\mbox{ for all
}i,j&|&&\displaystyle\lambda_{j|(i-1,i)} \geq 0,\mbox{ for every }i,j\\
\!\!\!\!\!\!\!\!\!&\mathbb{W}(j|i+1)-\theta\mathbb{W}(j|i) \geq 0\mbox{ for all
}i,j,&|&&\lambda_{j|(i+1,i)}\geq 0\mbox{ for every }i,j,
\end{array}\label{Eqn:PrimalAndDualLPK=2}
\end{eqnarray}
where $\binpmf= {n \choose i}p_{1}^{i}(1-p_{1})^{n-i}$. We have suppressed 
dependence of $\binpmf$ on $p_{1}$. Furthermore, we let $p=p_{1}$ and 
$p_{2}=1-p$. We provide a 
complete solution, i.e., primal and dual feasible solutions that satisfy 
complementary slackness conditions. Recall that from complementary slackness, 
we are 
required to prove that (i) either
the primal constraint is tight or the corresponding dual variable is $0$, and 
(ii) either the primal variable is $0$ or the dual constraint is tight. For
ease 
of verification, we have stated
variables and constraints that are duals of each other on the same row of 
(\ref{Eqn:PrimalAndDualLPK=2}).

Let us begin with a primal feasible solution. Let $f_{i} = \sum_{j=0}^{i} 
2\mathscr{C}_{j}^{n}\theta^{i-j}, b_{i} =
\sum_{k=i}^{n}2\mathscr{C}_{k}^{n}\theta^{k-i}$, and\footnote{We assume,
without loss of generality that $p_{1}\leq \frac{1}{2}$}\vspace{-0.1in}
\begin{eqnarray}
 \label{Eqn:AssigmentOfDualVariableAnBn}
A_{n}\define \min\left\{ 
i \in [0,n]:\begin{array}{c}f_{k-1}-\theta b_{k}\geq 0\\\mbox{for every }k 
\geq i \end{array}\right\}, 
 B_{n}\define \max\left\{ 
i \in [0,n]:\begin{array}{c}b_{k+1}-\theta f_{k}\geq 0\\\mbox{ for every }k 
\leq i \end{array}\right\}.
\end{eqnarray}
$A_{n}-1$ and $B_{n}+1$ will represent the left and right ends of a truncated 
geometric mechanism which we prove is optimal. In Appendix 
\ref{AppSec:AnBnInDualVarAssignment}, we prove that $A_{n} < np_{1} 
< B_{n}$. We use the same in the following assignment. Consider the truncated 
geometric mechanisms that are folded at 
$A_{n}-1$ on the left and $B_{n}+1$ on the right. Specifically, let 
$\mathbb{U}^{n}:\nhist_{2} \Rightarrow \nhistext$ and $\mathbb{V}^{n} : 
\nhistext \Rightarrow \nhist_{2}$, where $\nhistext \define \{ (i,n-i): i \in 
\mathbb{Z} \}$. As stated earlier, we refer to $(i,n-i) \in \nhistext$ by its 
first co-ordinate $i$. Let
\begin{eqnarray}
 \label{Eqn:DefnVnForKIs2}
 \mathbb{U}^{n}(k|i) = 
\theta^{|k-i|}\frac{1-\theta}{1+\theta}\mbox{ for }k \in \mathbb{Z}, i \in 
[0,n] ,~~ \mathbb{V}^{n}(j|i) = \left\{ \begin{array}{ll}1&\mbox{if 
}j=i,j \in 
[A_{n}-1,B_{n}+1]\\
1& \mbox{if }i \leq A_{n}-1, j = A_{n}-1 \\
1 & \mbox{if }i \geq B_{n}+1 , j = B_{n}+1\\
0& \mbox{otherwise,}
\end{array}   
\right.
\end{eqnarray}
and $\mathbb{W}^{n}(j|i) = \sum_{k \in \mathbb{Z}} 
\mathbb{U}^{n}(k|i)\mathbb{V}^{n}(j|k)$. It can be verified that 
\begin{eqnarray}
 \label{Eqn:TruncatedGeometric}
\mathbb{W}^{n}(j|i) = \left\{  \!\!\!
\begin{array}{ll}
\theta^{|j-i|}\frac{1-\theta}{1+\theta} &i \in [ 0,n], j \in [A_{n},B_{n}]\\
\frac{\theta^{|j-i|}}{1+\theta} &j=B_{n}+1,i \leq j\\&\mbox{or }j=A_{n}-1, i 
\geq j\\
0&j\notin [A_{n}-1,B_{n}+1],
\end{array}\right.\!\!
 \mathbb{W}(j|i) = \left\{ \!\!
 \begin{array}{ll}
1-\frac{\theta^{A_{n}-i}}{1+\theta}& i < A_{n}-1, j = A_{n}-1\\
1-\frac{\theta^{i-B_{n}}}{1+\theta}& i > B_{n}+1, j = B_{n}+1\\
0&\mbox{otherwise.}
 \end{array}
 \right.
 \end{eqnarray}
It can be easily verified that the above assignment satisfies the constraints in 
(\ref{Eqn:GeneralDPProblemPosedAsALP}). This can be done in 
either of two ways. The 
first is just by the fact that $\mathbb{U}^{n}$ being $\theta-$DP implies 
$\mathbb{W}^{n}$ is $\theta-$DP. The second is by verifying 
that $\mathbb{W}^{n}$ as 
assigned in (\ref{Eqn:TruncatedGeometric}) satisfies 
(\ref{Eqn:GeneralDPProblemPosedAsALP}a) and 
(\ref{Eqn:GeneralDPProblemPosedAsALP}b). We leave this to the reader.

What 
are the complementary slackness conditions with regard to the 
above primal feasible assignment? We make the following observations with 
regard to the above assignment. Firstly,
\begin{eqnarray}
 \label{Eqn:CompSlackObser1}
 \mathbb{W}^{n}(j|i-1)-\theta \mathbb{W}^{n}(j|i) > 0\mbox{ if } j \leq 
i-1\mbox{ and }\mathbb{W}^{n}(j|i+1)-\theta\mathbb{W}^{n}(j|i) > 0 
\mbox{ if }j \geq i+1.
\end{eqnarray}
Secondly,
\begin{eqnarray}
 \label{Eqn:ConstraintsLooseAtA-1AndB+1}
 \theta < \frac{\mathbb{W}(A_{n}-1|i)}{\mathbb{W}(A_{n}-1|i-1)} < 
\frac{1}{\theta}\mbox{
if }i\leq A_{n}\mbox{ and similarly }\theta <
\frac{\mathbb{W}(B_{n}+1|i+1)}{\mathbb{W}(B_{n}+1|i)} < 
\frac{1}{\theta}\mbox{ 
if
}i\geq B_{n}+1.\nonumber
\end{eqnarray}
Moreover, for $j \in [A_{n}-1,B_{n}+1]$, we have $\mathbb{W}^{n}(j|i) > 0$ and 
hence the corresponding constraints have to be met with equality in the dual 
LP. Specifically, our dual feasible assignment must satisfy
\begin{eqnarray}
 \label{Eqn:LowerBoundForK=2-1}
 \mu_{i}= 2\mathscr{C}_{i}^{n}|j-i| + 
\theta\lambda_{j|(i-1,i)}+\theta\lambda_{j|(i+1,i)}-\lambda_{j|(i,i-1)} 
-\lambda_{ j|(i,i+1)}\mbox{ for }j \in [A_{n}-1,B_{n}+1].
\end{eqnarray}
We now provide a feasible assignment for the dual variables. Let 
$\lambda_{A_{n}-1|(i-1,i)} = 0 $ for $i 
\leq
A_{n}-1$ and $\lambda_{B_{n}+1|(i+1,i)} = 0 $ for $i \geq 
B_{n}+1$. Let $\lambda_{j|(i-1,i)} = 0$ if $j \leq i-1$ and 
$\lambda_{j|(i+1,i)}=0$ if $j \geq i+1$.\footnote{For the 
general
$K$, we will assign $\lambda_{\ulineg | (\ulinehath,\ulineh)} = 0$ if $|\ulineg
- \ulinehath|_{1} \leq |\ulineg - \ulineh|_{1}$. Note that this simple
observation halves the number of decision variables.} With this, the reader can 
verify that we have handled the last three rows of 
(\ref{Eqn:PrimalAndDualLPK=2}). We are only left to provide an assignment for 
the rest of the dual variables that satisfy (\ref{Eqn:LowerBoundForK=2-1}). For 
$i \in \{
1,\cdots,A_{n}-1\}$
and $j \in \{i,\cdots,A_{n}-1\}$, set $\lambda_{j|(i-1,i)} \define 0$. 
\begin{eqnarray}
\label{Eqn:Assignment1}
&\mbox{For $i \in \{ 1,\cdots,A_{n}-1\}$ and $j \in \{A_{n}-1,\cdots,n\}$, set
}\lambda_{j|(i-1,i)} \define [j-(A_{n}-1)]f_{i-1}.\\
\label{Eqn:Assignment2}
&\mbox{For $i \in [A_{n},n]$ and $j \in [i,n]$, set }\lambda_{j|(i-1,i)} \define
\frac{f_{i-1}-\theta b_{i}}{1-\theta^{2}}+(j-i)f_{i-1}.\\
\label{Eqn:Assignment3}
&\mbox{For $i \in [B_{n}+1,n-1]$ and $j \in [B_{n}+1,i]$, set 
}\lambda_{j|(i+1,i)}
\define 0.\nonumber\\
\label{Eqn:Assignment4}
&\mbox{For $i \in [B_{n}+1,n-1]$ and $j \in [0,B_{n}+1]$, set 
}\lambda_{j|(i+1,i)}
\define [(B_{n}+1)-j]b_{i+1}.\nonumber\\
\label{Eqn:Assignment5}
&\mbox{For $i \in [0,B_{n}]$ and $j \in [0,i]$, set }\lambda_{j|(i+1,i)} \define
\frac{b_{i+1}-\theta f_{i}}{1-\theta^{2}}+(i-j)b_{i+1}.
\mbox{ For }i < A_{n}-1, \mbox{ set}\\
\label{Eqn:MuiOutsideCriticalWindow}
&  \mu_{i} = 2\mathscr{C}^{n}_{i}|(A_{n}-1)-i|,\mbox{ and }i > 
B_{n}+1, \mbox{
set }\mu_{i} = 2\mathscr{C}^{n}_{i}|i-(B_{n}+1)|\\
\label{Eqn:MuiInsideCriticalWindow}
&\mbox{For }i \in [A_{n}-1,B_{n}+1], \mbox{ set }\mu_{i} \define
f_{i}+b_{i}-\frac{4}{(1-\theta^{2})}\mathscr{C}^{n}_{i}&\\&\mbox{For }i \in 
[A_{n}-1,B_{n}+1]\mbox{ verify 
}\mu_{i}=\theta(f_
{ i-1 } +b_{ i+1} )-\frac { 4\theta^{ 2}}{ (1-\theta^ { 2 } ) } 
\mathscr{C}^{n}_{i}.
\end{eqnarray}
The above assignment is indeed non-trivial. We refer the reader to Appendix 
\ref{AppSec:ShadowPriceForDualVars} for an interpretation of 
the above assignment via shadow prices. This interpretation will prove very 
valuable in arriving at the dual variable
assignment for the general $K$ case in (\ref{Eqn:DualVarAssignment}). We will
now use the above assignment to verify (\ref{Eqn:LowerBoundForK=2-1}).

Recall $\mathscr{C}_{i}^{n} = {n \choose i}p^{i}(1-p)^{n-i}$. We first prove 
that for any $i < A_{n}-1$, $j \in [A_{n}-1,B_{n}+1]$,
\begin{eqnarray}
 \label{Eqn:BoundOnMuiForinvalidVars}
 \begin{array}{c}
{n \choose
i}p^{i}(1-p)^{n-i}2|j-i|+\theta\lambda_{j|(i-1,i)}\\+\theta\lambda_{j|(i+1,i)}
-\lambda_{j|(i,i+1)}-\lambda_{j|(i,i-1)}  
 \end{array}
 = {n \choose
i}p^{i}(1-p)^{n-i}2|(A_{n}-1)-i|.
\end{eqnarray}
Towards that end, note that $\lambda_{j|(i+1,i)}=\lambda_{j|(i,i-1)}=0$ for the
considered values for $i,j$. Substituting 
$\lambda_{j|(i-1,i)}=[j-(A_{n}-1)]f_{i-1}$
from (\ref{Eqn:Assignment1}), we have
$\theta\lambda_{j|(i-1,i)}-\lambda_{j|(i,i+1)}=[j-(A_{n}-1)](\theta 
f_{i-1}-f_{i}) =
-[j-(A_{n}-1)]{n \choose i}2p^{i}(1-p)^{n-i}$, and we therefore have
(\ref{Eqn:BoundOnMuiForinvalidVars}). From the assignment 
(\ref{Eqn:Assignment3}), (\ref{Eqn:MuiOutsideCriticalWindow}), we conclude 
validity of (\ref{Eqn:LowerBoundForK=2-1}) for $i < A_{n}-1$. Before we 
continue, we note that
\begin{eqnarray}
 \label{Eqn:BasicIdentities}
 f_{i} = \theta f_{i-1}+{n \choose i}2p^{i}(1-p)^{n-i}, ~\mbox{ and }~ b_{i} =
\theta b_{i+1} + {n \choose i}2p^{i}(1-p)^{n-i}.
\end{eqnarray}
We now consider upper bounds on $\mu_{i}$ for the range $i \in [A-1,B+1],j \in 
[i+1,n]$. Substituting
(\ref{Eqn:Assignment2}), (\ref{Eqn:Assignment5}) and using
(\ref{Eqn:BasicIdentities}), one can verify that
\begin{eqnarray}
 \label{Eqn:ValidMusLaterBounds}
 \begin{array}{c}
{n \choose i}p^{i}(1-p)^{n-i}2|j-i| \\\\+ \theta
\lambda_{j|(i-1,i)}-\lambda_{j|(i,i+1)}  
 \end{array}
 &=& {n \choose i}p^{i}(1-p)^{n-i}2|j-i|+\frac{\theta
f_{i-1}-f_{i}-\theta^{2}b_{i}+\theta b_{i+1}}{1-\theta^{2}}
\nonumber\\
 &&~~~+(j-i)(\theta f_{i-1}-f_{i})+f_{i}\nonumber\\
 \label{Eqn:UpperBoundWhenjIsGreaterThani}
 &=& f_{i}+b_{i}-\frac{4}{(1-\theta^{2})}{n \choose i}p^{i}(1-p)^{n-i}.
\end{eqnarray}
Similarly, for $i\in [A-1,B+1], j \in [0,i-1]$, one can substitute
(\ref{Eqn:Assignment2}), (\ref{Eqn:Assignment5}) and use
(\ref{Eqn:BasicIdentities}) to establish
\begin{eqnarray}
 \label{Eqn:ValidMusEarlierBounds}
 \begin{array}{c}
{n \choose i}p^{i}(1-p)^{n-i}2|j-i| \\\\+ \theta
\lambda_{j|(i+1,i)}-\lambda_{j|(i,i-1)}   
 \end{array}
&=& {n \choose i}p^{i}(1-p)^{n-i}2|j-i|
+\frac{\theta b_{i+1}-\theta^{2}f_{i}-b_{i}+\theta f_{i-1}}{1-\theta^{2}}
\nonumber\\
 &&~~~+(i-j)(\theta b_{i+1}-b_{i})+b_{i}\nonumber\\
 \label{Eqn:UpperBoundWhenjIsLessThani}
 &=& f_{i}+b_{i}-\frac{4}{(1-\theta^{2})}{n \choose i}p^{i}(1-p)^{n-i}.
\end{eqnarray}
Suppose $i \in [A-1,B+1]$ and $j = i$; the upper bound on $\mu_{i}$ is
\begin{eqnarray}
 \label{Eqn:BoundOnMuiWithjEquali}
 \theta\lambda_{i|(i-1,i)}+\theta\lambda_{i|(i+1,i)} &=& \frac{\theta
f_{i-1}-\theta^{2} b_{i}+\theta b_{i+1}-\theta^{2} f_{i}}
 {1-\theta^{2}} \nonumber\\&=& \theta f_{i-1} + \theta b_{i+1}
-\frac{4\theta^{2}}{(1-\theta^{2})}{n \choose i}p^{i}(1-p)^{n-i}.
\end{eqnarray}
The expressions in (\ref{Eqn:UpperBoundWhenjIsGreaterThani}),
(\ref{Eqn:UpperBoundWhenjIsLessThani}) and (\ref{Eqn:BoundOnMuiWithjEquali})
being equal to the assignment
(\ref{Eqn:MuiInsideCriticalWindow}) for $\mu_{i}$ in the range $i \in [A-1,
B+1]$, we conclude validity of (\ref{Eqn:LowerBoundForK=2-1}). We are left to 
prove validity of (\ref{Eqn:LowerBoundForK=2-1}) for $i \geq B_{n}+1$. This is 
similar to (\ref{Eqn:BoundOnMuiForinvalidVars}). Substituting 
(\ref{Eqn:MuiOutsideCriticalWindow}), one can verify that
\begin{eqnarray}
\label{Eqn:BoundOnMuiForinvalidVars2}
  \begin{array}{c}
{n \choose
i}p^{i}(1-p)^{n-i}2|j-i|+\theta\lambda_{j|(i-1,i)}\\+\theta\lambda_{j|(i+1,i)}
-\lambda_{j|(i,i+1)}-\lambda_{j|(i,i-1)}  
 \end{array}
 = {n \choose
i}p^{i}(1-p)^{n-i}2|i-(B_{n}+1)|.
\end{eqnarray}
From the assignment for $\mu_{i}$ in (\ref{Eqn:MuiOutsideCriticalWindow}) for 
$i > B_{n}+1$, we have the validity of (\ref{Eqn:LowerBoundForK=2-1}) for $i > 
B_{n}+1$. We have thus proved the validity of (\ref{Eqn:LowerBoundForK=2-1}) 
for 
all values of $i$ and $j \in [A_{n}-1,B_{n}+1]$. The non-negativity of 
$\lambda_{j|(i-1,i)}$ and $\lambda_{j|(i+1,i)}$ follows from (i) definition of 
$A_{n},B_{n}$, and (ii) non-negativity of $f_{i},b_{i}$. We have thus proved 
that the above assignments are valid primal and feasible assignments 
and satisfy complementary slackness conditions. We only need to evaluate the 
objective of one of these values and prove that it tends to 
$\frac{4\theta}{1-\theta^{2}}$ in the limit $n \rightarrow \infty$.

It is easier to evaluate the objective value of the above feasible dual 
assignment. Substituting
(\ref{Eqn:MuiInsideCriticalWindow}), (\ref{Eqn:MuiOutsideCriticalWindow}), we
have \begin{eqnarray}
 \label{Eqn:ObjectiveValueForK=2}
 \lefteqn{C^{n}(\ulinelambda,\ulinemu) =\sum_{i=0}^{n}\mu_{i} = \sum_{i
=A_{n}-1}^{B_{n}+1}\!\!\!(f_{i}+b_{i}) + \sum_{i <
A_{n}-1}\!\!{n \choose i}p^{i}(1-p)^{n-i}2|A_{n}-1-i|} 
\nonumber\\&&~~~~~~~~~~~~~~~~~~~~+ \sum_{i >
B_{n}+1}\!\!{n \choose
i}p^{i}(1-p)^{n-i}2|i-B_{n}-1|-\frac{4}{(1-\theta^{2})}\sum_{A_{n}-1}^{B_{n}+1}
\mathscr{C}_{i}^{n} \nonumber\\
& \geq & \sum_{i=0}^{n}(f_{i}+b_{i}) - \frac{4}{(1-\theta^{2})} - 
\sum_{i=0}^{A_{n}-2}(f_{i}+b_{i}) - \sum_{i=B_{n}+2}^{n}(f_{i}+b_{i}).
\nonumber\end{eqnarray}
We focus on the first term above:
\begin{eqnarray}
\sum_{i=0}^{n}(f_{i}+b_{i}) &=& 2\sum_{i=0}^{n}\sum_{j=0}^{i} {n \choose 
j}p^{j}(1-p)^{n-j}\theta^{i-j} + 2\sum_{i=0}^{n}\sum_{k=i}^{n} {n \choose 
k}p^{k}(1-p)^{n-k}\theta^{k-i} \nonumber\\
&=&2\sum_{j=0}^{n}\sum_{i=j}^{n}{n \choose j}p^{j}(1-p)^{n-j}\theta^{i-j}+ 
2\sum_{k=0}^{n}\sum_{i=0}^{k} {n \choose k}p^{k}(1-p)^{n-k}\theta^{k-i} 
\nonumber\\
&=&2\sum_{j=0}^{n}{n \choose 
j}p^{j}(1-p)^{n-j}\frac{1-\theta^{n-j+1}}{1-\theta}+ 
2\sum_{k=0}^{n} {n \choose 
k}p^{k}(1-p)^{n-k}\frac{1-\theta^{k+1}}{1-\theta}\nonumber\\
&=& \frac{4}{1-\theta}- \frac{2\theta^{n+1}}{1-\theta}\mathbb{E}\{ 
\theta^{-X_{n}}\} - \frac{2\theta}{1-\theta}\mathbb{E}\{ 
\theta^{X_{n}}\}, \nonumber
\end{eqnarray}
where $X_{n}$ is a Bernoulli RV with parameters $n,p$. Since 
$\mathbb{E}\left\{ \theta^{X_{n}}\right\} \overset{\sim}{=} (p\theta 
+(1-p))^{n} \underset{n \rightarrow \infty}{\rightarrow} 0 $, 
we\footnote{Recall that $\theta \in (0,1)$.} have $\lim_{n 
\rightarrow \infty} \sum_{i=0}^{n}(f_{i}+b_{i}) = \frac{4}{1-\theta}$. We 
therefore have
\begin{eqnarray}
 \lim_{n \rightarrow \infty} C^{n}(\ulinelambda,\ulinemu) \geq 
\frac{4\theta}{1-\theta^{2}} - \lim_{n \rightarrow \infty}
\sum_{i=0}^{A_{n}-2}(f_{i}+b_{i}) - \lim_{n \rightarrow \infty} 
\sum_{i=B_{n}+2}^{n}(f_{i}+b_{i}) = \frac{4\theta}{1-\theta^{2}}.
\label{Eqn:FinalK=2LowerBound}
\end{eqnarray}
In arriving at (\ref{Eqn:FinalK=2LowerBound}), we have used $np-A_{n} = B_{n}-np 
= O(\sqrt{n})$ and standard results in concentration of binomial 
probabilities. This
concludes the proof for the case $K=2$. A step-by-step proof of 
(\ref{Eqn:FinalK=2LowerBound}) is provided in \cite{201803arXiv_PadKumSzp}.

We now leverage the shadow price interpretation provided in Appendix 
\ref{AppSec:ShadowPriceForDualVars} to provide an assignment for general $K$. 
The proof of
feasibility of the following assignment follows from arguments analogous to
those presented in Eqns. (\ref{Eqn:BoundOnMuiForinvalidVars}) -
(\ref{Eqn:BoundOnMuiWithjEquali}) for the $K=2$ case.

Refer to Appendix \ref{AppSec:L1DistanceAndGraphDistance} for definition of 
the PC 
graph $G$ and its properties. For $\ulinea \in 
\nhist$, let
$\mathcal{N}(\ulinea) \define \{ \underline{\hat{a}} \in \nhist:
|\ulinea - \underline{\hat{a}}|_{1}=2\}$ be the set of neighbors of $\ulinea$.
For $\ulinea, \ulineb \in \nhist$, let
\begin{eqnarray}
 \label{Eqn:FarCloseAndEqualDistanceSets}
 &\mathcal{F}(\ulineb,\ulinea) \define \left\{ \underline{\tilde{a}} \in
\mathcal{N}(\ulinea) : |\ulineb -\underline{\tilde{a}}|_{1} > |\ulineb
-\underline{{a}}|_{1} \right\}, \mathcal{C}(\ulineb,\ulinea) \define
\left\{ \underline{\tilde{a}} \in
\mathcal{N}(\ulinea) : |\ulineb -\underline{\tilde{a}}|_{1} < |\ulineb
-\underline{{a}}|_{1} \right\}&\nonumber\\
 &\mbox{and }\mathcal{E}(\ulineb,\ulinea) \define \left\{ \underline{\tilde{a}}
\in
\mathcal{N}(\ulinea) : |\ulineb -\underline{\tilde{a}}|_{1} = |\ulineb
-\underline{{a}}|_{1} \right\}&\nonumber
\end{eqnarray}
be the set of histograms farther to, closer to, and at equidistant from
$\ulineb$ than $\ulinea$ respectively. Recall that $2d_{G}(\ulinea,\ulineb) =
|\ulinea-\ulineb|_{1}$ (Lemma \ref{Lem:PropertiesPCGraph}). Complementary
slackness conditions imply
\begin{eqnarray}
 \label{Eqn:DualVarAssComplementarySlackness}
 \lambda_{\ulineg|(\ulineh,\underline{\hat{h}})} = 0
\mbox{ whenever }|\ulineg - \underline{\hat{h}}|_{1} > |\ulineg-\ulineh|_{1}.
\end{eqnarray}
When $|\ulineg - \underline{\hat{h}}|_{1} < |\ulineg-\ulineh|_{1}$, let
\begin{eqnarray}
 \label{Eqn:DualVarAssignment}
 \displaystyle\lambda_{\ulineg|(\ulineh,\underline{\hat{h}})} =
\frac{ \displaystyle\sum_{\ulinea \in \mathcal{C}(\ulineh,\ulinehath)}\!\!\!{n
\choose \ulinea} \ulinep^{\ulinea}~ 2~\theta^{ d_{G}(\ulinea,\ulineh)
}-\theta\sum_{\ulineb \in \mathcal{C}(\ulinehath,\ulineh)}\!\!\!{n \choose
\ulineb} \ulinep^{\ulineb}~2~\theta^{d_{G}(\ulinehath,\ulineb)}
}{1+|\mathcal{C}(\ulineh,\ulinehath)|\theta^{ 2 } -(K(K-1))\theta^ { 2 }
+\theta|\mathcal{E}(\ulineh,\ulinehath)| } + |\ulineg -
\ulineh|_{1}\!\!\!\!\!\displaystyle\sum_{\ulinea \in
\mathcal{C}(\ulineh,\ulinehath)}\!\!\!\!{n
\choose \ulinea} \ulinep^{\ulinea}~\theta^{ d_{G}(\ulinea,\ulineh)
}\\
\label{Eqn:MuAssignment}
\mu_{\ulineg} = \frac{\displaystyle
\theta\sum_{\ulineh \in \mathcal{N}(\ulineg)}\left\{  \displaystyle\sum_{\ulinea
\in \mathcal{C}(\ulineh,\ulineg)}\!\!\!{n
\choose \ulinea} \ulinep^{\ulinea}~ 2~\theta^{ d_{G}(\ulinea,\ulineh)
}-\theta\sum_{\ulineb \in \mathcal{C}(\ulineg,\ulineh)}\!\!\!{n \choose
\ulineb} \ulinep^{\ulineb}~2~\theta^{d_{G}(\ulineh,\ulineb)}
\right\}}{1+|\mathcal{C}(\ulineh,\ulinehath)|\theta^{ 2 } -(K(K-1))\theta^ { 2 }
+\theta|\mathcal{E}(\ulineh,\ulinehath)| }.~~~~~~~~~~~~~~~~
\end{eqnarray}
Having provided the above assignments, the natural question that arises is
whether these are feasible for (\ref{Eqn:DualLPGeneral}), and if yes, what do
they evaluate to? A couple of remarks are in order. The first term in
(\ref{Eqn:DualVarAssignment}) is negative if $\ulineg=\ulinehath$,
$|\ulinehath - n\ulinep|>|\ulineh - n\ulinep|+2$ and $|\ulinehath - n\ulinep| >
\Theta(\sqrt{n})$. This is the case analogous to (\ref{Eqn:Assignment1}).
Therein, note that when $i \in [A,B]$, the assignment is
(\ref{Eqn:Assignment2}). In fact, the fraction in (\ref{Eqn:DualVarAssignment})
is analogous to the fraction in (\ref{Eqn:Assignment2}). The reader will
recognize $\mathcal{E}(\ulineh,\ulinehath) = 0$ and
$\mathcal{C}(\ulineh,\ulinehath)=\mathcal{F}(\ulineh,\ulinehath)=1$. The first
term in the numerator of the fraction in (\ref{Eqn:DualVarAssignment}) is
analogous to $f_{i-1}$ in (\ref{Eqn:Assignment2}). The rest of the terms can
also be related to the assignment in (\ref{Eqn:Assignment1}) -
(\ref{Eqn:MuiInsideCriticalWindow}). The above assignment is a slightly
simplified version, in the sense that the variables corresponding to non-active
constraints have been ignored. Appendix \ref{AppSec:ShadowPriceForDualVars} 
provides a clear interpretation for the above assignment for $K=2$. An 
analogous argument to our thorough description for the $K=2$ case, its
feasibility and the evaluation of its objective value completes the proof.

\section{Concluding Remarks}
\label{Sec:ConcludingRemarks}
Our work is aimed at initiating a systematic information theoretic study of the 
fundamental trade-off between the utility lost and the privacy preserved in any 
data obfuscation mechanism. It is addressed in the information theoretic spirit 
by characterizing the expected fidelity in the asymptotic regime of large
databases. In this work, we have adopted DP 
as the framework to quantify the vulnerability of the obfuscated data to 
privacy breaches. Our measure of utility - the $\mathbb{L}_{1}-$distance 
measure between the histograms - is simple and yet provides us with an ideal 
setting to put forth the connections between DP, Ehrhart 
theory, analytic combinatorics and linear programming.

Going further, one may 
ask questions at two different levels. At a technical level, it would be 
interesting to build on the following questions and provide suitable answers. 
Can one derive simple closed form computable expressions characterizing the 
utility-privacy trade-off for other pertinent distortion measures? 
What would be the optimal sanitizing mechanisms? We conjecture that such a 
study will involve enumerating integer points on the intersection of convex 
polytopes.

At a more strategic level, we deem it necessary to ask the following question. 
Given that we require certain utility and accuracy from our data mining 
algorithms, can we provide the stringent guarantees sought by DP for 
sanitizing databases or responding to individual queries? Our work proves that
the minimal distortion (\ref{Eqn:DK*OfTheta}) scales linearly with the 
dimensionality of the database, even if the number of records grows unbounded. 
Given the fine-grained and 
high-dimensional nature of our databases, is this adequate? Why are we not 
able to exploit the presence of a large number of records in our sanitization? 
The answer lies in the fact that DP is attempting to be 
robust against an adversary that knows $n-1$ records perfectly. As the number 
of records grow, the \textit{fraction} of entries that the adversary knows 
\textit{increases} to $1$. Indeed, this is a conservative model. Since the 
adversary's `power' is increasing with the size of the DB, a DP mechanism is 
unable to exploit the presence of a large number 
of records to `minimize the necessary randomization'. In other words, it is 
unable to hide one subject's record in the pool of all records without the 
help of randomization. We therefore conclude by asking the questions: Is 
a very low utility 
the inevitable price to pay for provable guarantees on privacy for large 
databases that DP promises? or, can we 
provide a more realistic framework to quantify privacy and vulnerability of 
query-response mechanisms?\vspace{-0.05in}
\appendices
\section{Summary of Notation}
\label{AppSec:Notation}
\begin{table}[h]
\begin{center}
\begin{tabular}{|c|c|}
\hline
Symbol & Meaning \\
\hline \hline
$\mathbb{Z},\mathbb{N},\mathbb{R}$ & Sets of integers, natural and real 
numbers\\
\hline
$[a,b]$ & For $a,b \in
\integers$, we let $[a,b] \define \{a,a+1,\cdots,b\}$\\
\hline
$[n]$& For $n \in \naturals$, we let $[n] = [1,n]$.
\\\hline
$M:\mathcal{A} \Rightarrow \mathcal{B}$ & A randomized algorithm, referred to 
herein as a\\&mechanism, with set $\mathcal{A}$ of
inputs and set $\mathcal{B}$ of outputs. \\
\hline
$\mathbb{W}_{M}(b|a)$&Probability that mechanism $M$ 
produces\footnote{In the computer science literature, such as in
\cite{2006ICALP_Dwo, 2006CTOC_DwoMcsNisSmi}, $\mathbb{W}(b|a)$ is denoted
$M(b)_{a}$.}\\&output $b \in \mathcal{B}$ when input with $a \in 
\mathcal{A}$.\\
\hline
$\mathbb{W}_{M}:\mathcal{A}
\rightarrow \mathbb{P}(\mathcal{B})$ &Alternate notations for mechanism 
$M:\mathcal{A} \Rightarrow \mathcal{B}$.\\or $\mathbb{W}_{M}:\mathcal{A} 
\Rightarrow
\mathcal{B}$ &$\mathbb{P}(\mathcal{B})$ 
is the set
of probability distributions on $\mathcal{B}$ \\ 
\hline
$d_{G}(v_{1},v_{2})$ & Length of a shortest path from $v_{1}
\in V$ to $v_{2} \in V$\\
& in graph $G=(V,E)$ \\\hline
${n \choose \ulineh}$&  When $\sum_{k=1}^{K}h_{k}=n$, we let ${n \choose 
\ulineh} = {n \choose h_{1}\cdots
h_{K}}$.\\&\\\hline
Uppercase&Random variables and (generic) parameters\\letters&that remain 
fixed throughout.\\
\hline
Calligraphic& Represent finite sets\\letters & Examples :
$\mathcal{A},
\mathcal{R}$
\\\hline
\end{tabular}
\end{center}
\caption{Description of symbols used in the article}\vspace{-0.05in}
\label{Table:Notation}
\end{table}\vspace{-0.25in}
\section{It suffices to focus on mechanisms that are a function only of the
histogram of the database}
\label{AppSec:RestrictingSearchForMechanisms}
In our search for an optimal database sanitizing mechanism, we prove here that 
we may
restrict attention to mechanisms that satisfy $\mathbb{W}(\ulinea|\ulineb) =
\mathbb{W}(\ulinea|\underline{\tilde{b}})$ whenever $\histof(\ulineb) =
\histof(\underline{\tilde{b}})$.

\begin{lemma}
 \label{Lem:RestrictAttentionToHistogramMechanisms}
 Given a privacy constraint $\theta > 0$, there exists a mechanism
$(\mathbb{W}(\cdot|\ulinea):\ulinea \in \databasespace)$ such that (i)
$\mathbb{W}(\cdot|\ulinea)=\mathbb{W}(\cdot|\underline{\tilde{a}})$ whenever
$\histof(\ulinea) = \histof(\underline{\tilde{a}})$, (ii)
$\frac{\mathbb{W}(\ulineb|\ulinea)}{\mathbb{W}(\ulineb|\underline{\tilde{a}})}
\in [\theta, \frac{1}{\theta}]$ for every pair $\ulinea, \underline{\hat{a}}$ of
neighboring databases and every database $\ulineb$, and (iii) $D^{n}(\mathbb{W}) 
\leq
D^{n}(\mathbb{U})$ for every sanitizing mechanism $\mathbb{U}$ that is
$\theta-$DP.
 \end{lemma}

 \begin{proof}
  We prove the following statement. Given any $\theta-$DP database sanitizing
mechanism $(\mathbb{U}(|\cdot|\ulinea):\ulinea \in \databasespace)$, there
exists a $\theta-$DP sanitizing mechanism $(\mathbb{W}(\cdot|\ulinea):\ulinea
\in \mathcal{R}^{n})$ that satisfies (i) and (ii) in the theorem statement and
$D^{n}(\mathbb{W}) \leq D^{n}(\mathbb{U})$. Towards that end, define
  \begin{equation}
   \underline{c}^{*}_{\ulineg} \in \arg_{\ulined : h(\ulined)= \ulineg} \min
\sum_{\ulineb \in
\databasespace}\mathcal{F}(\histof(\ulineb),\histof(\ulined))\mathbb{U}
(\ulineb|\ulined) \mbox{ and let } \nonumber
   \mathbb{W}(\ulinea|\ulineb) = \mathbb{U}(\ulinea | \ulinec_{h(\ulineb)}^{*})
\mbox{ for all }\ulinea \in \databasespace, \ulineb \in \databasespace.
\nonumber
\end{equation}
Suppose $h(\ulineb) = h(\underline{\tilde{b}})$, then
$\mathbb{W}(\ulinea|\ulineb) = \mathbb{U}(\ulinea | \ulinec^{*}_{h(\ulineb)})=
\mathbb{U}(\ulinea|\ulinec^{*}_{h(\underline{\tilde{b}})})=\mathbb{W}(\underline
{a}|\underline{\tilde{b}})$. Suppose $\underline{b}$ and $\underline{\hat{b}}$
are neighboring databases, then $|h(\underline{b})-h(\underline{\hat{b}})|=2$.
Since $c^{*}_{h(\ulineb)}$ and $c^{*}_{h(\underline{\hat{b}})}$ are neighboring
and $\mathbb{U}$ is $\theta-$DP, we have 
\begin{eqnarray}
 \label{Eqn:DPConstraints}
 \frac{\mathbb{W}(\ulinea|\ulineb)}{\mathbb{W}(\underline{a}|\underline{\hat{b}})} = \frac{\mathbb{U}(\ulinea | c^{*}_{h(\underline{b})})}{\mathbb{U}(\ulinea | c^{*}_{h(\underline{\hat{b}})})} \in [\theta, \frac{1}{\theta}] \mbox{ for all } \underline{a} \in \mathcal{R}^{n}. \nonumber
\end{eqnarray}
Lastly, we study $D^{n}({\mathbb{W}})$:
\begin{eqnarray}
 D^{n}(\mathbb{W}) &=& \sum_{\ulinea \in \databasespace} \sum_{\ulineb \in
\databasespace} p(\ulinea)\mathbb{W}(\ulineb |
\ulinea)\mathcal{F}(\histof(\ulinea),\histof(\ulineb)) = \sum_{\ulineg \in
\nhist}\sum_{\substack{\ulinea \in \databasespace:\\\histof(\ulinea) =
\ulineg}}p(\ulinea)\sum_{\ulineb \in \databasespace} \mathbb{W}(\ulineb |
\ulinea)\mathcal{F}(\histof(\ulinea),\histof(\ulineb))\nonumber\\
 &=& \sum_{\ulineg \in \nhist} \sum_{\substack{\ulinea \in
\databasespace:\\\histof(\ulinea) = \ulineg}}p(\ulinea) \sum_{\ulineb \in
\databasespace} \mathbb{U}(\ulineb |
c^{*}_{\histof(\ulinea)})\mathcal{F}(\histof(\ulinea),\histof(\ulineb))
\leq  \sum_{\ulineg \in \nhist} \sum_{\substack{\ulinea \in
\databasespace:\\\histof(\ulinea) = \ulineg}}p(\ulinea) \sum_{\ulineb \in
\databasespace} \mathbb{U}(\ulineb |
\ulinea)\mathcal{F}(\histof(\ulinea),\histof(\ulineb)) =
D^{n}(\mathbb{U})\nonumber
\end{eqnarray}
Suppose $\mathbb{U} : \databasespace \rightarrow \databasespace$ and $\mathbb{V} : \databasespace \rightarrow \databasespace$ are DSMs such that 
\begin{eqnarray}
\sum_{\substack{\ulinea \in  \databasespace:\\\histof(\ulinea)=\ulineh}}
\mathbb{U}(\ulinea | \ulineb) = \sum_{\substack{\ulinec \in
\databasespace:\\\histof(\ulinea)=\ulineh}} \mathbb{V}(\ulinec | \ulineb)~~~
\forall \ulineh \in \nhist, \forall \ulineb \in \databasespace, \mbox{ then}
\nonumber\end{eqnarray}\begin{eqnarray}
D^{n}(\mathbb{U})  &=& \sum_{\ulinea \in \databasespace}\sum_{\ulineb \in
\databasespace} p(\ulinea)\mathbb{U}(\ulineb |
\ulinea)\mathcal{F}(\histof(\ulineb),\histof(\ulinea)) = \sum_{\ulinea \in
\databasespace}          \sum_{\ulineh \in \nhist }\sum_{\substack{\ulineb \in
\databasespace:\\\histof(\ulineb)=\ulineh}} p(\ulinea)\mathbb{U}(\ulineb |
\ulinea)\mathcal{F}(\ulineh,\histof(\ulinea)) \nonumber\\
&=& \sum_{\ulinea \in \databasespace}          \sum_{\ulineh \in \nhist
}p(\ulinea)\mathcal{F}(\ulineh,\histof(\ulinea))\sum_{\substack{\ulineb \in
\databasespace:\\\histof(\ulineb)=\ulineh}} \mathbb{U}(\ulineb | \ulinea) =
\sum_{\ulinea \in \databasespace}          \sum_{\ulineh \in \nhist
}p(\ulinea)\mathcal{F}(\ulineh,\histof(\ulinea))\sum_{\substack{\ulineb \in
\databasespace:\\\histof(\ulineb)=\ulineh}} \mathbb{V}(\ulineb | \ulinea)
\nonumber\\&=&
D^{n}(\mathbb{V}).\nonumber
\end{eqnarray}
The above discussion narrows our search to histogram sanitizing mechanisms
$\mathbb{W} : \nhist \rightarrow \nhist$. The prior distribution on $\nhist$ is
given by (\ref{Eqn:ProbOfRandomDatabaseAndHistograms}). Our goal, is therefore 
to only identify a $\theta-$DP HSM that minimizes
\begin{eqnarray}
 \label{Eqn:ObjectiveFunctionOfHSM}
 D^{n}(\mathbb{W}) = \sum_{\ulineh \in \nhist}\sum_{\ulineg \in \nhist}
{n\choose \ulineh}\ulinep^{\ulineh}
\mathbb{W}(\ulineg | \ulineh)\histdist(\ulineg,\ulineh).\nonumber
\end{eqnarray}
\end{proof}
\section{Proof of Corollary \ref{Cor:ExtractingMainTerms}}
\label{AppSec:ProofOfCorollary}
We let $y=\frac{1+\theta}{1-\theta}$ and note $\frac{dy}{d\theta} = 
\frac{2}{(1 - \theta)^{2}}$. Observe that
\begin{eqnarray}
 \frac{S_{K-1}'(\theta)}{S_{K-1}(\theta)} &=& 
\frac{1}{(1-\theta)^{K-1}L_{K-1}(y)}\frac{d}{d\theta}\left\{ 
(1-\theta)^{K-1}L_{K-1 }(y) \right\} = \frac{-(K-1)}{(1-\theta)} + \frac{ 
(1-\theta)^{K-1}\frac{d L_{K-1}(y)
}{d\theta} }{(1-\theta)^{K-1}L_{K-1}(y)} 
\nonumber\\
 \label{Eqn:ProofOfCorollary1}
&=& \frac{-(K-1)}{(1-\theta)} + 
\frac{\frac{dL_{K-1}(y)}{dy}\frac{2}{(1-\theta)^{2}}}{L_{K-1}(y)} = 
\frac{-(K-1)}{(1-\theta)} + 
\frac{L_{K-1}'(y)\frac{2}{(1-\theta)^{2}}}{L_{K-1}(y)}.
\end{eqnarray}
We now utilize the recurrence relations $(1-y^{2})L_{n}'(y) = nL_{n-1}(y) 
- nyL_{n}(y)$ for every $n \geq 2$ and $(m+1)L_{m+1}(y)  - 
(2m+1)yL_{m}(y)+mL_{m-1}(y) = 0$ for every $m \geq 1$. Substituting $n=K-1$ and 
$m = K-1$ in these relations, we conclude $(1-y^{2})L_{K-1}'(y) = KyL_{K-1}(y) 
- KL_{K}(y)$, and hence $\frac{L_{K-1}'(y)}{L_{K-1}(y)} = 
\frac{Ky}{(1-y^{2})}-\frac{K}{(1-y^{2})}\frac{L_{K}(y)}{L_{K-1}(y)}$. 
Substituting this in (\ref{Eqn:ProofOfCorollary1}), one can verify
\begin{eqnarray}
 \label{Eqn:ProofOfCorollary2}
 2\theta \left\{ \frac{K-1}{1-\theta} + 
\frac{S_{K-1}'(\theta)}{S_{K-1}(\theta)}  \right\} = 2\theta \left\{  
\frac{-K}{\theta}\left(\frac{1+\theta}{1-\theta}\right) + \frac{K}{2\theta} 
\frac{L_{K}(y)}{L_{K-1}(y)} \right\} = K \left\{ \frac{L_{K}(y)}{L_{K-1}(y)} + 
\frac{1+\theta}{1- \theta}  \right\}, \nonumber
\end{eqnarray}
and this concludes the proof.
\section{Properties of Privacy-Constraint Graph and $\nhist$}
\label{AppSec:L1DistanceAndGraphDistance}
We list and prove some simple properties of the set of histograms $\nhist$ and
the PC graph involved in our study.
\begin{lemma}
 \label{Lem:PropertiesPCGraph}
Consider the set $\nhist_{K}$ of histograms defined in
(\ref{Eqn:HistogramModelling}) and the PC graph $G = (V,E)$, wherein $V =
\nhist_{K}$ and $E = \left\{ (\ulineh, \underline{\hat{h}}) \in \nhist \times
\nhist: |\ulineh-\underline{\hat{h}}|_{1} = 2\right\}$. The following are true
(i) For any $\ulineg, \ulineh \in \nhist$, $|\ulineg - \ulineh|_{1}$ is an even
integer and at most $2n$. (ii) $2d_{G}(\ulineg,\ulineh) = |\ulineg -
\ulineh|_{1}$.
\end{lemma}

\begin{proof}
(i) For any $\ulineg, \ulineh \in \nhist$, we have $\sum_{k=1}^{K}g_{k} =
\sum_{k=1}^{K}h_{k} = n$, and hence for any subset $S \subseteq [K]$, we have
$\sum_{i \in S} (g_{i}-h_{i}) = \sum_{j \in [K]\setminus S} (h_{j}-g_{j})$. Note
that \[|\ulineg - \ulineh|_{1} = \sum_{i=1}^{n}|g_{i}-h_{i}|=\sum_{i : g_{i}
\geq h_{i}}(g_{i}-h_{i})+\sum_{j : h_{j}> g_{j}}(h_{j}-g_{j}) = 2 \sum_{i :
g_{i} \geq h_{i}}(g_{i}-h_{i}),\] which is an even integer. Moreover $\sum_{i :
g_{i} \geq h_{i}}(g_{i}-h_{i}) \leq \sum_{i=1}^{K}g_{i} =n$, and hence $|\ulineg
- \ulineh|_{1} \leq 2n$.

(ii) We prove this by induction on $K$. When $K=1$, we have $\nhist_{1} =
\{(n)\}$, and the statement is true. When $K =2$, we note that
$|(n-i,i)-(n-j,j)|_{1} = 2|i-j|$ and the nodes $(n-i,i), (n-j,j)$ are indeed
$|i-j|$ hops apart (Fig. \ref{Fig:PCGraphFork=2}). Hence
$|i-j|=d_{G}((n-i,i),(n-j,j))$ and the statement is true. We assume the truth of
this statement for $K=1,\cdots,L-1$ and any $n$. Suppose $K=L$ and let $\ulineg,
\ulineh \in \nhist_{L}$. If for some coordinate $i$, we have $g_{i}=h_{i}$, 
then, let $\underline{\tilde{g}} \define (g_{j}:j \neq i)$ and $\ulinetildeh
\define (h_{j}: j \neq i)$. We have $\ulinetildeg,\ulinetildeh \in
\mathcal{H}^{n-g_{i}}_{L-1}$. By our induction hypothesis, we have
$2d_{\tilde{G}}(\ulinetildeg,\ulinetildeh) = |\ulinetildeg-\ulinetildeh|=
|\ulineg-\ulineh|$, where $\tilde{G}$ is the PC graph corresponding to
$\mathcal{H}^{n-g_{i}}_{L-1}$. It can now be verified that a shortest path from
$\ulineg$ to $\ulineh$ on $G$ corresponds to a shortest path between 
$\ulinetildeg$ to $\ulinetildeh$ in $\tilde{G}$ and hence
$d_{\tilde{G}}(\ulinetildeg,\ulinetildeh) = d_{G}(\ulineg,\ulineh)$. In fact,
observe that the graph induced on the set of vertices on a horizontal line in
Fig. \ref{Fig:PCGraphFork=3n=5} is isomorphic to the graph in Fig.
\ref{Fig:PCGraphFork=2} for an appropriate choice of $n$. Let us now consider
the alternate case where $\ulineg, \ulineh \in \nhist_{L}$ are such that for
\textit{no} co-ordinate $i$ do we have $g_{i}=h_{i}$. Without loss of
generality, assume $a=g_{1}-h_{1} > 0$. Let $i_{1}, \cdots, i_{R} \in [2,L]$ be
coordinates such that $h_{i_{r}} > g_{i_{r}}$ for $r \in [1,R]$ and
$\sum_{r=1}^{R}(h_{i_{r}} - g_{i_{r}}) \geq a$. The existence of coordinates
$i_{1},\cdots,i_{R}$ can be easily proved by using the fact that $\ulineg,
\ulineh  \in \nhist_{L}$. Now, let $b_{1}, \cdots b_{R} >0$ be integers such
that $h_{i_{r}} - g_{i_{r}} \geq b_{r} > 0$ for $r \in [R]$ and
$\sum_{r=1}^{R}b_{i_{r}} = a$. Now consider $\ulinef \in \nhist_{L}$ such that 
$f_{1} = g_{1}-a$, $f_{i_{r}} = g_{i_{r}}+b_{r}$ and $f_{j}=g_{j}$ if $j \notin 
\{1,i_{1},\cdots,i_{R}\}$. It can now be verified, by using the induction 
hypothesis on $\ulinef,\ulineh$, that 
$d_{G}(\ulineg,\ulineh)=d_{G}(\ulineg,\ulinef)+d_{G}(\ulinef,\ulineh)$, 
$2d_{G}(\ulineg,\ulinef)=|\ulineg-\ulinef|_{1}$,
$2d_{G}(\ulinef,\ulineh)=|\ulinef-\ulineh|_{1}$, and 
$|\ulineg-\ulinef|_{1}+|\ulinef-\ulineh|_{1}=|\ulineg-\ulineh|_{1}$, thereby 
proving the statement for $K=L$.
\end{proof}

\section{The weak duality theorem of LP}
\label{AppSec:WDT}
We refer the reader to \cite{Ber1997_IntroToLP}
for a description of the dual linear program. Following the same notation,
we state WDT below.

Weak Duality Theorem : Consider the following primal and dual LP problems. Let $\bold{A}$ be a matrix with rows $\bold{a}'_{i}$ and columns $\bold{A}_{j}$.
\begin{eqnarray}
\begin{array}{llccllc}
&\mbox{Primal LP}&&&&\mbox{Dual LP}&\\
\mbox{Minimize } &\bold{c}'\bold{x}&&&\mbox{Maximize } &\bold{p}'\bold{b}&\\
\mbox{subject to } &\bold{a}_{i}'\bold{x}\geq \bold{b}_{i}&i \in M_{1}&&\mbox{subject to } &p_{i}\geq 0&i \in M_{1}\\
 &\bold{a}_{i}'\bold{x} = \bold{b}_{i}&i \in M_{3}&&&p_{i}\mbox{ free}&i \in M_{3}\\
 &x_{j}\geq 0&j \in N_{1},&&&\bold{p}'\bold{A}_{j}\leq c_{j}&j \in N_{1}.
\end{array}\nonumber
\end{eqnarray}
If $\bold{x}$ and $\bold{p}$ are feasible solutions to the primal and dual problems respectively, then $\bold{p}'\bold{b} \leq \bold{c}'\bold{x}$.
\section{Mechanism $\mathbb{U} : \nhist \Rightarrow \nhistext$ is a $\theta-$DP 
mechanism}
\label{AppSec:UIsThetaDP}
Recall, $\mathbb{U}: \nhist \Rightarrow \nhistext$ is specified in (\ref{Eqn:DefnOfUn}), and we let
\begin{eqnarray}
 \label{Defn:LambdaOfTheta}
 \EhrFaceSer(\theta) = (1-\theta)\EhrSer_{\mathcal{P}}(\theta) = 
1+\displaystyle\sum_{d=1}^{\infty}N_{d}\theta^{d}.
\end{eqnarray}
Clearly, $\mathbb{U}^{n}(\ulineg|\ulineh) \geq 0$. We note that
\begin{eqnarray}
 \label{Eqn:UnMeetsSunConstraint}
 \sum_{\ulineg \in \nhistext}\mathbb{U}^{n}(\ulineg | \ulineh) \!\!&=&\!\!
\frac{1}{\EhrFaceSer(\theta)}\sum_{\ulineg \in
\nhistext}\!\!\!\theta^{\frac{|\ulineg-\ulineh|_{1}}{2}}=
\frac{1}{\EhrFaceSer(\theta)}\sum_{d=0}^{\infty}\sum_{\substack{\ulineg \in
\nhistext:\\|\ulineg-\ulineh|_{1}=2d}}\!\!\!\theta^{\frac{|\ulineg-\ulineh|_{1}}
{2 }} = \frac{1}{\EhrFaceSer(\theta)}\sum_{d=0}^{\infty}\sum_{\substack{\ulineg
\in \nhistext:\\|\ulineg-\ulineh|_{1}=2d}}\theta^{d}\nonumber\\
 &=&\frac{1}{\EhrFaceSer(\theta)}\sum_{d=0}^{\infty}N_{d}\theta^{d} =
\frac{1}{\EhrFaceSer(\theta)}\left( 1+\sum_{d=1}^{\infty}N_{d}\theta^{d}
\right)=1.\nonumber
\end{eqnarray}
Lastly, suppose $\ulineh \in \nhist$ and $\underline{\tilde{h} }\in \nhist$ are a pair of neighboring histograms,
\begin{eqnarray}
 \label{Eqn:UnMeetsThetaDPConstraints}
\mathbb{U}^{n}(\ulineg|\ulineh) / \mathbb{U}^{n}(\ulineg|\underline{\tilde{h}}) = \theta^{\frac{|\ulineg - \ulineh|_{1}}{2}}/\theta^{\frac{|\ulineg-\underline{\tilde{h}}|_{1}}{2}} = \theta^{\frac{\left( |\ulineg - \ulineh|_{1}-|\ulineg-\underline{\tilde{h}}|_{1}\right)}{2}}.\nonumber
\end{eqnarray}
By the triangle inequality, $-2=-|\ulineh - 
\underline{\tilde{h}}|_{1}\leq|\ulineg-\underline{\tilde{h}}|_{1}-|\ulineg - 
\ulineh|_{1}\leq |\ulineh - \underline{\tilde{h}}|_{1}=2$, and we wee that 
the above ratio is in $[\theta,\frac{1}{\theta}]$. $\mathbb{U}^{n}$ is 
therefore a $\theta-$DP mechanism.

\section{For $n$ sufficiently large, $D_{\hist}^{n}(\mathbb{W}^{n}) \leq 
D(\mathbb{U}^{n})$}
\label{AppSec:DUnIs=DWn}
Here we prove that the expected distortion of $\mathbb{W}^{n}$ is, in the 
limit, at most that of $\mathbb{U}^{n}$, i.e., $\lim_{n \rightarrow 
\infty}D(\mathbb{W}^{n},\ulinep) \leq  \lim_{n \rightarrow 
\infty}D(\mathbb{U}^{n})$. Towards this end, we let $B(\delta,\ulineh) \define 
\left\{ \ulineg \in \nhist : |\ulineg-\ulineh|_{1} \leq \delta\right\}$ and 
${B^{c}(\delta,\ulineh)} \define \nhist\setminus B(\delta,\ulineh)$ its 
complement. We abbreviate 
$B(\frac{1}{2})=B(\frac{R}{2}n^{\frac{2}{3}},n\ulinep)$, $B^{c}(\frac{1}{2}) = 
B^{c}(\frac{R}{2}n^{\frac{2}{3}},n\ulinep)$, 
$B(1)=B({R}n^{\frac{2}{3}},n\ulinep)$, $B^{c}(1) = 
B^{c}({R}n^{\frac{2}{3}},n\ulinep)$. Observe that
\begin{eqnarray}
 \label{Eqn:ExpectedFidelityOfCascadeChannels}
 D(\mathbb{W}^{n},\ulinep) &=& \sum_{\ulineh \in \nhist}\sum_{\ulineg \in 
\nhist}{n 
\choose 
\ulineh}\ulinep^{\ulineh}\mathbb{W}^{n}(\ulineg|\ulineh)|\ulineg-\ulineh|_{1} 
\nonumber\\&=& \sum_{\ulineh \in B(\frac{1}{2})}\sum_{\ulineg \in \nhist}{n 
\choose 
\ulineh}\ulinep^{\ulineh}\mathbb{W}^{n}(\ulineg|\ulineh)|\ulineg-\ulineh|_{1} 
 + \sum_{\ulineh \in B^{c}(\frac{1}{2})}\sum_{\ulineg \in \nhist}{n \choose 
\ulineh}\ulinep^{\ulineh}\mathbb{W}^{n}(\ulineg|\ulineh)|\ulineg-\ulineh|_{1} 
\nonumber\\
 &\leq&  \sum_{\ulineh \in B(\frac{1}{2})}\sum_{\ulineg \in \nhist}{n \choose 
\ulineh}\ulinep^{\ulineh}\mathbb{W}^{n}(\ulineg|\ulineh)|\ulineg-\ulineh|_{1}  
+ 
\sum_{\ulineh \in B^{c}(\frac{1}{2})}\sum_{\ulineg \in \nhist}{n \choose 
\ulineh}\ulinep^{\ulineh}\mathbb{W}^{n}(\ulineg|\ulineh)2n \nonumber\\
 &\leq&  \sum_{\ulineh \in B(\frac{1}{2})}\sum_{\ulineg \in \nhist}{n \choose
\ulineh}\ulinep^{\ulineh}\mathbb{W}^{n}(\ulineg|\ulineh)|\ulineg-\ulineh|_{1}  +
2n\sum_{\ulineh \in B^{c}(\frac{1}{2})}{n \choose \ulineh}\ulinep^{\ulineh}.
\nonumber
\end{eqnarray}
It can be easily shown that $\sum_{\ulineh \in B^{c}(\frac{1}{2})}{n \choose
\ulineh}\ulinep^{\ulineh} \leq \exp\left\{ -n\alpha\right\}$, and hence the
second term above can be made arbitrarily small by choosing $n$ large enough. We
henceforth focus on the first term above which is given by
\begin{eqnarray}
 \label{Eqn:FirstTermOfTheDistortion}
  \lefteqn{\sum_{\ulineh \in B(\frac{1}{2})}\sum_{\ulineg \in B(1)}{n \choose 
\ulineh}\ulinep^{\ulineh}\mathbb{W}^{n}(\ulineg|\ulineh)|\ulineg-\ulineh|_{1}
+\sum_{\ulineh \in B(\frac{1}{2})}\sum_{\ulineg \in B^{c}(1)}{n \choose 
\ulineh}\ulinep^{\ulineh}\mathbb{W}^{n}(\ulineg|\ulineh)|\ulineg-\ulineh|_{1} 
}\nonumber\\
  \label{Eqn:SecondTermIsZero}
  &=& \sum_{\ulineh \in B(\frac{1}{2})}{n \choose 
\ulineh}\ulinep^{\ulineh}\left( 
|n\ulinep-\ulineh|_{1}\mathbb{W}^{n}(n\ulinep|\ulineh)+ \sum_{\ulineg \in 
B(1)\setminus\{ 
n\ulinep\}}\mathbb{W}^{n}(\ulineg|\ulineh)|\ulineg-\ulineh|_{1}\right)  \\
  \label{Eqn:ThirdEquality}
  &=& \sum_{\ulineh \in B(\frac{1}{2})}{n \choose 
\ulineh}\ulinep^{\ulineh}\left( |n\ulinep-\ulineh|_{1} 
[\mathbb{U}^{n}(n\ulinep|\ulineh)+\sum_{\underline{\tilde{g}}\in B^{c}(1) 
}\mathbb{U}^{n}(\underline{\tilde{g}}|\ulineh)]+ \sum_{\ulineg \in 
B(1)\setminus\{ 
n\ulinep\}}\mathbb{U}^{n}(\ulineg|\ulineh)|\ulineg-\ulineh|_{1}\right) \\
  \label{Eqn:FourthEquality}
  &\leq& \sum_{\ulineh \in B(\frac{1}{2})}{n \choose
\ulineh}\ulinep^{\ulineh}\left( |n\ulinep-\ulineh|_{1}
\mathbb{U}^{n}(n\ulinep|\ulineh)+\sum_{\underline{\tilde{g}}\in B^{c}(1)
}\!\!\!|\underline{\tilde{g}}-\ulineh|_{1}\mathbb{U}^{n}(\underline{\tilde{g}}
|\ulineh)+ \sum_{\ulineg \in B(1)\setminus\{
n\ulinep\}}\!\!\!\!\!\!\mathbb{U}^{n}(\ulineg|\ulineh)|\ulineg-\ulineh|_{1}
\right) \\
  &\leq& \sum_{\ulineh \in B(\frac{1}{2})}{n \choose 
\ulineh}\ulinep^{\ulineh}\left( |n\ulinep-\ulineh|_{1} 
\mathbb{U}^{n}(n\ulinep|\ulineh)+\sum_{\underline{\tilde{g}}\in B^{c}(1) 
}|\underline{\tilde{g}}-\ulineh|_{1}\mathbb{U}^{n}(\underline{\tilde{g}}
|\ulineh)+ \sum_{\ulineg \in B(1)\setminus\{ 
n\ulinep\}}\mathbb{U}^{n}(\ulineg|\ulineh)|\ulineg-\ulineh|_{1}\right) 
\nonumber\\
  &=& \sum_{\ulineh \in B(\frac{1}{2})}{n \choose \ulineh}\ulinep^{\ulineh} 
\sum_{\ulineg \in \nhist}\mathbb{U}^{n}(\ulineg|\ulineh)|\ulineg-\ulineh|_{1}
\leq D(\mathbb{U}^{n}),\nonumber
\end{eqnarray}
where (i) (\ref{Eqn:SecondTermIsZero}) follows from
$\mathbb{W}^{n}(\underline{\tilde{g}}|\ulineh) = 0 $ for $\underline{\tilde{g}}
\in B^{c}(1)$ implying\footnote{Note that the range of $f_{\mathbb{V}^{n}}$ is
$B(1)$.} that the second term is zero, (ii) (\ref{Eqn:ThirdEquality}) follows
from the definition of $\mathbb{W}^{n}$ in terms of $\mathbb{U}^{n}$, (iii)
(\ref{Eqn:FourthEquality}) is true since, for every $\ulineh \in B(\frac{1}{2})$
and every $\underline{\tilde{g}} \in B^{c}(1)$, $|n\ulinep-\ulineh|_{1}\leq
\frac{R}{2}n^{\frac{2}{3}} \leq Rn^{\frac{2}{3}} \leq
|\underline{\tilde{g}}-\ulineh|_{1}$.
\section{Characterization of $A_{n},B_{n}$ defined in 
(\ref{Eqn:AssigmentOfDualVariableAnBn})}
\label{AppSec:AnBnInDualVarAssignment}
$A_{n}$ on the left and $B_{n}$ on the right constitute the boundaries of the 
support of the truncated geometric mechanism. It is instructive to study 
$A_{n},B_{n}$ for different distributions $\binpmf$. Suppose one 
replaces $\mathscr{C}_{i}^{n}$ by $\frac{1}{n+1}$ - the uniform pmf on 
the set of histograms $\nhist_{2}$, then simple calculation shows that $A_{n} 
\leq \mathcal{N}_{\theta} \define \min \{ i \in 
\naturals : \theta^{i} < 1-\theta \}$ and $B_{n} \geq n -\mathcal{N}_{\theta}$. 
Since this will provide us with important intuition, we first proceed with 
these steps. We recall the definitions for ease of reference:
\begin{eqnarray}
 \label{Eqn:AppFiBiDefns}
 f_{i} \define 2\sum_{j=0}^{i}\mathscr{C}_{j}^{n}\theta^{i-j}, ~~ b_{i} \define
2\sum_{k=i}^{n}\mathscr{C}_{k}^{n}\theta^{k-i}, \nonumber
\end{eqnarray}
\begin{eqnarray}
 \label{Eqn:AppAssigmentOfDualVariableAnBn}
A_{n}\define \min\left\{ 
i \in [0,n]:\begin{array}{c}f_{k-1}-\theta b_{k}\geq 0\\\mbox{for every }k 
\geq i \end{array}\right\}, 
 B_{n}\define \max\left\{ 
i \in [0,n]:\begin{array}{c}b_{k+1}-\theta f_{k}\geq 0\\\mbox{ for every }k 
\leq i \end{array}\right\}.
\end{eqnarray}
Since we are interested in $f_{i-1}-\theta b_{i}$ and $b_{i+1}-\theta f_{i}$, 
we will ignore the multiplier $2$ in the definitions of $f_{i}$ and $b_{i}$. We 
work out a simple case to understand the core problem. Let us begin with the 
case $\mathcal{C}_{i}^{n}=\frac{1}{n+1}$ for $i \in [0,n]$. It can be verified 
that
\begin{eqnarray}
 f_{i-1}-\theta b_{i} &=& \frac{1}{n+1}\left[ \theta^{i-1}+\theta^{i-2}+\cdots 
+ \theta + 1 - \theta \left( 1 + \theta + \theta^{2}+\cdots + \theta^{n-i} 
\right)  \right] \nonumber\\
&=& \frac{1}{n+1}\left[ \frac{1-\theta^{i}}{1-\theta} - \theta \left(  
\frac{1-\theta^{n-i+1}}{1-\theta} \right) \right] = \frac{1}{n+1}\left[   
 1- \frac{\theta^{i}-\theta^{n-i+2}}{1-\theta} \right]\nonumber\\
 & \geq & \frac{1}{n+1}\left[ 1 - \frac{\theta^{i}}{1-\theta} \right].\nonumber
\end{eqnarray}
Clearly, $A_{n} < \min \{ i : \theta^{i} < 1-\theta \}$. A similar sequence of 
steps leads one to conclude that $B_{n} > \max \{i : \theta^{n-i} < 1-\theta 
\}$. We observe $A_{n}=\mathcal{O}(1)$ and $n-B_{n}=\mathcal{O}(1)$. Our 
characterization for $A_{n}$ and $B_{n}$ for $\mathscr{C}_{i}^{n} = {n \choose 
i}p^{i}(1-p)^{n-i}$ is based on the above intuition. The key property of 
the binomial pmf, that it is near-uniform in the window 
$[np-\mathcal{O}(\sqrt{n}) , np + \mathcal{O}(\sqrt{n})]$ is employed. 
Specifically, note that for sufficiently large $n$
\begin{eqnarray}
\label{Eqn:AppSecAnBn1}
 \max \left\{ \frac{\mathscr{C}^{n}_{np}}{\mathscr{C}^{n}_{np-x} }, 
\frac{\mathscr{C}^{n}_{np}}{\mathscr{C}^{n}_{np+x} } \right\} \leq 2\exp\{ 
\frac{x^{2}}{2np(1-p)} \},
\end{eqnarray}
where (\ref{Eqn:AppSecAnBn1}) follows from \cite[Eqn. 
106]{201107TIT_SzpVer}.\footnote{Note that $\mathscr{C}_{i}^{n} = {n \choose 
i}2^{-nH(X)}\left( \frac{p}{1-p} \right)^{i-np}$.} For $x \sim 
\sqrt{\frac{n}{(\log n)^{4}}}$, the above ratio shrinks as $\frac{1}{n^{4}}$. 
Note that ${n \choose np}$ scales as $\frac{1}{\sqrt{n}}$. We can use this to 
bound the ratio between the largest and the smallest binomial probability 
masses in the range $[np-\sqrt{\frac{n}{(\log n)^{4}}},np+\sqrt{\frac{n}{(\log 
n)^{4}}}]$, and we can use the same sequence of steps used above. It can be 
proved that $np-A_{n} = \mathcal{O}(\sqrt{\frac{n}{(\log n)^{4}}})$ and 
$B_{n}-np = 
\mathcal{O}(\sqrt{\frac{n}{(\log n)^{4}}})$. The reader may refer to 
\cite{201803arXiv_PadKumSzp} for a detailed proof of these claims.
\section{Interpretation of dual variable assignments via shadow prices}
\label{AppSec:ShadowPriceForDualVars}
We provide an interpretation for the assignments of the dual variables in Eq.
(\ref{Eqn:Assignment1})-(\ref{Eqn:MuiInsideCriticalWindow}) via shadow prices.
Assignment (\ref{Eqn:Assignment2}) for $j=i$ can be interpreted via mechanism
$\hat{\mathbb{W}}(\cdot|\cdot)$ defined as
$\hat{\mathbb{W}}(k|j)=\mathbb{W}(k|j)+d\mathbb{W}(k|j)$, where
$\mathbb{W}(\cdot|\cdot)$ is the truncated geometric mechanism defined in
(\ref{Eqn:TruncatedGeometric}) and
\begin{eqnarray}
 \label{Eqn:DiffMech1}
 d\mathbb{W}(k|j) = \left\{ 
 \begin{array}{ll}
0&\mbox{if }k\neq (i-1)\mbox{, and }k\neq i,\\
-\epsilon \theta^{|j-(i-1)|}&\mbox{if }k = (i-1),\\
+\epsilon \theta^{|j-(i-1)|}&\mbox{if }k = i.
 \end{array}
 \right.
\end{eqnarray}
It is straightforward to verify that $\hat{\mathbb{W}}$ satisfies all the
constraints of a $\theta-$DP mechanism (just as $\mathbb{W}$), and more
importantly, $\hat{\mathbb{W}}(i|i-1)-\theta\hat{\mathbb{W}}(i|i) =
\epsilon(1-\theta^{2})$. In fact, except for this constraint, $\mathbb{W}$ and
$\hat{\mathbb{W}}$ are identical wrt all other constraints.
$\mathbb{W}$ and $\hat{\mathbb{W}}$ are identical vertices in their
corresponding feasible regions, with the only difference being that 
$\hat{\mathbb{W}}$ satisfies the constraint
$\hat{\mathbb{W}}(i|i-1)-\theta\hat{\mathbb{W}}(i|i) \geq
\epsilon(1-\theta^{2})$. Moreover, it can be verified that
$D^{n}_{\hist}(\hat{\mathbb{W}})-D^{n}_{\hist}({\mathbb{W}}) =
\epsilon(f_{i-1}-\theta b_{i})$. Recognize that \begin{eqnarray}\lim_{\epsilon
\rightarrow 0}
\frac{D^{n}_{\hist}(\hat{\mathbb{W}})-D^{n}_{\hist}({\mathbb{W}})}{\hat{\mathbb{
W}}(i|i-1)-\theta\hat{\mathbb{W}}(i|i)} =\lim_{\epsilon
\rightarrow 0}
\frac{D^{n}(d{\mathbb{W}})}{\hat{\mathbb{
W}}(i|i-1)-\theta\hat{\mathbb{W}}(i|i)} = \lim_{\epsilon
\rightarrow 0}
\frac{\epsilon(f_{i-1}-\theta b_{i})}{\epsilon(1-\theta^{2})} =
\lambda_{i|(i-1,i)}.\nonumber\end{eqnarray}
These are indeed the shadow prices that we alluded to. We continue and 
discuss the interpretation for
the rest of the variables. Consider assignment (\ref{Eqn:Assignment2}) for $j >
i$. Consider $\hat{\mathbb{W}}(\cdot|\cdot)$ defined as
$\hat{\mathbb{W}}(a|b)=\mathbb{W}(a|b)+d\mathbb{W}(a|b)$, where
$\mathbb{W}(\cdot|\cdot)$ is the truncated geometric mechanism defined in
(\ref{Eqn:TruncatedGeometric}), and $d\mathbb{W}$ is now defined as
\begin{eqnarray}
 \label{Eqn:DiffMech2}
 d\mathbb{W}(a|b) = \left\{ 
 \begin{array}{ll}
0&\mbox{if }a\neq (i-1)\mbox{, and }a\neq i,\mbox{ and }a\neq j\\
-\epsilon \theta^{|b-(i-1)|}&\mbox{if }a = (i-1),\\
+\epsilon \theta^{|b-(i-1)|}&\mbox{if }a = i, b \geq i\\
+\epsilon \theta^{|b-(i-1)|+2}&\mbox{if }a = i, b \leq i-1\\
+\epsilon (\theta^{|b-(i-1)|}-\theta^{|b-(i-1)|+2})&\mbox{if }a = j, b \leq
i-1\\
0&\mbox{if }a = j, b \geq i.
 \end{array}
 \right.
\end{eqnarray}
As earlier, it is straightforward to verify that $\hat{\mathbb{W}}$
satisfies all the
constraints of a $\theta-$DP mechanism (just as $\mathbb{W}$), and more
importantly, $\hat{\mathbb{W}}(j|i-1)-\theta\hat{\mathbb{W}}(j|i) =
\epsilon(1-\theta^{2})$. In fact, except for this constraint, $\mathbb{W}$ and
$\hat{\mathbb{W}}$ are identical wrt all other constraints. Moreover, it can be
verified that $D^{n}_{\hist}(\hat{\mathbb{W}})-D^{n}_{\hist}({\mathbb{W}}) =
\epsilon(f_{i-1}-\theta b_{i})$. Recognize that \begin{eqnarray}\lim_{\epsilon
\rightarrow 0}
\frac{D^{n}_{\hist}(\hat{\mathbb{W}})-D^{n}_{\hist}({\mathbb{W}})}{\hat{\mathbb{
W}}(j|i-1)-\theta\hat{\mathbb{W}}(j|i)} &=&\lim_{\epsilon
\rightarrow 0}
\frac{D^{n}(d{\mathbb{W}})}{\hat{\mathbb{
W}}(j|i-1)-\theta\hat{\mathbb{W}}(j|i)} \nonumber\\&=& \lim_{\epsilon
\rightarrow 0}
\frac{\epsilon(\theta^{2}f_{i-1}+(j-i+1)(1-\theta^{2})f_{i-1}-\theta
b_{i})}{\epsilon(1-\theta^{2})} =
\lambda_{j|(i-1,i)}.\nonumber\end{eqnarray}
Now consider (\ref{Eqn:Assignment3}) with $j=i$. Analogous to
(\ref{Eqn:DiffMech1}), consider
\begin{eqnarray}
 \label{Eqn:DiffMech3}
 d\mathbb{W}(k|j) = \left\{ 
 \begin{array}{ll}
0&\mbox{if }k\neq (i+1)\mbox{, and }k\neq i,\\
-\epsilon \theta^{|j-(i+1)|}&\mbox{if }k = (i+1),\\
+\epsilon \theta^{|j-(i+1)|}&\mbox{if }k = i.
 \end{array}
 \right.
\end{eqnarray}
Following the same arguments as above, it can be verified by straightforward
substitutions that
\begin{eqnarray}\lim_{\epsilon \rightarrow 0}
\frac{D^{n}_{\hist}(\hat{\mathbb{W}})-D^{n}_{\hist}({\mathbb{W}})}{\hat{\mathbb{
W}}(i|i+1)-\theta\hat{\mathbb{W}}(i|i)} =\lim_{\epsilon
\rightarrow 0}
\frac{D^{n}(d{\mathbb{W}})}{\hat{\mathbb{
W}}(i|i+1)-\theta\hat{\mathbb{W}}(i|i)} = \lim_{\epsilon
\rightarrow 0}
\frac{\epsilon(b_{i+1}-\theta f_{i})}{\epsilon(1-\theta^{2})} =
\lambda_{i|(i+1,i)},\nonumber\end{eqnarray}
where, as before, $\hat{\mathbb{W}}(\cdot|\cdot)$ defined as
$\hat{\mathbb{W}}(k|j)=\mathbb{W}(k|j)+d\mathbb{W}(k|j)$, and 
$\mathbb{W}(\cdot|\cdot)$ is the truncated geometric mechanism. Similarly, for
$j < i$ we can verify the assignment in (\ref{Eqn:Assignment3})
through the following. Define mechanism
$\hat{\mathbb{W}}(\cdot|\cdot)=\mathbb{W}(k|j)+d\mathbb{W}(k|j)$, where
$\mathbb{W}(\cdot|\cdot)$ is the truncated geometric mechanism defined in
(\ref{Eqn:TruncatedGeometric}), and
\begin{eqnarray}
 \label{Eqn:DiffMech4}
 d\mathbb{W}(a|b) = \left\{ 
 \begin{array}{ll}
0&\mbox{if }a\neq (i+1)\mbox{, and }a\neq i,\mbox{ and }a\neq j\\
-\epsilon \theta^{|b-(i+1)|}&\mbox{if }a = (i+1),\\
+\epsilon \theta^{|b-(i+1)|}&\mbox{if }a = i, b \geq i\\
+\epsilon \theta^{|b-(i+1)|+2}&\mbox{if }a = i, b \geq i+1\\
+\epsilon (\theta^{|b-(i+1)|}-\theta^{|b-(i+1)|+2})&\mbox{if }a = j, b \geq
i+1\\
0&\mbox{if }a = j, b \leq i.
 \end{array}
 \right.
\end{eqnarray}
Following the same arguments as above, it can be verified by straightforward
substitutions that
\begin{eqnarray}\lim_{\epsilon
\rightarrow 0}
\frac{D^{n}_{\hist}(\hat{\mathbb{W}})-D^{n}_{\hist}({\mathbb{W}})}{\hat{\mathbb{
W}}(j|i+1)-\theta\hat{\mathbb{W}}(j|i)} &=&\lim_{\epsilon
\rightarrow 0}
\frac{D^{n}(d{\mathbb{W}})}{\hat{\mathbb{
W}}(j|i+1)-\theta\hat{\mathbb{W}}(j|i)} \nonumber\\&=& \lim_{\epsilon
\rightarrow 0}
\frac{\epsilon(\theta^{2}b_{i+1}+(i+1-j)(1-\theta^{2})b_{i+1}-\theta
f_{i})}{\epsilon(1-\theta^{2})} =
\lambda_{j|(i+1,i)},\nonumber\end{eqnarray}
where, as before, $\hat{\mathbb{W}}(\cdot|\cdot)$ defined as
$\hat{\mathbb{W}}(k|j)=\mathbb{W}(k|j)+d\mathbb{W}(k|j)$,
and $\mathbb{W}(\cdot|\cdot)$ is the truncated geometric mechanism. Finally, we
explain the assignment for $\mu_{i}$ in the range $[A-1,B+1]$. Consider
$\hat{\mathbb{W}}(b|a) = \mathbb{W}(b|a)+d\mathbb{W}(b|a)$ where $\mathbb{W}$
is the truncated Geometric mechanism as before, and
\begin{eqnarray}
 \label{Eqn:DiffMech5}
 d\mathbb{W}(a|b) = \left\{ 
 \begin{array}{ll}
0&\mbox{if }a\neq (i-1)\mbox{, and }a\neq i,\mbox{ and }a\neq (i+1)\\
-\epsilon \theta^{|b-(i-1)|+1}&\mbox{if }a = (i-1),\\
-\epsilon \theta^{|b-(i+1)|+1}&\mbox{if }a = i+1,\\
+\epsilon \theta^{|b-(i-1)|+1}+\epsilon \theta^{|b-(i+1)|+1}&\mbox{if }a = i, b
\neq i\\
+\epsilon (1+\theta^{2})&\mbox{if }a = i, b = i
 \end{array}
 \right.
\end{eqnarray}
The following can be verified easily : $\sum_{j=0}^{n}\hat{\mathbb{W}}(j|i) =
1+\epsilon-\epsilon\theta^{2}$. $\hat{\mathbb{W}}$ and $\mathbb{W}$ are
identical with respect to the set of DP constraints they satisfy, and
\begin{eqnarray}\lim_{\epsilon
\rightarrow 0}
\frac{D^{n}_{\hist}(\hat{\mathbb{W}})-D^{n}_{\hist}({\mathbb{W}})}{\sum_{j=0}^{n
}\hat{\mathbb{W}}(j|i)-1} &=&\lim_{\epsilon
\rightarrow 0}
\frac{D^{n}(d{\mathbb{W}})}{\sum_{j=0}^{n
}\hat{\mathbb{W}}(j|i)-1} \nonumber\\&=& \lim_{\epsilon
\rightarrow 0}
\frac{\epsilon[\theta(1-\theta^{2})(f_{i-1}+b_{i+1})-4\theta^{2}{n
\choose i}p^{i}(1-p)^{n-i}]}{\epsilon(1-\theta^{2} ) } =
\mu_{i}.\nonumber\end{eqnarray}
The key import of the above interpretation is the relationship between the
assignments (\ref{Eqn:DiffMech1})-(\ref{Eqn:DiffMech5}). (\ref{Eqn:DiffMech2})
can be obtained from (\ref{Eqn:DiffMech1}) by just shifting mass from $i$ to 
$j$.
Similarly, (\ref{Eqn:DiffMech4})
can be obtained from (\ref{Eqn:DiffMech3}) by just shifting mass from $i$ to 
$j$.
This provides an alternate proof of feasibility of this dual variable
assignment. Also note that the assignment (\ref{Eqn:DiffMech5}) is obtained as
$\theta$ times the assignment (\ref{Eqn:DiffMech1}) summed to $\theta$ times the
assignment (\ref{Eqn:DiffMech3}). The feasibility of this assignment is now an
immediate consequence of these relationships. This shadow price interpretation
is the basis for (\ref{Eqn:DualVarAssignment}), whose feasibility follows
immediately from the geometry of the constraints.

\bibliographystyle{IEEEtran}
{
\bibliography{InfoTheoreticDatabaseSanitizing}
\end{document}